\documentclass[11pt]{article}
\usepackage[T1]{fontenc}
\usepackage{amsmath,amsxtra,amssymb,latexsym,amscd,amsthm,pb-diagram,fancyhdr,euscript}
\usepackage{amsthm}
\usepackage{indentfirst}
\usepackage{epsfig}
\usepackage{mathrsfs}
\usepackage{slashed}
\usepackage{cases}
\usepackage{hyperref}
\usepackage{float}
\hypersetup{
   colorlinks,
    citecolor=blue,
    filecolor=black,
    linkcolor=blue,
    urlcolor=magenta
}
\hypersetup{linktocpage}

\renewcommand{\d}{\mathrm{d}}

\newcommand{\scri}{{\mathscr I}}

\newcommand{\hook}{{\setlength{\unitlength}{11pt}   
                   \begin{picture}(.833,.8)
                   \put(.15,.08){\line(1,0){.35}}
                   \put(.5,.08){\line(0,1){.5}}
                   \end{picture}}}
\newtheorem{definition}{Definition}
\newtheorem{theorem}{Theorem}
\newtheorem{proposition}{Proposition}
\newtheorem{corollary}{Corollary}
\newtheorem{lemma}{Lemma}
\newtheorem{remark}{Remark}

\begin{document}
\mbox{} \thispagestyle{empty}

\begin{center}
\bf{\Huge Conformal scattering theories for tensorial wave equations on Schwarzschild spacetime} \\

\vspace{0.1in}

{Truong Xuan PHAM\footnote{Faculty of Pedagogy, VNU University of Education, Vietnam National University, Hanoi, 144 Xuan Thuy, Cau Giay, Hanoi, Viet Nam. Email~: phamtruongxuan.k5@gmail.com}}
\end{center}

{\bf Abstract.} In this paper, we establish the constructions of conformal scattering theories for the tensorial wave equation such as the tensorial Fackerell-Ipser and the spin $\pm 1$ Teukolsky equations on Schwarzschild spacetime. In our strategy, we construct the conformal scattering for the tensorial Fackerell-Ipser equations which are obtained from the Maxwell equation and spin $\pm 1$ Teukolsky equations. Our method combines Penrose's conformal compactification and the energy decay results of the tensorial fields satisfying the tensorial Fackerell-Ipser equation to prove the energy equality of the fields through the conformal boundary $\mathfrak{H}^+\cup \scri^+$ (resp. $\mathfrak{H}^-\cup \scri^-$) and the initial Cauchy hypersurface $\Sigma_0 = \left\{ t=0 \right\}$. We will prove the well-posedness of the Goursat problem by using a generalization of H\"ormander's results for the tensorial wave equations. By using the results for the tensorial Fackerell-Ipser equations we will establish the construction of conformal scattering for the spin $\pm 1$ Teukolsky equations.  



\tableofcontents

\section{Introduction}
The analytic scattering theories of field equations outside black holes of spacetimes in general relativity have been studied since 1985. The first work of Dimock \cite{Di1} established the scattering theory for scalar wave equation on the Schwarzschild spacetime by using Cook's method. Then, the series works of Dimock and Kay provided the scattering theory for massive Klein-Gordon equations \cite{Di2} and classical and quantum scattering theory for linear scalar fields on the Schwarzschild spacetime \cite{Di3,Di4}. The works of Dimock and Kay have been developed by Bachelot to study the scattering theory for the Maxwell equation on the Schwarzschild spacetime \cite{Ba1}. In this work, Bachelot has also provided the connection between the Characteristic Cauchy problem (i.e., the Goursat problem) in the Penrose conformal spacetime and the existence of wave operators. After that, Bachelot studied the asymptotic completeness and scattering theory for massive Klein-Gordon equations on the Schwarzschild spacetime in \cite{Ba2} by using the invariance principle for long range potentials, and constructed the scattering operator by Dollar-modified wave operators. Concerning the scattering of Dirac fields outside a Schwarzschild black hole, Nicolas \cite{Ni1995'} provided a scattering theory for classical massless Dirac fields by using Cook's method; Jin \cite{Jin} constructed wave operators, classical at the event horizon and Dollard-modified at infinity and obtained the scattering for the massive Dirac fields. Moreover, Melnik \cite{Mel} gave a complete scattering theory for massive charged Dirac fields on the Reissner-Nordstrøm spacetime.

A complete scattering theory for the wave equations, on stationary, asymptotically
flat spacetimes (which consists of Kerr spacetimes) has been established by H\"afner \cite{Ha2001}
by using Mourre's theory. Then, the work \cite{Ha2001} has been extended by H\"afner and Nicolas \cite{HaNi2004} to construct the scattering theory for massless Dirac fields outside a Kerr black hole. 
By using Mourre's theory again, Daud\'e \cite{Dau} proved the existence and asymptotic completeness of wave operators, classical at the event horizon and Dollard-modified at infinity, for classical massive Dirac particles on the Kerr-Newman spacetime; Riton \cite{Riton} studied the scattering for massive Dirac equations on the Schwarzschild-Anti-de Sitter spacetime.
On the other hand, Batic \cite{Batic} has provided another approach from \cite{HaNi2004} to construct the scattering theory for massive Dirac particles outside the event horizon of a nonextreme Kerr black hole spacetime. The method in \cite{Batic} is based on an integral representation of the Dirac propagator in the exterior region of the Kerr spacetime.

Conformal scattering theory is a geometric approach to construct the scattering for field equations on spacetimes in general relativity that is based on a conformal technique and vector field methods. The idea of the conformal compactification structure of spacetimes was posed initially by Penrose \cite{Pe1964} in the 1960's. Since then, this structure plays an important role in the study of peeling and conformal scattering, the two aspects of conformal asymptotic analysis. In particular, the conformal scattering theory (i.e., the geometric scattering theory) has been studied extensively from the early works by Friedlander \cite{Fri1962,Fri1964,Fri1967,Fri1980, Fri2001}, Baez et al. \cite{BaSeZho1990}, H\"ormander \cite{Ho1990} to recent ones by Mason and Nicolas \cite{MaNi2004}, Joudioux \cite{Jo2012,Jo2019}, Nicolas \cite{Ni2016}, Mokdad \cite{Mo2019,Mo2022}, Taujanskas \cite{Ta2019} and Pham \cite{Pha2020,Pha2022}.

The works of Nicolas and Mason \cite{MaNi2004} and Nicolas \cite{Ni2016} put farther a program of conformal scattering theories for the Dirac, Maxwell and scalar wave equations on the asymptotic simple or flat spacetimes. In particular, a conformal scattering theory on the exterior domains of the black hole spacetimes such as Schwarzschild and Kerr ones consists of three following steps: first, we prove the well-posedness of Cauchy problem of the rescaled equations on the rescaled spacetime, then we define and extend the trace operators $\mathcal{T}^{\pm}$ from the finite energy space of initial data on $\Sigma_0=\left\{t=0\right\}$ to the scattering data spaces on conformal boundaries. Second, we show that the extension of the trace operator is injective by proving the energy identity up to the future timelike infinity $i^+$. Third, we prove the well-posedness of Goursat problem with the initial data on  conformal boundaries (which is the scattering data); then as a consequence, we obtain that the extensions of the trace operators $\mathcal{T}^\pm$ are surjective. Therefore, the extended trace operator $\mathcal{T}^+$ (resp. $\mathcal{T}^-$) is an isometry between the space of the initial data on $\Sigma_0$ and the space of the future (resp. past ) scattering data on conformal boundaries. As a consequence, we define the conformal scattering operator $S:= \mathcal{T}^+\circ (\mathcal{T}^-)^{-1}$ that is an isometry that maps the past scattering data to the future scattering data.

Continuing this program, Mokdad \cite{Mo2019,Mo2022} constructed explicitly the conformal scattering theories for the Maxwell and Dirac equations on the exterior and interior of black hole of Reissner-Nordstr\"om de Sitter spacetime (which is outside a spherically symmetric charged body), respectively. On the other hand, Pham \cite{Pha2020} constructed conformal scattering theories for the scalar Reeger-Wheeler and Zerelli equations arising from the linearized gravity fields and the spin $\pm 2$ Teukolsky equations. This is the first step to obtain the conformal scattering theory for the linearized gravity fields on the Schwarzschild spacetime which is spherical symmetric. The extension of the conformal scattering theory on Kerr spacetime (which is non-static and non-spherical symmetric) has been established recently by Pham \cite{Pha2022} for the massless Dirac equations. In the works on the exterior domains of black hole spacetimes \cite{Mo2019,Pha2020,Pha2022}, the authors used the results about the uniformly bounded energy, Morawertz estimate and pointwise decay of the fields to establish the energy identity up to the future (resp. past) timelike infinity $i^+$ (resp. $i^-$) in the second step of the conformal scattering theory's construction. In order to prove the well-posedness of the Goursat problem, the authors used the generalization of H\"ormander's results (see \cite{Ho1990,Ni2006}) in the third step of the construction. 

There are some related works that also use the uniformly bounded energy and pointwise decay results to construct the scattering theory. We refer the readers to the works about the scattering theories for the scalar wave equation on the interior of Reissner-Nordstr\"om de Sitter by Keller et al. \cite{Ke2019}, on the extremal Reissner-Nordstr\"om spacetime by Angelopoulos et al. \cite{An2020}; on the exterior of slowly Kerr spacetime by Dafermos et al. \cite{Da2018}, and on Oppenheimer–Snyder spacetime by Alford \cite{Alford}. The uniformly bounded energy, Morawertz's estimate, energy and pointwise decays are obtained in the program to prove linear and nonlinear stability of black hole spacetimes and the related problems (see \cite{ABlu2,Da2016,Da2018,Da2019,El2020',El2021,El20222,Hu2018,Kl2018,Kl2021,Jo2019,Pa2019'}). The method of $r^p$-theory of Dafermos and Rodnianski \cite{DaRo2010} is an essential tool of the proof in a lot of later works.  
 
The spin $\pm 1$ Teukolsky equations are derived from the extreme components of the Maxwell fields (see Subsection \ref{Equation} and more details in \cite{Ba1973,Pa2019}).
There are two ways to establish the tensorial Fackerell-Ipser equations. The first one is obtained by commuting the spin $\pm 1$ Teukolsky equations with the projected covariant derivatives $\slashed{\nabla}_L$ and $\slashed{\nabla}_{\underline L}$ on the $2$-sphere $\mathbb{S}^2_{(t,r)}$ at $(t,r)$, where $L$ and $\underline{L}$ are outgoing and incoming principal null directions, respectively. The second one is obtained by commuting the scalar Fackerell-Ipser equation with the angular derivatives $r\slashed{\nabla}_{\partial_{x^a}}$. The potentials (which are of zero order in the term of derivatives) in the tensorial Fackerell-Ipser and Teukolsky equations decay as $r^{-2}$, whence the ones in the scalar Regger-Wheeler and Zerelli equations (see \cite{Pha2020}) and also the scalar (real or complex) Fackerell-Ipser equations (see \cite{ABlu2015,Blue2008}) decay as $r^{-3}$. 

The spin $\pm 1$ Teukolsky equations are studied in some recent works by Pasqualotto \cite{Pa2019}, Giorgi \cite{El2019} and Ma \cite{Ma20}. In particular, the authors used $r^p$-method (see \cite{DaRo2010}) to establish the boundedness of energy and study time decays of the associated solutions of Teukolsky equations on Schwarzschild, Reissner-Nordstr\"om and Kerr spacetimes in \cite{Pa2019,El2019,Ma20}, respectively. On the other hand, the peeling for spin $\pm 1$ Teukolsky equations on Schwarzschild spacetime has been studied by Pham in a recent work \cite{Pham2022}.

In this paper, we explore the method in \cite{Mo2019,Ni2016,Pha2020} to establish conformal scattering theories for the tensorial Fackerell-Ipser and spin $\pm 1$ Teukolsky equations on Schwarzschild spacetime. First, we construct the conformal scattering theories for the tensorial Fackerell-Ipser equations in Sections \ref{Fackerell-Ipser} and \ref{ConFac}. In Subsection \ref{ConserFac}, we establish the conservation law \eqref{zero1} for the tensorial Fackerell-Ipser equations by using the energy momentum tensor for tensorial wave equations and the Killing vector field $T=\partial_t$. 
Integrating this conservation law, we obtain the energy equality between the energy flux of solution throughs the initial hypersurface $\Sigma_0 = \left\{ t=0 \right\}$ and energy fluxes through the following null hypersurfaces: $\mathfrak{H}^+_K=\mathfrak{H}^+\cap \left\{ v\leq V_K\right\}$, $\mathcal{H}^+_{K} = \left\{ u=U_K,\, v\geq V_K\right\}$, $\mathcal{I}^+_K= \left\{ v=v_K,\, u\geq U_K \right\}$, $\scri^+_K= \scri^+\cap \left\{ u \leq U_K\right\}$. 
In Subsection \ref{FacSpace}, we define the finite energy spaces $\mathcal{H}(\Lambda^1(\mathbb{S}^2)|_{\Sigma_t})\, (t\geq 0)$ of tensorial fields, then we establish the well-posedness of Cauchy problem for tensorial Fackerell-Ipser equations by extending the method in the previous work of Saka \cite{Saka1985}. The well-posedness of Cauchy problem allows us to define the trace operator $\mathcal{T}^+$ (resp. $\mathcal{T}^-$) for the smooth solution of tensorial Fackerell-Ipser equation which maps the initial data (with smooth and compact support) to the restrictions of the smooth solution on the conformal boudary $\mathfrak{H}^+\cup \scri^+$ (resp. $\mathfrak{H}^-\cup \scri^-$).  

In order to prove the energy identity up to the timelike infinity $i^+$ (and also to $i^-$), we need to use the energy decay results obtained previously in the literature. The decays of the solution of the tensorial Fackerell-Ipser equations can be established from the ones of the scalar Fackerell-Ipser equations. There are some works on the decay of solutions of scalar Fackerell-Ipser equations in Schwarzschild spacetime such as \cite{Blue2008,Ga2014,Me2012}. However, in this work, we will use the energy decay results which have been obtained in a recent work of Pasqualotto \cite{Pa2019}. This energy decay helps us to prove that the energy fluxes through null hypersurfaces $\mathcal{H}^+_{K} = \left\{ u=U_K,\, v\geq V_K\right\}$ and $\mathcal{I}^+_K= \left\{ v=v_K,\, u\geq U_K \right\}$ tend to zero as $U_K$ and $V_K$ tend to infinity. This together with the energy equality obtained in Subsection \ref{ConserFac} lead to the energy identity up to $i^+$, i.e., the energy flux of tensorial Fackerell-Ipser solution through the initial hypersurface $\Sigma_0=\left\{ t=0 \right\}$ is equal to the sum of energy fluxes of solution through the future hoziron $\mathfrak{H}^+$ (resp. the past horizon $\mathfrak{H}^-$) and the future infinity $\scri^+$ (resp. the past infinity $\scri^-$) (see Theorem \ref{EqualityFac}). Therefore, we can extend the future trace operator to an injective operator: $\mathcal{T}^+: \mathcal{H}\to \mathcal{H}^+$ between the finite energy space on $\Sigma_0 = \left\{ t=0 \right\}$ and the scattering data spaces on $\mathfrak{H}^+\cup \scri^+$ (see Theorem \ref{Trace}). Similarly, the extended past trace operator $\mathcal{T}^-: \mathcal{H} \to \mathcal{H}^-$ is also injective. Here, the spaces $\mathcal{H}^+$ (resp. $\mathcal{H}^-$) is the scattering data space which is completion of smooth and compact support tensorial fields on $\mathfrak{H}^+\cup \scri^+$ (resp. $\mathfrak{H}^-\cup \scri^-$) under energy norm (see Definition \ref{ScatteringDataFac}).

In Section \ref{ConFac} we prove that the trace operator is surjective. For this purpose, we establish the well-posedness of the Goursat problem with the smoothly supported compact initial data on the conformal boundary $\mathfrak{H}^+\cup \scri^+$ (resp. $\mathfrak{H}^-\cup \scri^-$). This work is done by developing H\"ormander's work \cite{Ho1990}, for the tensorial wave equations on Schwarzschild spacetime. We project the tensorial Fackerell-Ipser equations on the basic frame of the unit $2$-sphere $\mathbb{S}^2$, we get a symmetrical hyperbolic system which consists of two scalar wave equations with potentials at the first order of derivatives. The well-posedness of the Goursat problem consists of two parts: in the first one, we extend the results in \cite{Ho1990} to solve the Goursat problem of the symmetrical hyperbolic system in the future of a spacelike hypersurface $\mathcal{S}$ which intersects with the horizon at the crossing sphere and crosses $\scri^+$ strictly in the past of the support of the data (in details see Lemma \ref{partlyGoursat}, Corollary \ref{partGo} and Appendix \ref{appendix}); in the second one, we extend the solution obtained in the first part down to $\Sigma_0$, where the method is developed from \cite{Ni2016} (see Theorem \ref{Goursat}). The well-posedness of the Goursat problem shows that the extended trace operator $\mathcal{T}^+$ (resp. $\mathcal{T}^-$) is surjective, hence an isometry. Therefore, we can define the conformal scattering operator $S: \mathcal{H}^-\to \mathcal{H}^+$ for the tensorial Fackerell-Ipser equations that maps the past scattering data to the future scattering data by
$$S:= \mathcal{T}^+\circ (\mathcal{T}^-)^{-1}.$$

Finally, in Section \ref{ConTeu}, we will construct the conformal scattering theories for spin $\pm 1$ Teukolsky equations by using the results obtained in Sections \ref{Fackerell-Ipser} and \ref{ConFac}.
Our method is developed from a recent work of  Masaood (see \cite{Masao}) for the scattering theories of the spin $\pm 2$ Teukolsky equations. In Subsection \ref{TraceScat}, we prove that we can define a $\mathcal{H}^1$-norm of tensorial fields on the spacelike hypersurface $\Sigma_\tau = \left\{ t=\tau \right\}$ which satisfies the spin $+1$ Teukolsky equation via the norm of corresponding tensorial Fackerell-Ipser field (see Proposition \ref{norm}). Then, we prove the well-posedness of the Cauchy problem of the spin $+1$ Teukolsky equations for the initial data in $\mathcal{H}^1(\Lambda^1(\mathbb{S}^2)|_{\Sigma_0})$ (see Theorem \ref{CauchyTeu}) by extending the method in \cite{Saka1985}.
We define the trace operator $\mathfrak{T}^+$ (resp. $\mathfrak{T}^-$) in Definition \ref{TRACETEU} and the energy space $\mathcal{H}^{2,+}$ (resp. $\mathcal{H}^{2,-}$) on the conformal boundary $\mathfrak{H}^+\cup \scri^+$ (resp. $\mathfrak{H}^-\cup \scri^-$) in Definition \ref{normboundary}. 
By using the equality energy obtained for tensorial Fackerell-Ipser equation and the $\mathcal{H}^1$-norm defined on the solution of spin $+1$ Teukolsky equation, we prove that
the extended trace operator $\mathfrak{T}^+:\mathcal{H}^1\to \mathcal{H}^{2,+}$ under $\mathcal{H}^1$-energy norm is injective (see Theorem \ref{TeuTraceInj}). 

In Subsection \ref{GoursatTeu}, we use the well-posedness of Goursat problem of tensorial Fackerell-Ipser equations to prove the one for the Teukolsky equations (see Theorem \ref{GGoursatTeu} and Theorem \ref{GGoursatTeu'}). The well-posedness of Goursat problem shows that the extended trace operator $\mathfrak{T}^+: \mathcal{H}^1\to \mathcal{H}^{2,+}$ (resp. $\mathfrak{T}^-: \mathcal{H}^1\to \mathcal{H}^{2,-}$) is surjective, hence $\mathfrak{T}^+$ is an isometric operator. The conformal scattering operator $\mathfrak{S}: \mathcal{H}^{2,-}\to \mathcal{H}^{2,+}$ for spin $+1$ Teukolsky equation that maps the past scattering data to the future scattering data are given by
$$\mathfrak{S}:= \mathfrak{T}^+\circ (\mathfrak{T}^-)^{-1}.$$
{\bf Notation.}\\
Through this paper, we follow the notations which were used in \cite{Pa2019,Pa2019'} (see also \cite{Chris2009,Da2019}) on the round metric and projected covariant derivatives on the $2$-sphere $\mathbb{S}^2_{(t,r)}$.\\
$\bullet$ We denote the bundle tangent to each $2$-sphere $\mathbb{S}^2_{(t,r)}$ at $(t,r)$ by $\mathcal{B}$ and the vector space of all smooth sections of $\mathcal{B}$ by $\Gamma(\mathcal{B})$. 
We denote local coordinates for $\mathbb{S}^2_{(t,r)}$ by $(x^a,x^b)$ and the associated vector fields to $x^a,\,x^b$ by $\partial_{x^a},\,\partial_{x^b}$, respectively. The space of all $1$-forms on $\mathbb{S}^2_{(t,r)}$ is denoted by $\Lambda^1(\mathcal{B})$.\\
$\bullet$ We denote the metric on $2$-sphere $\mathbb{S}^2_{(t,r)}$ by $\slashed{g}$. Note that $\slashed{g}$ is a round metric and $\slashed{g}=r^2g_{\mathbb{S}^2}$, where $g_{\mathbb{S}^2}$ is the metric on the unit $2$-sphere $\mathbb{S}^2$. \\
$\bullet$ Let $V,\, W\in \Gamma(\mathcal{B})$. We define a connection on $\mathcal{B}$ by
$$\slashed{\nabla}_VW = \left( \nabla_VW\right)^\perp,$$
where $(\cdot)^\perp: T\mathcal{M} \to \mathcal{B}$ is the orthogonal projection on the $2$-sphere $\mathbb{S}^2(t,r)$ for a given $(t,r)$. Here, $\mathcal{M}$ denotes the region outside the Schwarzschild black-hole equipped with the metric $g$ (see Subsection \ref{PENROSE}). 
This connection coincides with the Levi-Civita connection associated with the metric $\slashed{g}$.\\
$\bullet$ For $V\in \Gamma(\mathcal{B})$, there are two other covariant operators (projected covariant derivatives) which are defined by
$$\slashed{\nabla}_L V= \left( \nabla_L V \right)^{\perp},\, \slashed{\nabla}_{\underline L} V = \left(  \nabla_{\underline L} V\right)^{\perp}, $$
where $\nabla$ is the Levi-Civita connection on $(\mathcal{M},g)$ and $L$, $\underline{L}$ are outgoing and incoming principal null directions (see Subsection \ref{PENROSE}).\\
$\bullet$ We denote local coordinates for the unit $2$-sphere $\mathbb{S}^2$ by $(\theta^a,\theta^b)$ and the associated vector fields to $\theta^a,\theta^b$ by $\partial_{\theta^a}$ and $\partial_{\theta^b}$, respectively. Normaly, we have $(\theta^a,\theta^b) =(\theta, \varphi)$.\\
$\bullet$ The space of $1$-forms on the unit $2$-sphere is denoted by $\Lambda^1(\mathbb{S}^2)$. The basic frame of $\Lambda^1(\mathbb{S}^2)$ is denoted by $(\slashed{\nabla}_{\partial_{\theta^a}},\,\slashed{\nabla}_{\partial_{\theta^b}})$, where $\slashed{\nabla}_{\partial_{\theta^a}}$ is the Levi-Civita connection associated with the metric $g_{\mathbb{S}^2}$, follows the vector field $\partial_{\theta^a}$. On the $2$-sphere $\mathbb{S}^2_{(t,r)}$, we have the relation $r\slashed{\nabla}_{\partial_{x^a}} =\slashed{\nabla}_{\partial_{\theta^a}}$. \\
$\bullet$ We denote the covariant Laplacian operator associated with the round metric $\slashed{g}$ on $\mathbb{S}^2_{(t,r)}$ by $\slashed{\Delta}$ and the one associated with the metric $g_{\mathbb{S}^2}$ on unit sphere $\mathbb{S}^2$ by $\slashed{\Delta}_{\mathbb{S}^2}$. We use the definition $\slashed{\Delta} = \slashed{g}^{ab}\slashed{\nabla}_{\partial_{x^a}}\slashed{\nabla}_{\partial_{x^b}}$ through this paper. Follows this definition, we have $\slashed{\Delta}_{\mathbb{S}^2} =r^2\slashed{\Delta}$.
\\
$\bullet$ Beside, we denote the space of smooth compactly supported scalar functions on $\mathcal{M}$ (a smooth manifold without boundary) by ${C}_0^\infty(\mathcal{M})$ and the space of distributions on $\mathcal{M}$ by $\mathcal{D}'(\mathcal{M})$. The space of smooth compactly supported $1$-forms in $\Lambda^1(\mathbb{S}^2)$ on $\mathcal{M}$ is denoted by ${C}_0^\infty(\Lambda^1(\mathbb{S}^2)|_{\mathcal{M}})$.\\
$\bullet$ Let $f(x)$ and $g(x)$ be two real functions. We write $f \lesssim g$ if there exists a constant $D \in (0,+\infty)$ which does not depend on $f,\, g$ and $x$, such that $f(x)\leq D g(x)$ for all $x$, and write $f\simeq g$ if both $f\lesssim g$ and $g\lesssim f$ are valid.\\
{\bf Acknowledgements.} The author would like to thank Prof. Jean-Philippe Nicolas (LMBA, Brest University) for some helpful discussions when this work started.
This work is supported by Vietnam Institute for Advanced Study in Mathematics (VIASM) 2023.

\section{Geometrical and analytical setting} 

\subsection{Schwarzschild metric and Penrose's conformal compactification}\label{PENROSE}
We consider the region outside the Schwarzschild black hole $({\cal{M}}=\mathbb{R}_t\times ]2M,+\infty[_r\times \mathbb{S}^2,g)$, equipped with the Lorentzian metric $g$ given by
$$g = F \d t^2 - F^{-1}\d r^2 - r^2\d\mathbb{S}^2, \, F=F(r)=1-\mu,\, \mu=\frac{2M}{r},$$
where $\d \mathbb{S}^2$ is the euclidean metric on the unit $2$-sphere $\mathbb{S}^2$, and $M>0$ is the mass of the black hole.

We recall that the Regge-Wheeler coordinate $r_*=r+2M \log(r-2M)$ which satisfies $\d r=F\d r_*$. In the coordinates $(t,r_*,\theta^a,\theta^b)$, the Schwarzschild metric takes the form
$$g = F(\d t^2- \d r_*^2) - r^2\d\mathbb{S}^2.$$
The retarded and advanced Eddington-Finkelstein coordinates $u$ and $v$ are defined by
$$u=t-r_*, \, v= t+r_*.$$
The outgoing and incoming principal null directions are 
$$L =\partial_v = \partial_t + \partial_{r_*}, \, \underline{L} = \partial_u = \partial_t - \partial_{r_*},$$
respectively.

Putting $\Omega = 1/r$ and $\hat{g} = \Omega^2g$. We obtain a conformal compactification of the exterior domain in the retarded variables $(u, \, R = 1/r, \, \theta^a,\, \theta^b)$ that is $\left(\mathbb{R}_u\times \left[0,\dfrac{1}{2M}\right] \times \mathbb{S}^2, \hat{g} \right)$ with the rescaled metric 
\begin{equation}
\hat{g} = R^2F \d u^2 - 2\d u\d R - \d\mathbb{S}^2.
\end{equation}
The future null infinity $\scri^+$ and the past horizon $\mathfrak{H}^-$ are null hypersurfaces of the rescaled spacetime
$$\scri^+ = \mathbb{R}_u \times \left\{ 0\right\}_R \times \mathbb{S}^2, \, \mathfrak{H}^- = \mathbb{R}_u \times \left\{ 1/2M\right\}_R \times \mathbb{S}^2.$$
If we use the advanced variables $(v, \, R=1/r, \, \theta^a,\, \theta^b)$, the rescaled metric $\hat{g}$ takes the form
\begin{equation}
\hat{g} = R^2F\d v^2 + 2 \d v \d R - \d\mathbb{S}^2.
\end{equation}
The past null infinity $\scri^-$ and the future horizon $\mathfrak{H}^+$ are described as the null hypersurfaces
$$\scri^- = \mathbb{R}_v \times \left\{ 0\right\}_R \times \mathbb{S}^2, \, \mathfrak{H}^+ = \mathbb{R}_v \times \left\{ 1/2M\right\}_R \times \mathbb{S}^2.$$
Penrose's conformal compactification of $\mathcal{M}$ is
$$\bar{\mathcal{M}} = \mathcal{M} \cup \scri^+ \cup \mathfrak{H}^+ \cup \scri^-\cup \mathfrak{H}^-\cup S_c^2,$$
where $S_c^2$ is the crossing sphere which is an intersection of $\mathfrak{H}^+$ and $\mathfrak{H}^-$.  The construction of $S_c^2$ can be done by using Kruskal-Szekeres coordinates (see \cite{Hawking,Wald}).

Note that, the compactified spacetime $\bar{\mathcal{M}}$ is not compact. There are three ``points'' missing to the boundary: $i^+$, or future timelike infinity, defined as the limit point of uniformly timelike curves as $t\to + \infty$; $i^-$, past timelike infinity, symmetric of $i^+$ in the distant past, and $i_0$, spacelike infinity, the limit point of uniformly spacelike curves as $r\to +\infty$.  These ``points'' are singularities of the rescaled metric $\hat{g}$.

\begin{figure}[H]
\begin{center}
\includegraphics[scale=0.3]{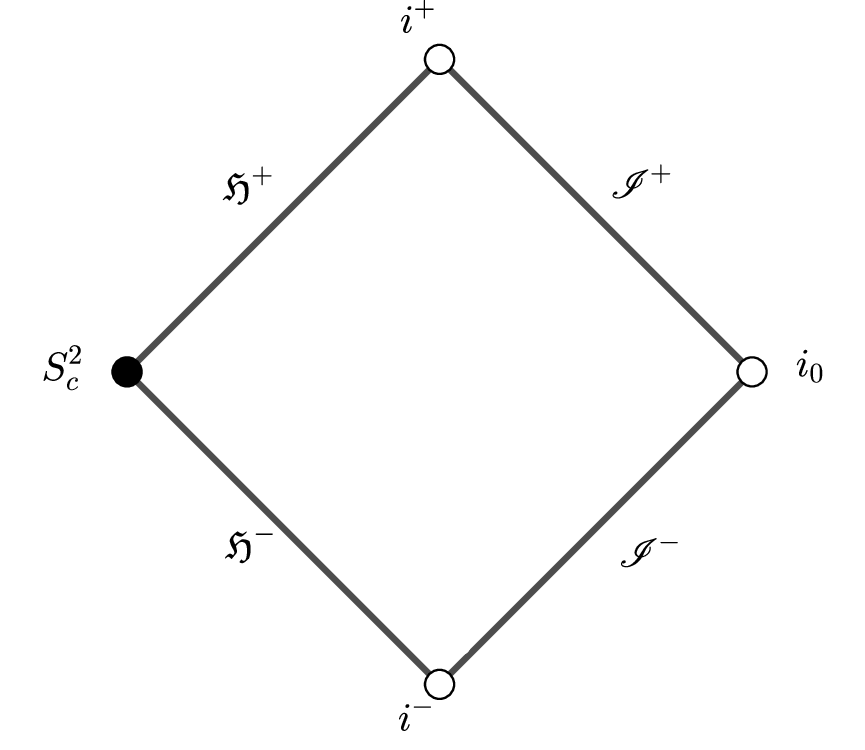}
\caption{Penrose's conformal compactification of $\mathcal{M}$.}
\end{center}
\end{figure}

In the retarded coordinates $(u,\, R, \, \theta^a,\, \theta^b)$, we have the following relations
\begin{equation}\label{Cor1}
\partial_R = - \frac{r^2}{F}(\partial_t + \partial_{r_*}) = -\frac{r^2}{F}L,
\end{equation}
and
\begin{equation}\label{Cor2}
\hat{L}=r^2L,\, \hat{\underline{L}} = \underline{L},\, \partial_{{\theta}^a} = r\partial_{x^a},\, \partial_{{\theta}^b} = r\partial_{x^b}.
\end{equation}
In the advanced coordinates $(v,\, R, \, \theta^a,\, \theta^b)$, we have the following relations
\begin{equation}\label{Cor3}
\partial_R = - \frac{r^2}{F}(\partial_t - \partial_{r_*}) = -\frac{r^2}{F}\underline{L},
\end{equation}
and
\begin{equation}\label{Cor4}
\hat{L}=L,\, \hat{\underline{L}} = r^2\underline{L},\, \partial_{{\theta}^a} = r\partial_{x^a},\, \partial_{{\theta}^b} = r\partial_{x^b},
\end{equation}
where $(x^a,x^b)$ denote local coordinates for the $2$-sphere $\mathbb{S}^2_{(t,r)}$ at $(t,r)$.

The scalar curvature of the rescaled metric $\hat g$ is
$$\mathrm{Scal}_{\hat g}= 12MR.$$

The volume forms associated with the Schwarzschild metric $g$ and the rescaled metric $\hat{g}$ are
$$\mathrm{dVol}_g = r^2F \d t \wedge \d r_* \wedge \d \mathbb{S}^2 \hbox{  and  } \mathrm{dVol}_{\hat g} = R^2F \d t \wedge \d r_* \wedge \d\mathbb{S}^2 = \frac{R^2F}{2}\d u \wedge \d v \wedge \d \mathbb{S}^2,$$
respectively, where $\d \mathbb{S}^2$ is the euclidean area element on unit $2$-sphere $\mathbb{S}^2$.

\subsection{The Maxwell and tensorial wave equations}\label{Equation}
Let $\mathbb{F}$ be an antisymmetric $2$-form on the exterior domain of Schwarzschild black hole $\mathcal{M}$. The Maxwell equations take the form
\begin{equation*}
\d \mathbb{F} = 0, \,\,\, \d *\mathbb{F}=0,
\end{equation*}
where $*$ denotes the Hodge dual operator of $2$-form, i.e,
$$(*\mathbb{F})_{\mu\nu} = \frac{1}{2}e_{\mu\nu\gamma\delta}\mathbb{F}^{\gamma\delta}.$$
The system can be reformulated as follows
\begin{equation*}
\nabla_{[\mu}\mathbb{F}_{\kappa\lambda]}=0, \,\,\, \nabla^\mu \mathbb{F}_{\mu\nu}=0,
\end{equation*}
where the square brackets denote antisymmetrization of indices.

The Maxwell field $\mathbb{F}$ can be decomposed into $1$-forms $\alpha_a, \, \underline{\alpha}_a \in \Lambda^1(\mathcal{B})$ and $\rho,\, \sigma \in C^{\infty}(\mathcal{M})$ which are defined as follows
\begin{gather*}
\alpha(V):= \mathbb{F}(V,L), \, \underline{\alpha}(V):= \mathbb{F}(V,\underline{L}) \hbox{  for all  } V\in \Gamma(\mathcal{B}),\\
\rho:= \frac{1}{2} \left( 1 - \frac{2M}{r} \right)^{-1}\mathbb{F}(\underline{L},L),\, \sigma:= \frac{1}{2}e^{cd}\mathbb{F}_{cd},
\end{gather*}
where $e_{cd}\in \Lambda^2(\mathcal{B})$ is the volume form of $2$-sphere $\mathbb{S}^2_{(t,r)}$ at $(t,r)$.

Let $\mathbb{F}$ be in $\Lambda^2(\mathcal{M})$ such that $\mathbb{F}$ satisfies the Maxwell equation on $\mathcal{M}$. Then, we have the following formulas (see \cite[Proposition 3.6]{Pa2019}):
\begin{equation*}
\frac{1}{r}\slashed{\nabla}_L(r\underline{\alpha}_a) = - (1-\mu)(\slashed{\nabla}_a\rho - e_{ab}\slashed{\nabla}^b\sigma) 
\end{equation*} 
and
\begin{equation*}
\frac{1}{r}\slashed{\nabla}_{\underline L}(r\alpha_a) = (1-\mu)(\slashed{\nabla}_a\rho + e_{ab}\slashed{\nabla}^b\sigma).
\end{equation*}
From this, we can define the $1$-forms in $\Lambda^1(\mathcal{B})$:
\begin{equation}\label{Tran}
\phi_a := \frac{r^2}{F}\slashed{\nabla}_{\underline L}(r\alpha_a), \,\,\, \underline{\phi}_a := \frac{r^2}{F}\slashed{\nabla}_{L}(r\underline{\alpha}_a).
\end{equation}
Moreover, the extreme components $\alpha_a$ and $\underline{\alpha}_a$ satisfy the spin $\pm 1$ Teukolsky equations, respectively (see original proof in \cite{Ba1973} and recent \cite[Proposition 3.6]{Pa2019}):
\begin{equation}\label{Teu1}
\slashed{\nabla}_L\slashed{\nabla}_{\underline L}(r\alpha_a) + \frac{2}{r}\left( 1-\frac{3M}{r} \right)\slashed{\nabla}_{\underline L}(r\alpha_a) - F\slashed\Delta(r\alpha_a) + \frac{F}{r^2}r\alpha_a=0,
\end{equation}
\begin{equation}\label{Teu2}
\slashed{\nabla}_L\slashed{\nabla}_{\underline L}(r\underline{\alpha}_a) - \frac{2}{r}\left( 1-\frac{3M}{r} \right)\slashed{\nabla}_{L}(r\underline{\alpha}_a) - F\slashed\Delta(r\underline{\alpha}_a) + \frac{F}{r^2}r\underline{\alpha}_a=0,
\end{equation}
where $F=1-2MR$ and $\slashed\Delta$ is the covariant Laplacian operator associated with the round metric $\slashed{g}$ on the $2$-sphere $\mathbb{S}^2_{(t,r)}$.

The tensorial Fackerell-Ipser equations are established from the spin $\pm 1$ Teukolsky equations by the following proposition (see also \cite[Proposition 3.7]{Pa2019}).
\begin{proposition}\label{relationEq}
Suppose that $(\alpha_a,\underline{\alpha}_a,\rho,\sigma)$ satisfy the Maxwell equation, then the $1$-forms $\phi_a$ and $\underline{\phi}_a$ satisfy the following tensorial Fackerell-Ipser equations
\begin{equation}\label{Fac01}
\slashed{\Box}_{\hat g}(\phi_a) + \phi_a=0,
\end{equation}
\begin{equation}\label{Fac02}
\slashed{\Box}_{\hat g}(\underline{\phi}_a) + \underline{\phi}_a=0,
\end{equation}
where we denote the tensorial wave operator (also called the tensorial Fackerell-Ipser operator) by
$$\slashed{\Box}_{\hat g} = r^2\slashed{\Box}_g =  \frac{r^2}{F}\slashed{\nabla}_L\slashed{\nabla}_{\underline L} - \slashed{\Delta}_{\mathbb{S}^2},$$
with $\slashed{\Delta}_{\mathbb{S}^2}$ is the covariant Laplacian operator associated with the metric $g_{\mathbb{S}^2}$ on the unit sphere $\mathbb{S}^2$.
\end{proposition}
\begin{proof}
We treat the equation for $\underline{\phi}_a$, the one for $\phi_a$ is obtained similarly. A straightforward calculation gives
$$\underline{L}\left( \frac{r^2}{F} \right) = -\frac{r^2}{F}\frac{2}{r}\left( 1 - \frac{3M}{r} \right).$$
Hence, the Teukolsky equation \eqref{Teu2} is equivalent to
\begin{equation*}
\frac{F}{r^2}\slashed{\nabla}_{\underline{L}}\left( \frac{r^2}{F}\slashed{\nabla}_L(r\underline{\alpha}_a) \right) - F\slashed{\Delta}(r\underline{\alpha}_a) + \frac{F}{r^2}r\underline{\alpha}_a =0.
\end{equation*}
Therefore,
\begin{equation*}
\slashed{\nabla}_{\underline{L}}\left( \frac{r^2}{F}\slashed{\nabla}_L(r\underline{\alpha}_a) \right) - r^2\slashed{\Delta}(r\underline{\alpha}_a) + r\underline{\alpha}_a =0.
\end{equation*}
By applying $\slashed{\nabla}_L$ to the above equation with noting that $[\slashed{\nabla}_L,r^2\slashed{\Delta}] = 0$ and $[\slashed{\nabla}_L,\slashed{\nabla}_{\underline L}] =0$, we get
\begin{equation*}
\slashed{\nabla}_L\slashed{\nabla}_{\underline L}(\underline{\phi}_a) - F\slashed{\Delta}(\underline{\phi}_a)+ \frac{F}{r^2}\underline{\phi}_a=0.
\end{equation*}
This equation is equivalent to \eqref{Fac02} because $r^2\slashed{\Delta} = \slashed{\Delta}_{\mathbb{S}^2}$.
\end{proof}
\begin{remark}
We have the following expressions of tensorial Fackerell-Ipser equations \eqref{Fac01} and \eqref{Fac02} in the retarded coordinates and advanced coordinates in $(\mathcal{M},\hat{g})$: 

In the retarded coordinates $(u,\,R,\,\theta^a,\,\theta^b)$: by using relations \eqref{Cor1} and \eqref{Cor2}, the tensorial Fackerell-Ipser \eqref{Fac01} has the following form
\begin{eqnarray}\label{ReFac011}
\slashed{\Box}_{\hat g}{\phi}_a + {\phi}_a&=& \frac{r^2}{F}\slashed{\nabla}_L\slashed{\nabla}_{\underline{L}}\phi_a - \slashed{\Delta}_{\mathbb{S}^2}\phi_a + \phi_a \cr
&=&\frac{1}{F}\slashed{\nabla}_{\hat L}\slashed{\nabla}_{\hat{\underline{L}}}{\phi}_a - \slashed{\Delta}_{\mathbb{S}^2}{\phi}_a + {\phi}_a \cr
&=& -2\slashed{\nabla}_u\slashed{\nabla}_R{\phi}_a - \slashed{\nabla}_R R^2(1-2MR)\slashed{\nabla}_R{\phi}_a - \slashed{\Delta}_{\mathbb{S}^2}{\phi}_a + {\phi}_a = 0
\end{eqnarray}
and the tensorial Fackerell-Ipser \eqref{Fac02} has the following form
\begin{eqnarray}\label{ReFac022}
\slashed{\Box}_{\hat g}\underline{\phi}_a + \underline{\phi}_a&=& \frac{r^2}{F}\slashed{\nabla}_L\slashed{\nabla}_{\underline{L}}\, {\underline{\phi}_a} - \slashed{\Delta}_{\mathbb{S}^2}\underline{\phi}_a + \underline{\phi}_a\cr
&=&\frac{1}{F}\slashed{\nabla}_{\hat L}\slashed{\nabla}_{\hat{\underline{L}}}\underline{\phi}_a - \slashed{\Delta}_{\mathbb{S}^2}\underline{\phi}_a + \underline{\phi}_a \cr
&=& -2\slashed{\nabla}_u\slashed{\nabla}_R\underline{\phi}_a - \slashed{\nabla}_R R^2(1-2MR)\slashed{\nabla}_R\underline{\phi}_a - \slashed{\Delta}_{\mathbb{S}^2}\underline{\phi}_a + \underline{\phi}_a = 0.
\end{eqnarray}

In the advanced coordinates $(v,\,R,\,\theta^a,\, \theta^b)$: by using relations \eqref{Cor3} and \eqref{Cor4}, the tensorial Fackerell-Ipser equation \eqref{Fac01} has the following form
\begin{eqnarray}\label{ReFac01}
\slashed{\Box}_{\hat g}\phi_a + \phi_a &=& \frac{r^2}{F}\slashed{\nabla}_{\underline{L}}\slashed{\nabla}_L\phi_a - \slashed{\Delta}_{\mathbb{S}^2}\phi_a + \phi_a\cr
&=& \frac{1}{F}\slashed{\nabla}_{\hat{\underline{L}}}\slashed{\nabla}_{\hat L}\phi_a - \slashed{\Delta}_{\mathbb{S}^2}\phi_a + \phi_a \cr
&=& -2\slashed{\nabla}_v\slashed{\nabla}_R\phi_a - \slashed{\nabla}_R R^2(1-2MR)\slashed{\nabla}_R\phi_a - \slashed{\Delta}_{\mathbb{S}^2}\phi_a + \phi_a = 0
\end{eqnarray}
and the tensorial Fackerell-Ipser equation \eqref{Fac02} has the following form
\begin{eqnarray}\label{ReFac02}
\slashed{\Box}_{\hat g}\underline{\phi}_a + \underline{\phi}_a&=& \frac{r^2}{F}\slashed{\nabla}_{\underline{L}}\slashed{\nabla}_L\underline{\phi}_a - \slashed{\Delta}_{\mathbb{S}^2}\underline{\phi}_a + \underline{\phi}_a\cr
&=&\frac{1}{F}\slashed{\nabla}_{\hat{\underline{L}}}\slashed{\nabla}_{\hat L}\underline{\phi}_a - \slashed{\Delta}_{\mathbb{S}^2}\underline{\phi}_a + \underline{\phi}_a \cr
&=& -2\slashed{\nabla}_v\slashed{\nabla}_R\underline{\phi}_a - \slashed{\nabla}_R R^2(1-2MR)\slashed{\nabla}_R\underline{\phi}_a - \slashed{\Delta}_{\mathbb{S}^2}\underline{\phi}_a + \underline{\phi}_a = 0.
\end{eqnarray}

\end{remark}

Another way to obtain the tensorial Fackerell-Ipser equations \eqref{Fac01} and \eqref{Fac02} is to use the scalar Fackerell-Ipser equation. In particular, since $L= \partial_t + \partial_{r_*}, \, \underline{L}= \partial_t - \partial_{r_*}$ and $\slashed{\Delta} = \frac{1}{r^2}\slashed{\Delta}_{\mathbb{S}^2}$, the scalar wave equation on Schwarzschild spacetime can be expressed as
\begin{eqnarray*}
\Box_g \psi &=& \frac{1}{F} \left( \partial_t^2 - \frac{1}{r^2}\partial_{r_*}r^2\partial_{r_*} \right)\psi - \frac{1}{r^2}\slashed\Delta_{\mathbb{S}^2}\psi \cr
&=& \frac{1}{F}L\underline{L}\psi -\frac{2}{r}\partial_{r_*}\psi - \frac{1}{r^2}\slashed{\Delta}_{\mathbb{S}^2}\psi\cr
&=& \frac{1}{F}L\underline{L}\psi -\frac{2}{r}\partial_{r_*}\psi - \slashed{\Delta}\psi.
\end{eqnarray*}
Hence
$$\Box_g \psi + \frac{2}{r}\partial_{r_*}\psi = \frac{1}{F}L\underline{L}\psi - \slashed{\Delta}\psi.$$
The right-hand side is the scalar Fackerell-Ipser operator which has the same form as the rescaled scalar wave operator by multiplying the factor $r^2$ due to
$$\Box_{\hat g} = \frac{r^2}{F}L\underline{L}\psi - \slashed{\Delta}_{\mathbb{S}^2}\psi.$$
Moreover, we have the following relations
$$\phi_a = r^3(\slashed\nabla_{\partial_{x^a}} \rho + e_{ab}\slashed\nabla^{\partial_{x^b}}\sigma),\,\,\, \underline{\phi}_a = r^3(-\slashed\nabla_{\partial_{x^a}} \rho + e_{ab}\slashed\nabla^{\partial_{x^b}}\sigma),$$
where $e_{ab}$  is the induced volume form on the sphere $\mathbb{S}^2_{(t,r)}$ and the scalar functions $r^2\rho$, $r^2\sigma$ satisfy the scalar Fackerell-Ipser equation (see \cite[Remark 2.10]{Pa2019} or \cite[Appendix D.1]{Pa2019'}):
\begin{equation}\label{Scalar}
\Box_g\psi + \frac{2}{r}\partial_{r_*}\psi= \frac{1}{F}L\underline{L}\psi - \slashed{\Delta}\psi=0.
\end{equation}
We have the following commutators on scalar fields (see the proof in Appendix \ref{appen}):
\begin{equation}\label{commutator1}
[r\slashed\nabla_{\partial_{x^a}},\, \slashed{\nabla}_L]= [r\slashed\nabla_{\partial_{x^a}},\slashed{\nabla}_{\underline{L}}]=0, \, [r\slashed{\nabla}_{\partial_{x^a}},\slashed{\Delta}]= \dfrac{1}{r^2}(r\slashed{\nabla}_{\partial_{x^a}}).
\end{equation}
Commuting the covariant angular derivative $r\slashed{\nabla}_{\partial_{x^a}}$ and its Hodge dual $r(e_{ab}\slashed{\nabla}^{\partial_{x^b}})$ to the scalar wave equation \eqref{Scalar} with $\psi= \pm r^2\rho$ and $\psi= r^2\sigma$, respectively; then by using the commutators \eqref{commutator1}, we get the tensorial Fackerell-Ipser equations \eqref{Fac01} and \eqref{Fac02} (see also  \cite[Remark 2.10]{Pa2019} and \cite[Remark 7.1]{Da2019}).

Since $r^2\rho$ and $r^2\sigma$ satisfy Equation \eqref{Scalar}, we have that $r\rho$ and $r\sigma$ satisfy the scalar wave equation $\Box_g\psi -\frac{2M}{r^2}\psi = 0$.
Applying the projected covariant angular derivative $r\slashed{\nabla}$ to this equation, we get
\begin{eqnarray}\label{TensorWave}
&&\slashed{\Box}_g \widetilde{\phi}_a - \frac{2}{r}\slashed{\nabla}_{r_*}\widetilde{\phi}_a + \frac{1-2M}{r^2} \widetilde{\phi}_a \cr
&=& \frac{1}{F}\slashed{\nabla}_L\slashed{\nabla}_{\underline L}\widetilde{\phi}_a - \slashed{\Delta}\widetilde{\phi}_a - \frac{2}{r}\slashed{\nabla}_{r_*}\widetilde{\phi}_a + \frac{1-2M}{r^2}\widetilde{\phi}_a \cr
&=& \frac{1}{F}\slashed{\nabla}_L\slashed{\nabla}_{\underline L}\widetilde{\phi}_a - \slashed{\Delta}\widetilde{\phi}_a - \frac{1}{r}(\slashed{\nabla}_{L} - \slashed{\nabla}_{\underline{L}})\widetilde{\phi}_a + \frac{1-2M}{r^2}\widetilde{\phi}_a =0,
\end{eqnarray}
where
$$\widetilde{\phi}_a = r[\slashed\nabla_a (r\rho) + e_{ab}\slashed\nabla^b (r\sigma)] = r^2[\slashed\nabla_a \rho + e_{ab}\slashed\nabla^b \sigma].$$
Similarly,
\begin{equation}\label{TensorWave1}
\frac{1}{F}\slashed{\nabla}_L\slashed{\nabla}_{\underline L}\underline{\widetilde{\phi}}_a - \slashed{\Delta}\underline{\widetilde{\phi}}_a - \frac{1}{r}(\slashed{\nabla}_{L} - \slashed{\nabla}_{\underline{L}})\underline{\widetilde{\phi}}_a + \frac{1-2M}{r^2}\underline{\widetilde{\phi}}_a =0,
\end{equation}
where
$$\underline{\widetilde{\phi}}_a = r[-\slashed\nabla_a (r\rho) + e_{ab}\slashed\nabla^b (r\sigma)] = r^2[-\slashed\nabla_a \rho + e_{ab}\slashed\nabla^b \sigma].$$
The conformal rescaled equation of \eqref{TensorWave} (resp. \eqref{TensorWave1}) is the tensorial Fackerell-Ipser equation \eqref{Fac01} (resp. \eqref{Fac02}). Therefore, Equation \eqref{Fac01} (resp. \eqref{Fac02}) can be considered as a conformal equation in the conpactification domain $(\bar{\mathcal{M}},\hat{g})$.

In the rest of this paper, we will construct the conformal scattering theory for the tensorial wave equation \eqref{TensorWave} (resp. \eqref{TensorWave1}), i.e, the scattering theory for the tensorial Fackerell-Ipser equation \eqref{Fac01} (resp. \eqref{Fac02}). Then, using the scattering for the tensorial Fackerell-Ipser equations we will establish the scattering for the Teukolsky equations \eqref{Teu1} and \eqref{Teu2}.  
\begin{remark}
We can see that, the potentials of the scalar Regge-Wheeler and Zerilli equations (see \cite{Pha2020}) decay as $r^{-3}$, whence the ones of the tensorial wave equations \eqref{TensorWave} and \eqref{TensorWave1} decay as $r^{-2}$.
\end{remark}

\section{Energies of the tensorial Fackerell-Ipser field}\label{Fackerell-Ipser}
\subsection{Energy conservation law and energy fluxes}\label{ConserFac}
For a $1$-form $\xi_a \in \Lambda^1(\mathcal{B})$ on the $2$-sphere $\mathbb{S}^2_{(t,r)}$, we define
$$\xi^a = \xi_b g_{\mathbb{S}^2}^{ab}, \hbox{    } \slashed{\nabla}_{\underline L}\xi^a = \slashed{\nabla}_{\underline L}\left(\xi_b g_{\mathbb{S}^2}^{ab}\right) \hbox{  and  } \slashed{\nabla}_L \xi^a = \slashed{\nabla}_L \left(\xi_b g_{\mathbb{S}^2}^{ab}\right).$$
We define also the pointwise norms for $1$-form $\xi_a$ and $2$-tensor $\zeta_{ab}$ on $\mathbb{S}^2_{(t,r)}$ by
\begin{equation}\label{contraction}
|\xi_a|^2 = g_{\mathbb{S}^2}^{ab}\xi_a\xi_b, \, |\zeta_{ab}|^2 = g_{\mathbb{S}^2}^{ac}g_{\mathbb{S}^2}^{bd}\zeta_{ab}\zeta_{cd},
\end{equation}
where $g^{ab}_{\mathbb{S}^2}$ is the inverse of metric $g_{\mathbb{S}^2}$ on the unit sphere $\mathbb{S}^2$.

Similar to the energy momentum tensors for wave equations on scalar functions (see \cite{Ni2016,Pha2020}) and for wave equations on tensor fields (see \cite{Saka1985}), we define the one for the tensorial Fackerell-Ipser equation \eqref{Fac01}: $\slashed{\Box}_{\hat g} \phi_a + \phi_a = 0$ (we use also the forms \eqref{ReFac011} and \eqref{ReFac01} of \eqref{Fac01} to calculate) as follows
\begin{equation}\label{energytensor}
\mathbb{T}_{cd}(\phi_a) = \mathbb{T}_{(cd)} (\phi_a) = \hat{\slashed{\nabla}}_c \phi_a \hat{\slashed{\nabla}}_d \phi^a - \frac{1}{2} \hat{g}_{cd}\hat{g}^{ef}\hat{\slashed{\nabla}}_e\phi_a \hat{\slashed{\nabla}}_f\phi^a + \frac{1}{2}|\phi_a|^2\hat{g}_{cd},
\end{equation}
where $\hat{\slashed{\nabla}}_a$ denotes the projection of rescaled covariant derivative $\hat{\nabla}_a$ (which is associated with the rescaled metric $\hat{g}=\frac{1}{r^2}g$) on the unit sphere $\mathbb{S}^2$. Since $\hat{\slashed{g}}= \frac{1}{r^2}\slashed{g}= g_{\mathbb{S}^2}$, and the relations \eqref{Cor2}, \eqref{Cor4}, we have  
\begin{equation}\label{ResDeri}
\hat{\slashed{\nabla}}_{L} = \slashed{\nabla}_{\hat L}, \, \hat{\slashed{\nabla}}_{\underline{L}}=\slashed{\nabla}_{\hat{\underline L}},\,\hat{\slashed{\nabla}}_{\partial_{x^a}} = \slashed{\nabla}_{\partial_{\theta^a}}.
\end{equation}
In order to obtain the conservation law for \eqref{Fac01}, we use timelike Killing vector $T=T^c\partial_c = \partial_t$, which satisfies $\hat{\slashed{\nabla}}_{(c}T_{d)} = 0$.
For a solution ${\phi}_a$ of the tensorial Fackerell-Ipser equation \eqref{Fac01}, we have
\begin{equation}\label{Non1}
\hat{\slashed{\nabla}}^c \mathbb{T}_{cd}({\phi}_a) = \left( \slashed{\Box}_{\hat g}{\phi}_a + {\phi}_a\right)\hat{\slashed{\nabla}}_d{\phi}^a=0,
\end{equation}
where $\slashed{\Box}_{\hat g} =  \frac{r^2}{F}\slashed{\nabla}_L\slashed{\nabla}_{\underline L} - \slashed{\Delta}_{\mathbb{S}^2}$.
Setting
\begin{equation}\label{current}
J_c({\phi}_a) := T^d \mathbb{T}_{cd}({\phi}_a).
\end{equation}
From \eqref{Non1} and $\slashed{\nabla}_{(c}T_{d)} = 0$, the nonlinear energy current $J_c({\phi}_a)$ satisfies the following conservation law
\begin{equation}\label{zero1}
\hat{\slashed{\nabla}}^cJ_c({\phi}_a) = \hat{\slashed{\nabla}}_{(c}T_{d)} \mathbb{T}^{cd}({\phi}_a) = 0.
\end{equation}

Now we define the energy fluxes for tensorial Fackerell-Ipser equations \eqref{Fac01} through oriented (null or spacelike) hypersurfaces by the same way in \cite{Ni2016,Pha2020}. We follow the convention used by Penrose and Rindler \cite{PeRi84} about the Hodge dual of a $1$-form $\beta_a$ on a spacetime $(\mathscr{M}, \mathrm{g})$ (i.e., a $4-$dimensional Lorentzian manifold that is oriented and time-oriented):
\begin{equation*}
(*\beta)_{abcd} = e_{abcd}{\beta}^d,
\end{equation*}
where $e_{abcd}$ is the volume form on $(\mathscr{M},\mathrm{g})$, denoted simply by $\mathrm{dVol}_{\mathrm{g}}$. We shall use the following differential operator of the Hodge star
\begin{equation*}
\d *\beta = -\frac{1}{4}(\nabla_a\beta^a)\mathrm{dVol}^4_{\mathrm{g}}.
\end{equation*}
If $\mathcal{S}$ is the boundary of a bounded open set $\Omega$ in $\mathscr{M}$, and has outgoing
orientation, then by using Stokes theorem, we have
\begin{equation}\label{Stokesformula}
-4\int_{\mathcal{S}}*\beta = \int_{\Omega}(\nabla_a\beta^a)\mathrm{dVol}^4_{\mathrm{g}}.
\end{equation}
Now, let ${\phi}_a$ be a solution of \eqref{Fac01} with smooth and compactly supported initial data on the rescaled spacetime $(\bar{\mathcal{M}},\hat{g})$. By using \eqref{Stokesformula} and \eqref{current}, we define the rescaled energy fluxes of $\phi_a$ associated with the Killing vector field $T=\partial_t$, through an oriented (null or spacelike) hypersurface $\mathcal{S}$ in $\bar{\mathcal{M}}$ as follows (see the same formula in Equation (2.6), page 184 in \cite{Saka1985} and also \cite{Ni2016,Pha2020} for similar formulas for scalar wave equations):
\begin{equation}\label{e1}
\mathcal{E}^T_{\mathcal {S}}({\phi}_a) = -4\int_{\mathcal{S}} *J_c({\phi}_a)\d x^c = \int_{\mathcal{S}} J_c({\phi}_a){\mathcal{N}}^c{\mathcal{L}}\hook \mathrm{dVol}_{\hat g},
\end{equation}
where ${\mathcal{L}}$ is a transverse vector to $\mathcal{S}$ and ${\mathcal{N}}$ is the normal vector field to $\mathcal{S}$ such that
$\hat{g}^{ab}{\mathcal{L}}_a{\mathcal{N}}_b=1$.

We consider a domain $\Omega \subset \bar{\mathcal{M}}$ (the colored domain in Figure 2 below) which has the boundary obtained by five hypersurfaces as follows
$${\Sigma}_0 = \left\{ t=0\right\},\, \mathcal{H}^+_K = \left\{u = U_K,\, v \geq V_K \right\},\,\mathcal{I}^+_K = \left\{v=V_K,u \geq U_K \right\}$$
and
$$\mathfrak{H}^{+}_K =\mathfrak{H}^+\cap \left\{v\leq V_K \right\},\, \scri^+_K = \scri^+\cap \left\{ u\leq U_K \right\}.$$
\begin{figure}[H]
\begin{center}
\includegraphics[scale=0.9]{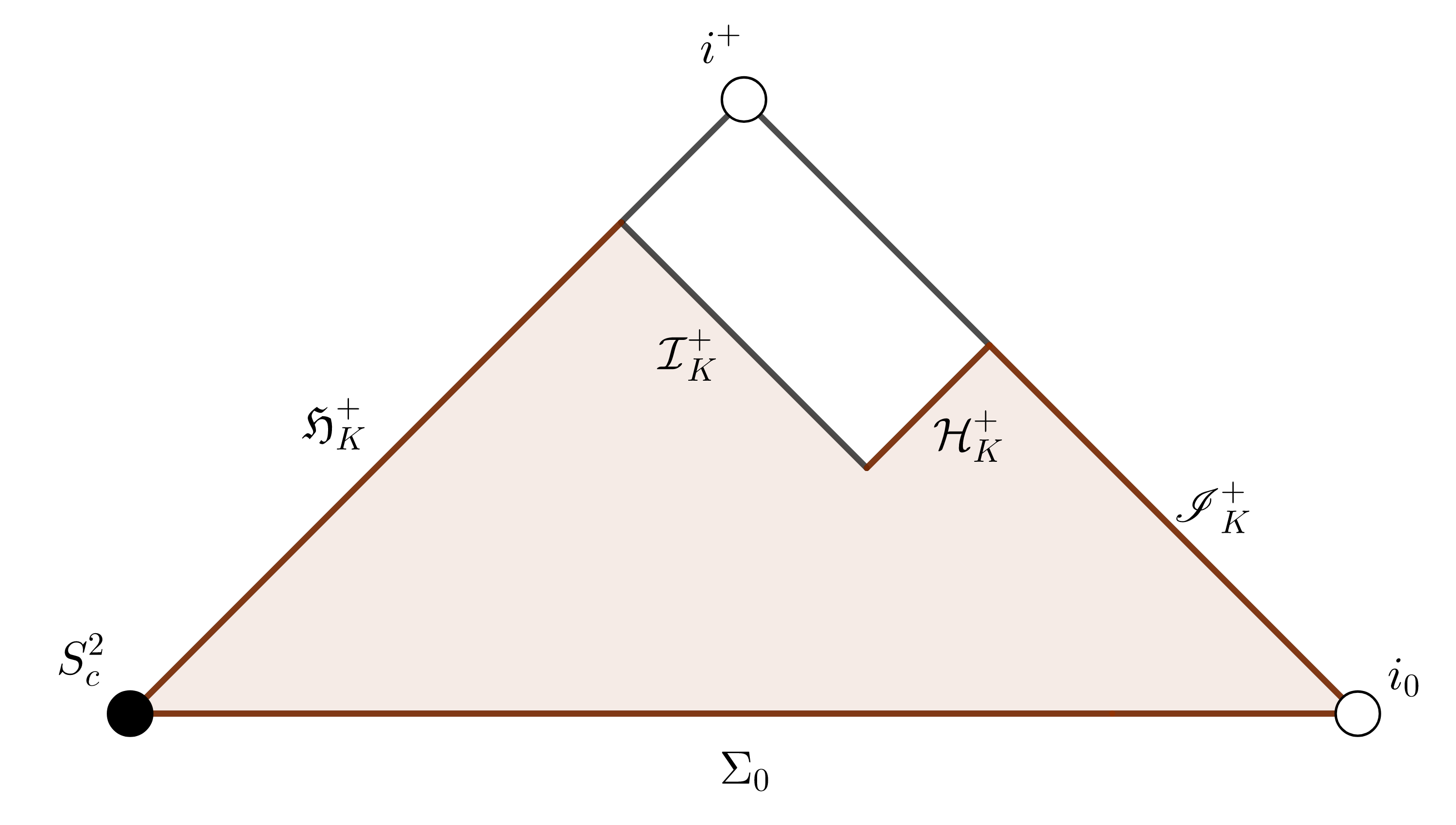}
\caption{The domain $\Omega$ in Penrose's conformal compactification $\bar{\mathcal{M}}$.}
\end{center}
\end{figure}

\begin{proposition}\label{UpToInfinity}
Consider the smooth and compactly supported initial data on $\Sigma_0$, we can define the energy fluxes of the solution $\phi_a$ of Equation \eqref{Fac01}, through the null conformal boundary $\mathfrak{H}^+\cup \scri^+$ by 
$${\mathcal{E}}^T_{\scri^+}(\phi_a) + {\mathcal{E}}^T_{\mathfrak{H}^+}(\phi_a) := \lim\limits_{U_K,V_K\rightarrow +\infty} \left( {\mathcal{E}}^T_{\scri_K^+}(\phi_a) + {\mathcal{E}}^T_{\mathfrak{H}_K^+}(\phi_a) \right).$$
Moreover, we have
$${\mathcal{E}}^T_{\scri^+}(\phi_a) + {\mathcal{E}}^T_{\mathfrak{H}^+}(\phi_a) \leq {\mathcal{E}}^T_{\Sigma_0}(\phi_a),$$
where the equality holds if and only if 
$$\lim\limits_{U_K,V_K\rightarrow +\infty}\left( {\mathcal{E}}^T_{\mathcal{H}^+_K}(\phi_a) + {\mathcal{E}}^T_{\mathcal{I}^+_K}(\phi_a) \right)= 0.$$ 
\end{proposition} 
\begin{proof} 
The proof of this proposition is similar to the one for the scalar wave equations (see \cite[Proposition 1]{Pha2020} or \cite[Section 3.2]{Ni2016}). Intergrating the conservation law \eqref{zero1} on $\Omega$ and by using the Stokes's formula \eqref{Stokesformula}, we get an exact energy identity between the hypersurfaces $\Sigma_0, \, \mathfrak{H}_K^+, \, \mathcal{H}^+_K,\, \mathcal{I}^+_K$ and $\scri_K^+$ as follows 
\begin{equation}\label{energyidentity}
{\mathcal{E}}^T_{\Sigma_0}(\phi_a)= {\mathcal{E}}^T_{\mathfrak{H}_K^+}(\phi_a) + {\mathcal{E}}^T_{\mathcal{I}^+_K}(\phi_a) + {\mathcal{E}}^T_{\mathcal{H}^+_K}(\phi_a) + {\mathcal{E}}^T_{\scri_K^+}(\phi_a).
\end{equation}
On $\Sigma_0$, we take
$$\mathcal{L}_{\Sigma_0}=\frac{r^2}{F}\partial_t, \, \mathcal{N}_{\Sigma_0}=\partial_t.$$
On $\scri^+_K$, take $\mathcal{L}_{\scri^+_K}= -\partial_R$ in the $(u,R,\omega)$ coordinates
$$\mathcal{L}_{\scri^+_K}= \frac{r^2F^{-1}}{2}l|_{\scri^+_K}.$$
On $\mathfrak{H}^+_K$, take $\mathcal{L}_{\mathfrak{H}^+_K} = \partial_R$ in the $(v,R,\omega)$ coordinates
$$\mathcal{L}_{\mathfrak{H}^+_K} = \frac{r^2F^{-1}}{2}n|_{\mathfrak{H}^+_K}.$$
Hence, we have $\mathcal{N}=\partial_t$ on both $\mathfrak{H}^+_K$ and $\scri^+_K$. This corresponds to $\mathcal{N}=\partial_v$ on $\mathfrak{H}^+_K$ and $\mathcal{N}=\partial_u$ on $\scri^+_K$. 
The transversal and normal vectors of the hypersurface $\mathcal{I}^+_K$ (resp. $\mathcal{H}_K$) can be choosen exactly as the ones of $\scri^+_K$ (resp. $\mathfrak{H}^+_K$). From this, it follows that the the energy identity given in \eqref{energyidentity} becomes
\begin{eqnarray*}
&&\int_{\mathfrak{H}_K^+}{J}_c(\phi_a)(\partial_v)^c\mathcal{L}_{\mathfrak{H}^+_K}\hook \mathrm{dVol}_{\hat g} + \int_{\scri_K^+}{J}_c(\phi_a)(\partial_u)^c\mathcal{L}_{\scri^+_K}\hook \mathrm{dVol}_{\hat g}\cr
&&+ \int_{\mathcal{I}^+_K}{J}_c(\phi_a)(\partial_u)^c\mathcal{L}_{\scri^+_K}\hook \mathrm{dVol}_{\hat g} + \int_{\mathcal{H}^+_K}{J}_c(\phi_a)(\partial_v)^c \mathcal{L}_{\mathfrak{H}^+_K}\hook \mathrm{dVol}_{\hat g} \cr
&&=\int_{\Sigma_0}{J}_c(\phi_a)(\partial_t)^a r^2F^{-1}\partial_t \hook \mathrm{dVol}_{\hat g}.
\end{eqnarray*}

Using relations \eqref{Cor2}, \eqref{Cor4}, \eqref{ResDeri} and formulas \eqref{energytensor}, \eqref{current}, we can calculate the energy fluxes through $\Sigma_0$, $\mathcal{I}^+_K$, $\mathcal{H}^+_K$, $\scri^+_K$ and $\mathfrak{H}^+_K$ as follows (see the same calculations for scalar wave equations in \cite{Ni2016,Pha2020}):
\begin{equation}\label{energyFlux1}
{\mathcal{E}}^T_{\Sigma_0}(\phi_a) = \frac{1}{2}\int_{{\Sigma}_0}\left( |\slashed{\nabla}_{\partial_{t}}{\phi}_a|^2 + |\slashed{\nabla}_{\partial_{r_*}}{\phi}_a|^2 + R^2F|\slashed{\nabla}_{\mathbb{S}^2}{\phi}_a|^2 + R^2F|{\phi}_a|^2 \right) \d r_*\d \mathbb{S}^2,
\end{equation}
\begin{equation}\label{energyFlux2}
{\mathcal{E}}^T_{\mathcal{I}^+_K}(\phi_a) = \int_{\mathcal{I}_K^+} \left( |\slashed{\nabla}_{{\underline L}}{\phi_a}|^2 + R^2F|\slashed{\nabla}_{\mathbb{S}^2}{\phi_a}|^2 + R^2F|{\phi_a}|^2 \right) \d u \d \mathbb{S}^2,
\end{equation}
\begin{equation}\label{energyFlux3}
\mathcal{E}^T_{\mathcal{H}^+_K}(\phi_a) = \int_{\mathcal{H}^+_K} \left( |\slashed{\nabla}_{L} {\phi_a}|^2 + R^2F|\slashed{\nabla}_{\mathbb{S}^2}{\phi_a}|^2 + R^2F|{\phi_a}|^2 \right) \d v \d \mathbb{S}^2,\end{equation}
where $|\slashed{\nabla}_{\mathbb{S}^2}{\phi_a}|^2 = |\slashed{\nabla}_{\partial_\theta}\phi_a|^2 + \frac{1}{\sin^2\theta}|\slashed{\nabla}_{\partial_{\varphi}}\phi_a|^2$, and
\begin{equation}\label{energyFlux4}
{\mathcal{E}}^T_{\scri^+_K}(\phi_a) = \int_{\scri_K^+} |(\slashed{\nabla}_{{\underline L}}{\phi_a})|_{\scri^+}|^2 \d u \d \mathbb{S}^2,
\end{equation}
\begin{equation}\label{energyFlux5}
{\mathcal{E}}^T_{\mathfrak{H}^+_K}(\phi_a) = \int_{\mathfrak{H}^+_K} |(\slashed{\nabla}_{L}{\phi_a})|_{\mathfrak{H}^+}|^2 \d v \d \mathbb{S}^2.
\end{equation}

We observe that the energy fluxes across $\scri_K^+$ and $\mathfrak{H}_K^+$ are non negative increasing functions of $U_K,\, V_K$ and their sum is bounded by ${\mathcal{E}}^T_{\Sigma_0}(\phi_a)$ by the energy indentity \eqref{energyidentity}. This can be deduced from the energy identity \eqref{energyidentity} and the positivity of ${\mathcal{E}}^T_{\mathcal{I}^+_K}(\phi_a) + {\mathcal{E}}^T_{\mathcal{H}^+_K}(\hat\psi)$. Therefore, the limit of ${\mathcal{E}}^T_{\scri^+_K}(\phi_a)+{\mathcal{E}}^T_{\mathfrak{H}^+_K}(\phi_a)$ exists and the following sum is well defined
\begin{eqnarray}\label{limitenergy}
{\mathcal{E}}^T_{\scri^+}(\phi_a) + {\mathcal{E}}^T_{\mathfrak{H}^+}(\phi_a) &=& \lim_{U_K,V_K\rightarrow+\infty} \left( {\mathcal{E}}^T_{\scri_K^+}(\phi_a) + {\mathcal{E}}^T_{\mathfrak{H}_K^+}(\phi_a) \right)\cr
&=& {\mathcal{E}}^T_{\Sigma_0}(\phi_a) - \lim_{U_K,V_K\rightarrow +\infty} \left( {\mathcal{E}}^T_{\mathcal{I}^+_K}(\phi_a) + {\mathcal{E}}^T_{\mathcal{H}^+_K}(\phi_a)\right).
\end{eqnarray}
The proposition now holds from the above identity.
\end{proof}

\subsection{Tensorial field space of initial data and Cauchy problem}\label{FacSpace}
We define the finite energy space of tensorial fields on the spacelike hypersurface $\Sigma_\tau = \left\{ t=\tau \right\}$ as follows
\begin{definition}\label{FunSpaFac}
We define $\mathcal{H}(\Lambda^1(\mathbb{S}^2)|_{\Sigma_\tau})$ which is the completion of ${C}_0^\infty (\Lambda^1(\mathbb{S}^2)|_{\Sigma_\tau})\times {C}_0^\infty (\Lambda^1(\mathbb{S}^2)|_{\Sigma_\tau})$ in the norm
\begin{equation}\label{ENERGY}
\left\|({\xi}_{a},{\zeta}_{a}) \right\|_{\mathcal{H}(\Lambda^1(\mathbb{S}^2)|_{\Sigma_\tau})} = \frac{1}{\sqrt 2}\left( \int_{\Sigma_\tau}\left( |{\zeta}_{a}|^2 + |\slashed{\nabla}_{\partial_{r_*}}{\xi}_{a}|^2 + R^2F|\slashed{\nabla}_{\mathbb{S}^2}{\xi}_{a}|^2 + R^2F|{\xi}_{a}|^2 \right) \d r_*\d \mathbb{S}^2 \right)^{1/2}.
\end{equation}
\end{definition}
In order to state and prove the well-posedness of Cauchy problem, we need the following definition of the Sobolev spaces for tensorial fields that are defined on open sets (see \cite[Definition 2]{Ni2016} for the original definition of scalar fields):
\begin{definition}
Let $s \in [0, +\infty)$, a tensorial field $u_a$ on $\Lambda^1(\mathbb{S}^2)|_{\mathcal{M}}$ is said to belong to $H^s_{loc}(\Lambda^1(\mathbb{S}^2)|_{\bar{\mathcal{M}})}$ if for any
local chart $(\Omega, \zeta)$, such that $\Omega \subset \mathcal{M}$ is an open set with smooth compact boundary in $\bar{\mathcal{M}}$ (note that this excludes neighbourhoods of either $i^\pm$ or $i^0$ but allows open sets whose boundary contains parts of the conformal boundary) and $\zeta$ is a smooth diffeomorphism from $\Omega$ onto a bounded open set $U \subset \mathbb{R}^4$ with smooth compact boundary, we have $u_a \circ \zeta^{-1} \in H^s(\Lambda^1(\mathbb{S}^2)|_U)$.
\end{definition}
To define the trace operator in Subsection \ref{TraceInj}, we need to prove the well-posedness of Cauchy problem for Equation \eqref{Fac01} in the conformal rescaled spacetime $(\bar{\mathcal{M}},\hat{g})$:
\begin{theorem}\label{Cauchyproblem}(Cauchy problem of Equation \eqref{Fac01} in $(\bar{\mathcal{M}},\hat{g})$).
The Cauchy problem of Equation \eqref{Fac01} is well-posed in ${C}(\mathbb{R}_t,\,\mathcal{H})$, where $\mathcal{H}= \cup_{t\in \mathbb{R}}\mathcal{H}(\Lambda^1(\mathbb{S}^2)|_{\Sigma_t})$. This means that for any $(\xi_{a},\zeta_{a}) \in \mathcal{H}(\Lambda^1(\mathbb{S}^2)|_{\Sigma_0})$, there exists a unique solution $\phi_a \in \mathcal{D}'(\Lambda^1(\mathbb{S}^2)|_{\bar{\mathcal{M}}})$ of Equation\eqref{Fac01} such that
\begin{equation}\label{CauchyCD}
(\phi_a,\slashed{\nabla}_{\partial_t}\phi_a) \in {C}(\mathbb{R}_t;\, \mathcal{H}): \, \phi_a|_{\Sigma_0}={\xi}_{a}; \, \slashed{\nabla}_{\partial_t}\phi_a|_{\Sigma_0} ={\zeta}_{a}.
\end{equation}
Moreover, $\phi_a$ belongs to $H^1_{loc}(\Lambda^1(\mathbb{S}^2)|_{\bar{{\mathcal M}}})$.
The same assertion holds for Equation \eqref{Fac02}.
\end{theorem}
\begin{proof}
We prove the well-posedness in the future domain $\mathcal{I}^+(\Sigma_0)=\left\{t\geq 0\right\}$ of $\bar{\mathcal{M}}$, the well-posedness in the past domain $\mathcal{I}^-(\Sigma_0)=\left\{ t\leq 0\right\}$ is done similarly.
The proof is done by using the same methods in \cite{Saka1985} (see also \cite{Bruhat1979}) which are based on Leray's theorem energy estimates for symmetrical hyperbolic system in smooth globally hyperbolic spacetime. In fact, the work of Saka (see Theorem 2 in \cite{Saka1985}) established the well-posedness in finite energy spaces for tensorial wave equations on smooth globally hyperbolic spacetimes. In the rest of this proof, we show how the method in \cite{Saka1985} can be extended to prove Theorem \ref{Cauchyproblem}. 

First, by projecting the tensorial Fackerell-Ipser equation \eqref{Fac01} on the basic frame $(\slashed{\nabla}_{\partial_{\theta^a}},\slashed{\nabla}_{\partial_{\theta^b}})$ of $\Lambda^1(\mathbb{S}^2)$, we get
\begin{equation}\label{system-wave1}
P_{\hat g}\Phi + L_1\Phi = 0,
\end{equation}
with 
$$P_{\hat g}= \left(\begin{matrix}
\Box_{\hat{g}} &&0\\
0&&\Box_{\hat{g}}
\end{matrix} \right)\hbox{      } (\hbox{here   }\Box_{\hat{g}} = \frac{r^2}{F}L\underline{L} - {\Delta}_{\mathbb{S}^2})$$
is a diagonal matrix,
$${\Phi} = \left( \begin{matrix} {\phi}_1\\ {\phi}_2\end{matrix} \right),$$
where $\phi_i \, (i=1,2)$ are the scalar components of $\phi_a$ decomposed on the basic frame $(\slashed{\nabla}_{\partial_{\theta^a}},\slashed{\nabla}_{\partial_{\theta^b}})$ of $\Lambda^1(\mathbb{S}^2)$ and 
$$L_1 = \left( \begin{matrix} L_1^{11}&&L_1^{12}\\ L_1^{21}&& L_1^{22} \end{matrix}  \right)$$
is a $2\times 2$-matrix, where $L_1^{11}=L_1^{22},\, L_1^{12}=L_1^{21}$ and $L_1^{ij}$ are first order differential operators with smooth coefficients
$$L_1^{ij} = b_0^{ij}\partial_t + b_1^{ij}\partial_x + c^{ij}.$$

To avoid the singularities, we cut off $\mathcal{I}^+(\Sigma_0)$ by $\mathcal{O}$ which is a union of far enough neighborhoods of $i^+$ and $i_0$ (see \cite[Theorem 1]{Pham2022} and also \cite[Section 4.2]{Ni2012}). Note that, we can do this and do not change the domain of the well-posedness of Cauchy problem because $\mathcal{H}(\Lambda^1(\mathbb{S}^2)|_{\Sigma_0})$
is the completion of tensorial fields which have smooth and compact supports and the energy fluxes of smooth solutions (if the existence holds) through the cut-off null hypersurfaces $\mathcal{H}^+_K,\, \mathcal{I}^+_K$ will tend to $0$ as $U_K,V_K$ tend to infinity (see Theorem \ref{EqualityFac} below).

\begin{figure}[H]
\begin{center}
\includegraphics[scale=0.9]{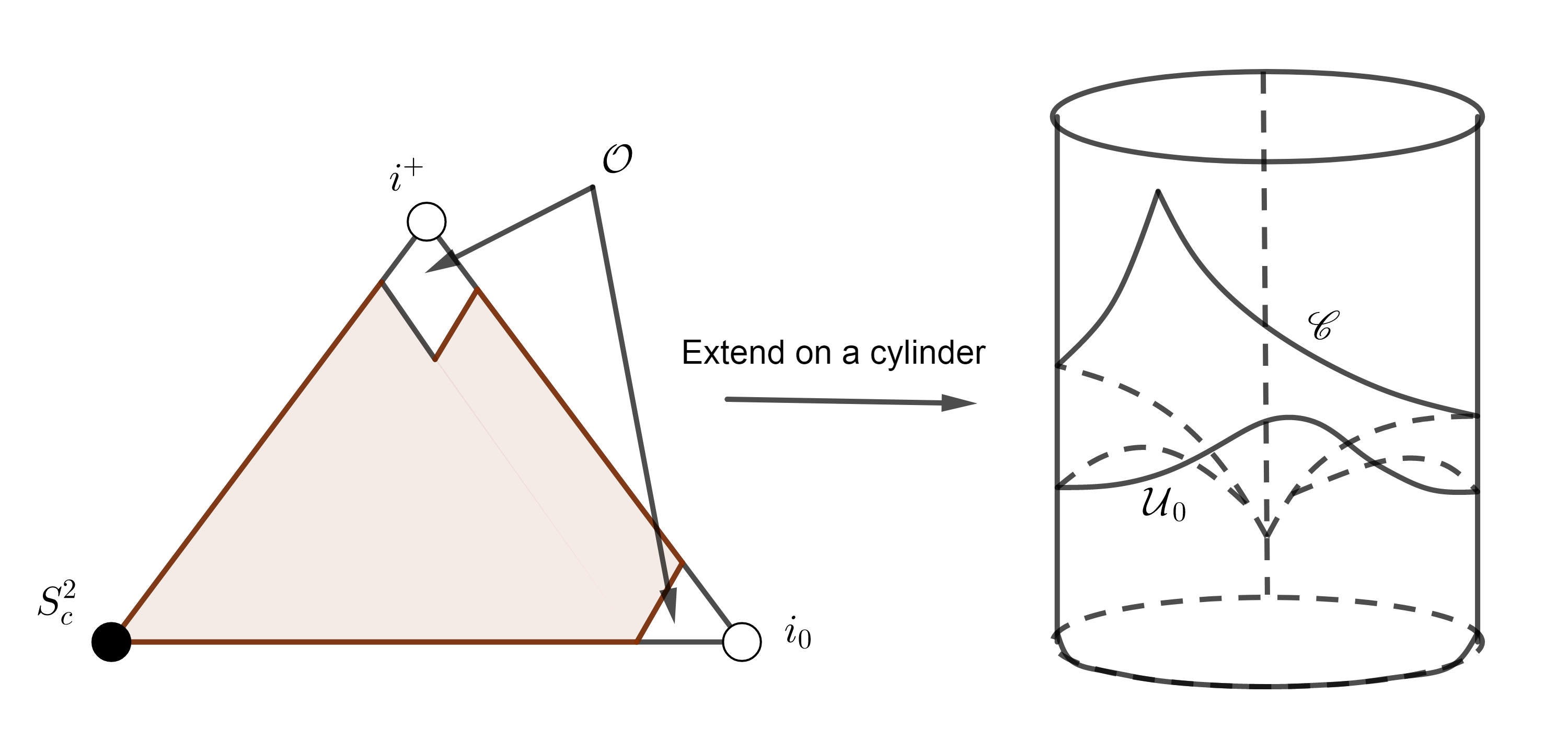}
\caption{The extension of $\mathcal{I}^+(\Sigma_0)-\mathcal{O}$ onto a global hyperbolic cylindre}
\end{center}
\end{figure}

In order to use the method in \cite{Saka1985}, we extend $(\mathcal{I}^+(\Sigma_0)-\mathcal{O},\hat{g})$ onto a cylindrical globally hyperbolic spacetime $(\mathbb{R}_t\times \mathbb{S}^3, \mathfrak{g})$. where $\mathfrak{g} = \d t^2 - h$ with $h(t)$ is a Riemannian metric on $\mathbb{S}^3$ smoothly with respect to $t$.  For each $t\geq 0$, the hypersurface $\Sigma_t-\mathcal{O}$ is extended inside $(\mathbb{R}_t\times \mathbb{S}^3, \mathfrak{g})$ as a spacelike hypersurface $\mathcal{U}_t$ and we obtain a spacelike foliation $\left\{ \mathcal{U}_t \right\}_{t\geq 0}$.
The conformal boundary $(\mathfrak{H}^+\cup \scri^+)-\mathcal{O}$ is extended inside $(\mathbb{R}_t\times \mathbb{S}^3, \mathfrak{g})$ as a null hypersurface $\mathscr{C}$, that is the graph of a Lipschitz function over $\mathbb{S}^3$ and the initial data by zero on the rest of the extended hypersurface $\mathscr{C}-(\mathfrak{H}^+\cup \scri^+)$. The initial data $(\xi_a,\zeta_a)$ is extended to $(\widetilde{\xi}_a,\widetilde{\zeta}_a)$ which vanishes on $\Sigma_0-\mathcal{U}_0$.
In the extending spacetime $(\mathbb{R}_t\times \mathbb{S}^3, \mathfrak{g})$, Equation \eqref{system-wave1} becomes
\begin{equation}\label{systemwave1}
P_{\mathfrak{g}}({\Phi}) + L_1(\Phi) = 0,
\end{equation}
where
$$P_{\mathfrak{g}}= \left(\begin{matrix}
\Box_{\mathfrak{g}} &&0\\
0&&\Box_{\mathfrak{g}}
\end{matrix} \right),\, \Box_{\mathfrak{g}} = \partial_t^2 - \Delta_h.$$
Equation \eqref{system-wave1} is equivalent to a symmetrical hyperbolic system which consists \eqref{EQ1} and \eqref{EQ2} in Appendix \ref{appendix}.

The well-posedness of Cauchy problem in Theorem \ref{Cauchyproblem} is extended to the one for Equation \eqref{system-wave1} which consists equations \eqref{EQ1}, \eqref{EQ2} with the initial data in $\mathcal{H}(\mathcal{U}_0)$.
Decomposing on the basic frame of $\Lambda^1(\mathcal{U}_0)$, we get the scalar form of $\widetilde{\xi}_a$ as $\widetilde{\xi}=(\widetilde{\xi}_1,\widetilde{\xi}_2)$ and of $\widetilde{\zeta}_a$ as $\widetilde{\zeta}=(\widetilde{\zeta}_1,\widetilde{\zeta}_2)$.
By using Leray's theorem, for smooth intitial data on $(\widetilde{\xi}_i,\widetilde{\zeta}_i)\in C^\infty(\mathcal{U}_0)\times C^\infty(\mathcal{U}_0)\, (i=1,2)$, equations \eqref{EQ1} and \eqref{EQ2} have a unique smooth solution $\widetilde{\Phi}=(\widetilde{\phi}_1,\widetilde{\phi}_2)$ in smooth globally hyperbolic spacetime $(\mathbb{R}_t\times \mathbb{S}^3, \mathfrak{g})$. 
For the initial data $(\widetilde{\xi}_i,\widetilde{\zeta}_i) \, (i=1,2)$ in $\mathcal{H}(\mathcal{U}_0)$, there exists the $C_0^\infty(\mathcal{U}_0)$ sequences $\left\{\widetilde{\xi}_i^n\right\}_{n\in \mathbb{N}},\, \left\{\widetilde{\zeta}_i^n\right\}_{n\in \mathbb{N}}\, (i=1,2)$
which converge to $\xi_i$ and $\zeta_i\, (i=1,2)$ under $H^1$-norm and $L^2$-norm, respectively (see the definition of $H^1$-norm and $L^2$-norm in Appendix \ref{appendix}). For each smooth initial data $(\widetilde{\xi}_i^n,\widetilde{\zeta}_i^n)\, (i=1,2)$, there is a unique smooth solution $\widetilde{\Phi}^n=(\widetilde{\phi}^n_1,\widetilde{\phi}^n_2)$ of Cauchy problem for equations \eqref{EQ1} and \eqref{EQ2}. By using energy estimate \eqref{ENG} in Appendix \ref{appendix}, we can show that $\left\{ (\widetilde{\Phi}^n,\partial_t\widetilde{\Phi}^n)\right\}_{n\in\mathbb{N}}$ is a Cauchy sequence in $C([0,T],\cup_{0\leq t\leq T}\mathcal{H}(\mathcal{U}_t))$, hence $(\widetilde{\Phi}^n,\partial_t\widetilde{\Phi}^n)$ converges to $(\widetilde{\Phi},\partial_t\widetilde{\Phi})\in C([0,T],\cup_{0\leq t \leq T}\mathcal{H}(\mathcal{U}_t))$ (in details, see the proof of \cite[Theorem 2]{Saka1985}). Clearly, $\widetilde{\Phi}$ is the local solution of Cauchy problem of equations \eqref{EQ1} and \eqref{EQ2} with $\widetilde{\phi}_i|_{t=0}=\widetilde{\xi}_i; \, \partial_t\widetilde{\phi}_i|_{t=0} = \widetilde{\zeta}_i$ for $i=1,2$. Using the local well-posedness result and energy estimate \eqref{ENG}, we can establish the global well-posedness of Cauchy problem of Equation \eqref{systemwave} in $C(\mathbb{R}_t,\cup_{t \geq 0}\mathcal{H}(\mathcal{U}_t))$ by the same methods in \cite[Theorem 2]{Cagnac} and \cite[Theorem 1]{Dossa}. 
Therefore, we obtain the global tensorial field solution $\widetilde{\phi}_a$ of the extended equation of \eqref{Fac01} in $C(\mathbb{R}_t,\cup_{t \geq 0}\mathcal{H}(\Lambda^1(\mathbb{S}^2)|_{\mathcal{U}_t}))$.

By local uniqueness and causality, using in particular the fact that as a consequence of the
finite propagation speed, the global solution $\phi_a$ of Cauchy problem of Equation \eqref{Fac01} is the restriction of $\widetilde{\phi}_a$ on $\mathcal{I}^+(\Sigma_0)-\mathcal{O}$ and it satifies Equation \eqref{CauchyCD}. 
\end{proof}

\subsection{Energy identity up to $i^+$ and trace operator}\label{TraceInj}
In this section, we will show that $\lim\limits_{U_K,V_K\rightarrow +\infty} \left( \mathcal{E}^T_{\mathcal{H}^+_{K}}(\phi_a) + \mathcal{E}^T_{\mathcal{I}^+_{K}}(\phi_a)\right) = 0$ and then we can obtain the energy equality
$$\mathcal{E}^T_{\Sigma_0}(\phi_a) = \mathcal{E}^T_{\mathfrak{H}^+}(\phi_a) + \mathcal{E}^T_{\scri^+}(\phi_a).$$
We recall the following energy decay of the tensorial field $\phi_a$ which satisfies Equation \eqref{Fac01} (see \cite[Lemma 5.8]{Pa2019}):
\begin{lemma}\label{DecayEnergy}
There exists a positive number $R_*$ such
that the following holds: let $\phi_a$ be a smooth solution to the tensorial fackerell–Ipser
equation \eqref{Fac01}  on $\left\{ u = u_0\right\} \cap \left\{ v = v_0 \right\}$. Let $U_K\geq u_0$ and let $P(U_K)$ be $(U_K,R_*+U_K)$. We have the decay of the flux
\begin{equation*}
\mathcal{F}^\infty[\phi_a](P(U_K)) \leq CU_K^{-2},
\end{equation*}
where $C$ is a positive constant depending on $\phi_a|_{u=u_0,\, v=v_0}$ and
\begin{equation*}
\mathcal{F}^\infty[\phi_a](P(U_K)) = \mathcal{F}_u^T[\phi_a](R_*+U_K,+\infty) + \mathcal{F}^N_{R_*+U_K}[\phi_a](U_K,+\infty),
\end{equation*}
with\footnote{in \cite{Pa2019}, the energy fluxes \eqref{Oen2r1} and \eqref{Oener2} use the notation $1-\mu$ for $F$ and $V$ for $R^2F$.}
\begin{equation}\label{Oen2r1}
\mathcal{F}_{U_K}^T[\phi_a](R_*+U_K,+\infty) = \int_{R_*+U_K}^{+\infty}\int_{\mathbb{S}^2} \left( |\slashed{\nabla}_{L} {\phi_a}|^2 + F|\slashed{\nabla}{\phi_a}|^2 + R^2F|{\phi_a}|^2 \right) \d v \d \mathbb{S}^2,
\end{equation}
\begin{equation}\label{Oener2}
\mathcal{F}_{R_*+U_K}^N[\phi_a](U_K,+\infty) = \int_{U_K}^{+\infty}\int_{\mathbb{S}^2}\left( F^{-1}|\slashed{\nabla}_{{\underline L}}{\phi_a}|^2 + F|\slashed{\nabla}{\phi_a}|^2 + R^2F|{\phi_a}|^2 \right) \d u \d \mathbb{S}^2,
\end{equation}
where the pointwise norms $|\cdot|$ in \eqref{Oen2r1} and \eqref{Oener2} are given as in \eqref{contraction}.
\end{lemma}

The above lemma helps us to obtain the energy decay for solution of the Fackerell-Ipser equation \eqref{Fac01} through null hypersurfaces $\mathcal{H}^+_K$ and $\mathcal{I}^+_K$.

\begin{theorem}\label{ResultDecay}(Energy Decay)
Let $\phi_a$ be a smooth tensorial solution of the tensorial Fackerell-Ipser equation \eqref{Fac01} on $\left\{ u \geq u_0\right\} \cap \left\{ v\geq v_0 \right\}$. 
There exist a positive number $R_*$ and a positive constant $C$ depending on the value of $\phi_a$ on $\left\{ u = u_0\right\} \cap \left\{ v = v_0 \right\}$ such that:
we have the following decay energy for the original field $\phi_a$ on $\mathcal{H}^+_K \cup \mathcal{I}^+_K$ with $U_K\geq u_0$ and $V_K= U_K +2R_*$:
\begin{equation}\label{decay}
\mathcal{E}^T_{\mathcal{H}^+_K}(\phi_a) + \mathcal{E}^T_{\mathcal{I}^+_K}(\phi_a) \leq C U_K^{-2}.
\end{equation}
 
\end{theorem}
\begin{proof}
Since $r\slashed{\nabla}_{\partial_{x^a}}=\slashed{\nabla}_{\partial_{\theta^a}}$, we have $|\slashed{\nabla}\phi_a|^2 = R^2|\slashed{\nabla}_{\mathbb{S}^2}\phi_a|^2$. Therefore, we obtain 
\begin{eqnarray*}
\mathcal{F}_{U_K}^T[\phi_a](R_*+U_K,+\infty)&=&\int_{R_*+U_K}^{+\infty}\int_{\mathbb{S}^2} \left( |\slashed{\nabla}_{L} {\phi_a}|^2 + F|\slashed{\nabla}{\phi_a}|^2 + R^2F|{\phi_a}|^2 \right) \d v \d \mathbb{S}^2\cr
&=& \int_{\mathcal{H}_K} \left( |\slashed{\nabla}_{L} {\phi_a}|^2 + R^2F|\slashed{\nabla}_{\mathbb{S}^2}{\phi_a}|^2 + R^2F|{\phi_a}|^2 \right) \d v \d \mathbb{S}^2\cr
&=& \mathcal{E}^T_{\mathcal{H}^+_K}(\phi_a),
\end{eqnarray*}
and
\begin{eqnarray*}
\mathcal{F}_{R_*+U_K}^N[\phi_a](U_K,+\infty) &=& \int_{U_K}^{+\infty}\int_{\mathbb{S}^2}\left( F^{-1}|\slashed{\nabla}_{{\underline L}}{\phi_a}|^2 + F|\slashed{\nabla}{\phi_a}|^2 + R^2F|{\phi_a}|^2 \right) \d u \d \mathbb{S}^2\cr
&=& \int_{\mathcal{I}_K} \left( (1-2MR)^{-1}|\slashed{\nabla}_{{\underline L}}{\phi_a}|^2 + R^2F|\slashed{\nabla}_{\mathbb{S}^2}{\phi_a}|^2 + R^2F|{\phi_a}|^2 \right) \d u \d \mathbb{S}^2\cr
&\geq&  \int_{\mathcal{I}_K} \left( |\slashed{\nabla}_{{\underline L}}{\phi_a}|^2 + R^2F|\slashed{\nabla}_{\mathbb{S}^2}{\phi_a}|^2 + R^2F|{\phi_a}|^2 \right) \d u \d \mathbb{S}^2\cr
&=& {\mathcal{E}}^T_{\mathcal{I}^+_K}(\phi_a).
\end{eqnarray*}
Combining these inequalities with Lemma \ref{DecayEnergy}, we obtain 
\begin{equation*}
\mathcal{E}^T_{\mathcal{H}^+_K}(\phi_a) + \mathcal{E}^T_{\mathcal{I}^+_K}(\phi_a)\leq \mathcal{F}_{U_K}^T[\phi_a](R_*+U_K,+\infty) + \mathcal{F}_{R_*+U_K}^N[\phi_a](U_K,+\infty) \leq C U_K^{-2}.
\end{equation*}
\end{proof}

Now, we state and prove the energy equality between the energy fluxes of $\phi_a$ through the Cauchy hypersurface $\Sigma_0$ and the one through conformal boundary $\mathfrak{H}^+\cup \scri^+$ in the following theorem.
\begin{theorem}\label{EqualityFac}
Let ${\phi}_a$ be a smooth tensorial solution of the rescaled tensorial Fackerell-Ipser equation \eqref{Fac01} on $\left\{ u \geq u_0\right\} \cap \left\{ v\geq v_0 \right\}$. The energies of ${\phi}_a$ through the null hypersurfaces $\mathcal{H}_K$ and $\mathcal{I}_K$ tend to zero as $U_K,\, V_K$ tend to infinity, i.e.,
\begin{equation}\label{lim1}
\lim_{U_K,V_K\rightarrow +\infty} \left( {\mathcal{E}}^T_{\mathcal{H}^+_K}(\phi_a) + {\mathcal{E}}^T_{\mathcal{I}^+_K}(\phi_a)\right) = 0.
\end{equation}
As a consequence, we have the energy equality between the energy flux of $\phi_a$ (resp. $\underline{\phi}_a$) through $\Sigma_0$ and the ones through the conformal boundary $\mathfrak{H}^+\cup \scri^+$, i.e., the energy identity up to $i^+$, as follows
\begin{equation}\label{equality}
\mathcal{E}^T_{\Sigma_0}(\phi_a) = \mathcal{E}^T_{\mathfrak{H}^+}(\phi_a) + \mathcal{E}^T_{\scri^+}(\phi_a).
\end{equation}
The same energy identity up to $i^-$ holds.
\end{theorem}
\begin{proof}
Note that for a fixed $R_*$, we have $V_K=R_*+U_K$ tends to $+\infty$ as $U_K$ tends to $+\infty$. The convergence \eqref{lim1} is valid by the energy decay \eqref{decay} obtained in Theorem \ref{ResultDecay}. Combining this convergence and Proposition \ref{UpToInfinity}, we obtain the energy identity up to $i^+$ such as equality \eqref{equality}. 
\end{proof}
The well-posedness of Cauchy problem obtained in Theorem \ref{Cauchyproblem} allows us to define the trace operator on the conformal boundary (note that, in the proof of Theorem \ref{Cauchyproblem}, we obtained the well-posedness of Cauchy problem for both the smooth initial data and the initial data in the finite energy space on $\Sigma_0$).
\begin{definition}(Trace operator for tensorial Fackerell-Ipser equation) Let $({\xi}_{a}, {\zeta}_{a}) \in {C}_0^\infty(\Lambda^1(\mathbb{S}^2)|_{\Sigma_0})\times {C}_0^\infty(\Lambda^1(\mathbb{S}^2)|_{\Sigma_0})$. Consider the solution of Equation \eqref{Fac01}, and let ${\phi}_a \in {C}^\infty(\Lambda^1(\mathbb{S}^2)|_{\bar{\mathcal{M}}})$  such that
$${\phi}_a|_{\Sigma_0} = {\xi}_{a}, \, \slashed{\nabla}_t {\phi}_{a}|_{\Sigma_0} = {\zeta}_{a}.$$
We define the trace operator $\mathcal{T}^+$ from ${C}_0^\infty(\Lambda^1(\mathbb{S}^2)|_{\Sigma_0})\times {C}_0^\infty(\Lambda^1(\mathbb{S}^2)|_{\Sigma_0})$ to ${C}^\infty(\Lambda^1(\mathbb{S}^2)|_{\mathfrak{H}^+})\times {C}_0^\infty(\Lambda^1(\mathbb{S}^2)|_{\scri^+})$ by
$$\mathcal{T}^+({\xi}_{a}, {\zeta}_{a}) = ({\phi}_a|_{\mathfrak{H}^+}, {\phi}_a|_{\scri^+}).$$
The trace operator for solution $\underline{\phi}_a$ of Equation \eqref{Fac02} is defined by the same way.
\end{definition}
We can extend the tensorial field space for scattering data of Equation \eqref{Fac01} by density as in the following definition:
\begin{definition}\label{ScatteringDataFac}
The tensorial field space for scattering data $\mathcal{H}^+$ is the completion of ${C}_0^\infty(\Lambda^1(\mathbb{S}^2)|_{\mathfrak{H}^+}) \times {C}_0^\infty(\Lambda^1(\mathbb{S}^2)|_{\scri^+})$ in the norm
$$\left\| (\xi_a,\zeta_a) \right\|_{\mathcal{H}^+} = \frac{1}{\sqrt 2}\left(\int_{\mathfrak{H}^+} |\slashed{\nabla}_L \xi|^2 \d v\d \mathbb{S}^2 + \int_{\scri^+} |\slashed{\nabla}_{\underline{L}} \zeta|^2 \d u \d \mathbb{S}^2 \right)^{1/2}.$$
This means that
$$\mathcal{H}^+ \simeq \dot{H}^1(\mathbb{R}_v; \, L^2(\Lambda^1(\mathbb{S}^2)|_{\mathfrak{H}^+})) \times \dot{H}^1(\mathbb{R}_u; \, L^2(\Lambda^1(\mathbb{S}^2)|_{\scri^+})).$$
\end{definition}
As a direct consequence of the energy equality \eqref{equality} and the well-posedness of Cauchy problem in the finite energy space in Theorem \ref{Cauchyproblem}, we have the following theorem.
\begin{theorem}\label{Trace}
The trace operator of solution $\phi_a$ of Equation \eqref{Fac01} extends uniquely as a bounded linear map from $\mathcal{H}$ to $\mathcal{H}^+$. The extended operator is a partial isometry, i.e., an injective operator. This means that for any $({\xi}_{a},{\zeta}_{a}) \in \mathcal{H}(\Lambda^1(\mathbb{S}^2)|_{\Sigma_0})$,
$$\left\| \mathcal{T}^+({\xi}_{a},{\zeta}_{a}) \right\|_{\mathcal{H}^+} = \left\| ({\xi}_{a},{\zeta}_{a}) \right\|_{\mathcal{H}(\Lambda^1(\mathbb{S}^2)|_{\Sigma_0})}.$$
The same assertion holds for the trace operator of solution $\underline{\phi}_a$ of Equation \eqref{Fac02}.
\end{theorem}

\section{Conformal scattering for the tensorial Fackerell-Ipser equations}\label{ConFac}
\subsection{Generalization of L. H\"ormander's result for tensorial wave equations}\label{GeneHorm}
To construct the conformal scattering operator, we need to show that the trace operator is surjective. This corresponds to prove the well-posedness of the Goursat problem for the rescaled equation \eqref{Fac01} with the initial data on the conformal boundaries $\mathfrak{H}^+\cup \scri^+$ (resp. $\mathfrak{H}^-\cup \scri^-$) in Penrose's conformal compactification $\bar{\mathcal{M}}$.

H\"ormander \cite{Ho1990} proved the well-posedness of the Goursat problem for the second-order scalar wave equations with regular first-order potentials in the spatially compact spacetime. Nicolas \cite{Ni2006} extended the results of H\"ormander with very slightly regular metric and potential, precisely a ${C}^1$-metric and potential with continuous coefficients of the first-order terms and locally $L^\infty$ coefficients for the terms of order $0$.
Here, we will prove the well-posedness of the Goursat problem for the tensorial wave equations with regular first-order tensorial potentials. More precisely, we will show how we can apply the results of H\"ormander for the tensorial wave equation \eqref{Fac01} (or \eqref{Fac02}) with the smooth compactly supported initial data on the conformal boundary, i.e, $(\xi_a, \zeta_a) \in {C}_0^\infty(\Lambda^1(\mathbb{S}^2)|_{\mathfrak{H}^+}) \times {C}_0^\infty(\Lambda^1(\mathbb{S}^2)|_{\scri^+})$ in Schwarzschild background.
\begin{figure}[H]
\begin{center}
\includegraphics[scale=0.8]{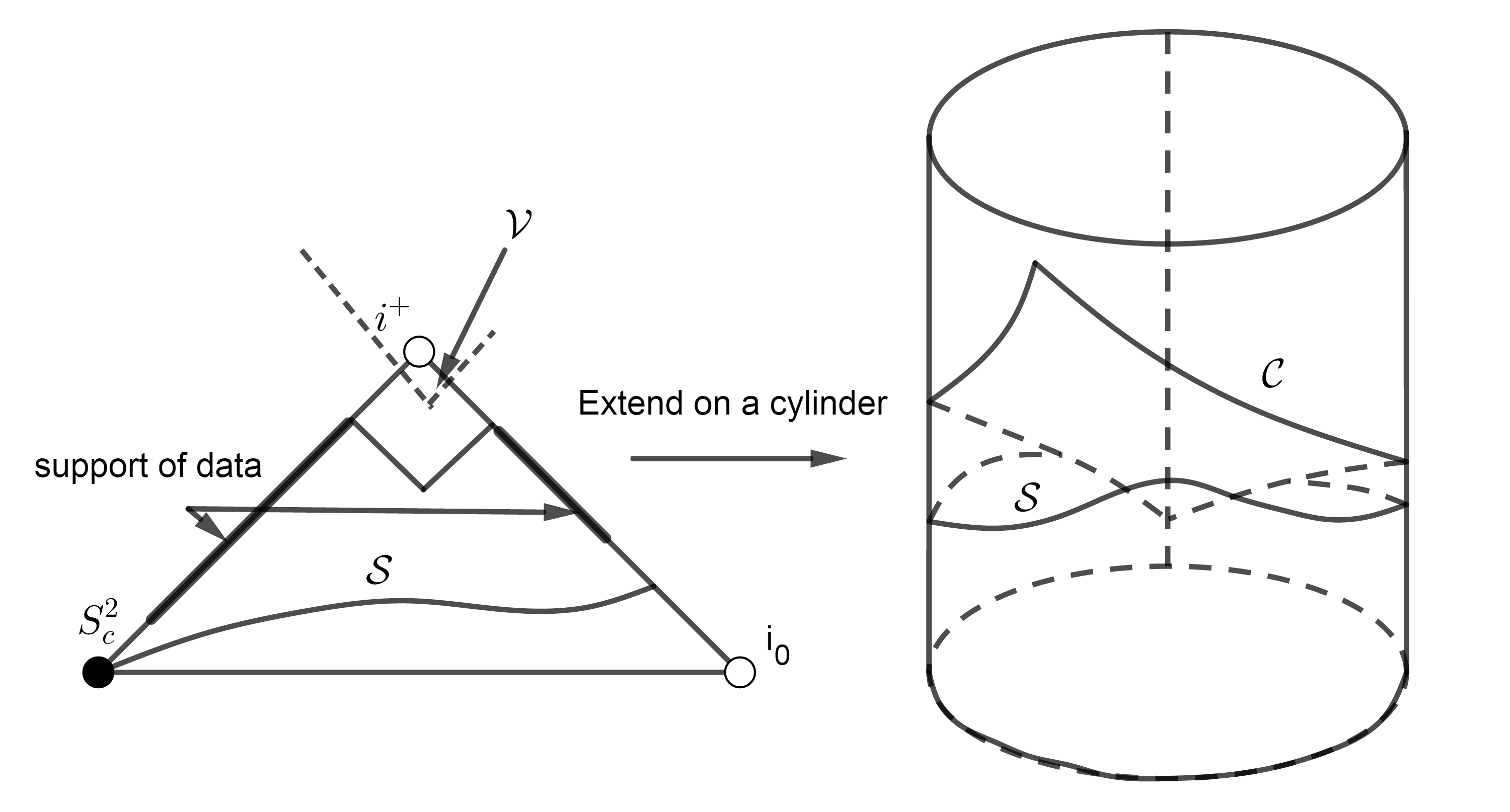}
\caption{The extension of $\mathfrak{M}=\mathcal{I}^+(\mathcal{S})-\mathcal{V}$ onto a global hyperbolic cylindre.}
\end{center}
\end{figure}
To avoid the singularities at $i^+$ and $i_0$, we use the same method as in Appendix B in \cite{Ni2016,Pha2020,Pha2022}. In particular, we take $\mathcal{S}$ which is a spacelike hypersurface on $\bar{\mathcal{M}}$ whose intersection with the horizon is the crossing sphere and which crosses $\scri^+$ strictly in the past of the support of the data. We cut $\mathcal{I}^+(\mathcal{S})$ by a neighbourhood $\mathcal{V}$ of a point in $\bar{\mathcal{M}}$ lying in the future of the support of the Goursat data and get a spacetime denoted by $\mathfrak{M}$. Then, we extend $\mathfrak{M}$ as a cylindrical globally hyperbolic spacetime $(\mathbb{R}_t\times \mathbb{S}^3, \mathfrak{g})$, where $\mathfrak{g} = \d t^2 - h$ with $h(t)$ is a Riemannian metric on $\mathbb{S}^3$ smoothly varying with $t$. The conformal boundary $\mathfrak{H}^+\cup \scri^+$ is extended inside $$(\mathbb{R}_t\times \mathbb{S}^3, \mathfrak{g})$$ as a null hypersurface $\mathcal{C}$, that is the graph of a Lipschitz function over $\mathbb{S}^3$ and the initial data by zero on the rest of the extended hypersurface $\mathcal{C}-(\mathfrak{H}^+\cup \scri^+)$. Here, we still use the notation of extending metric as in the proof of Theorem \ref{Cauchyproblem}.

Similar to the proof of Theorem \ref{Cauchyproblem}, we project the tensorial Fackerell-Ipser equation \eqref{Fac01} on the basic frame $(\slashed{\nabla}_{\partial_{\theta^a}},\slashed{\nabla}_{\partial_{\theta^b}})$ of $\Lambda^1(\mathbb{S}^2)$ and get
\begin{equation}\label{system-wave}
P_{\hat g}\Phi + L_1\Phi = 0,
\end{equation}
with 
$$P_{\hat g}= \left(\begin{matrix}
\Box_{\hat{g}} &&0\\
0&&\Box_{\hat{g}}
\end{matrix} \right)\hbox{      } (\hbox{here   }\Box_{\hat{g}} = \frac{r^2}{F}L\underline{L} - {\Delta}_{\mathbb{S}^2})$$
is a diagonal matrix, ${\Phi} = \left( \begin{matrix} {\phi}_1\\ {\phi}_2\end{matrix} \right)$
and $L_1 = (L_1^{ij})_{2\times 2}$ is a $2\times 2$-matrix with $L_1^{11}=L_1^{22},\, L_1^{12}=L_1^{21}$ and $L_1^{ij}$ are first order differential operators with smooth coefficients
$L_1^{ij} = b_0^{ij}\partial_t + b_1^{ij}\partial_x + c^{ij}.$

In the extending spacetime $(\mathbb{R}_t\times \mathbb{S}^3, \mathfrak{g})$, Equation \eqref{system-wave} becomes
\begin{equation}\label{systemwave}
P_{\mathfrak{g}}({\Phi}) + L_1(\Phi) = 0,
\end{equation}
where
$$P_{\mathfrak{g}}= \left(\begin{matrix}
\Box_{\mathfrak{g}} &&0\\
0&&\Box_{\mathfrak{g}}
\end{matrix} \right),\, \Box_{\mathfrak{g}} = \partial_t^2 - \Delta_h.$$

The following lemma is extended from resutls in \cite{Ho1990}:
\begin{lemma}(Goursat problem for Equation \eqref{systemwave} in $(\mathbb{R}_t\times \mathbb{S}^3,\mathfrak{g})$).\label{partlyGoursat}
For any foliation $\left\{\widetilde{S}_\tau\right\}_{\tau \in \mathbb{R}}$ of $(\mathbb{R}_t\times \mathbb{S}^3,\mathfrak{g})$, where $\widetilde{S}_0=\left\{ 0\right\}\times \mathbb{S}^3$ and $\widetilde{S}_\tau$ is a $\mathfrak{g}$-spacelike hypersurface which is topological $3$-sphere endowed with the Riemannian metric $-\mathfrak{g}|_{\mathcal{S}_\tau}$ for all $\tau$.
For the initial data $(\widetilde{\xi}_i,\,\widetilde{\zeta}_i) \in  C_0^{\infty}(\mathcal{C}) \, (i=1,2)$, Equation \eqref{systemwave} has a unique smooth solution $\widetilde{\Phi} = (\widetilde{\phi}_1,\, \widetilde{\phi}_2)$ satisfying
\begin{equation}\label{GoursatCond}
\widetilde{\phi}_i \in C^\infty(\mathbb{R}_\tau;\, \cup_{\tau\in \mathbb{R}} H^1(\widetilde{S}_{\tau}), \, \partial_t \widetilde{\phi}_i \in {C}^\infty(\mathbb{R}_\tau; \, \cup_{\tau\in \mathbb{R}} L^2(\widetilde{S}_\tau))
\end{equation}
for all $i = 1,\, 2.$ 
\end{lemma}
\begin{proof}
The proof is given in Appendix \ref{appendix}.
\end{proof}

Using Lemma \ref{partlyGoursat}, we obtain the well-posedness of the Goursat problem of \eqref{system-wave}, hence \eqref{Fac01} in $\mathcal{I}^+(\mathcal{S})$ in the following corollary.
\begin{corollary}(Goursat problem for Equation \eqref{systemwave} in $\mathcal{I}^+(\mathcal{S})$).\label{partGo}
For any foliation $\left\{ \mathcal{S}_\tau\right\}_{\tau\geq 0}$ of $\mathcal{I}^+(\mathcal{S})$, where $\mathcal{S}_\tau$ is a $\hat{g}$-spacelike hypersurface and $\mathcal{S}_0=\mathcal{S}$ and the initial data $(\xi_i, \zeta_i) \in {C}_0^\infty(\mathfrak{H}^+) \times {C}_0^\infty(\scri^+)$, where $i=1,\, 2$, Equation \eqref{system-wave} or the tensorial Fackerell-Ipser equation \eqref{Fac01} has a unique  smooth solution $\Phi=({\phi}_1,\phi_2)$ in $\mathcal{I}^+(\mathcal{S})$ satisfying
\begin{equation}\label{GousatCD}
{\phi}_i \in C^\infty(\mathbb{R}_\tau;\, \cup_{\tau\geq 0}H^1(\mathcal{S}_\tau)), \, \partial_\tau\phi_i \in {C}^\infty(\mathbb{R}_\tau; \, \cup_{\tau\geq 0}L^2(\mathcal{S}_\tau))
\end{equation}
for all $i=1,\, 2$. 
\end{corollary}
\begin{proof}
In the beginning of this section, we have extended $\mathfrak{M}=\mathcal{I}^+(\mathcal{S})-\mathcal{V}$ onto a global hyperbolic spacetime $(\mathbb{R}_t\times \mathbb{S}^3,\mathfrak{g})$. In this spacetime, Equation \eqref{system-wave} has the form \eqref{systemwave}. 
Now, we extend the $\hat{g}$-spacelike $\mathcal{S}_\tau$ to a $\mathfrak{g}$-spacelike hypersurface $\widetilde{\mathcal{S}}_\tau$ in $(\mathbb{R}_t\times \mathbb{S}^3,\mathfrak{g})$ for each $\tau$. The hypersurface $\widetilde{\mathcal{S}}_\tau$ is  topological $3$-spheres endowed with the Riemannian metric $-\mathfrak{g}|_{\widetilde{\mathcal{S}}_\tau}$. 
By using Lemma \ref{partlyGoursat}, the Goursat problem of equation \eqref{systemwave} has a unique smooth solution $\widetilde{\Phi}$ in $(\mathbb{R}_t\times \mathbb{S}^3,\mathfrak{g})$ that satisfies \eqref{GoursatCond}.
By local uniqueness and causality, using in particular the fact that as a consequence of the
finite propagation speed, the solution $\widetilde{\Phi}$ of \eqref{systemwave} obtained in Lemma \ref{partlyGoursat} vanishes in $\mathcal{V}$. 
Therefore, the Goursat problem of Equation \eqref{system-wave} has a unique smooth solution $\Phi=(\phi_1,\phi_2)$ in $\mathcal{I}^+(\mathcal{S})$, that is the restriction of $\widetilde{\Phi}=(\widetilde{\phi}_1,\widetilde{\phi}_2)$ to $\mathfrak{M}$. Since $\widetilde{\Phi}$ satisfies \eqref{GoursatCond}, we obatin that the solution $\Phi$ satifies \eqref{GousatCD}. Our proof is completed.
\end{proof}

\subsection{Goursat problem and conformal scattering operator}\label{GoursatFac}
In the previous section, we proved that the Goursat problem for the tensorial Fackerell-Ipser equation \eqref{Fac01} is well-posed in the future $\mathcal{I}^+(\mathcal{S})$. In order to establish the full solution of the Goursat problem, we need to extend the solution (which is obtained in the previous section) down to $\Sigma_0$, i.e., we prove the well-posedness of the Goursat problem in the past $\mathcal{I}^-(\mathcal{S})$. The solution of the Goursat problem is a union of the two solutions in $\mathcal{I}^+(\mathcal{S})$ and $\mathcal{I}^-(\mathcal{S})$.
\begin{theorem}(Goursat problem of Equation \eqref{Fac01} in $\mathcal{I}^+(\Sigma_0)$).\label{Goursat}
The Goursat problem of the tensorial Fackerell-Ipser equation \eqref{Fac01} is well-posed in $\mathcal{I}^+(\Sigma_0)$. This means that for the initial data $(\xi_a,\zeta_a)\in C_0^\infty(\Lambda^1(\mathbb{S}^2)|_{\mathfrak{H}^+})\times C_0^\infty(\Lambda^1(\mathbb{S}^2)|_{\scri^+})$, there exists a unique solution of Equation \eqref{Fac01} satisfying
$$({\phi}_a,\slashed{\nabla}_t{\phi}_a)\in {C}(\mathbb{R}_t; \, \cup_{t\geq 0}\mathcal{H}(\Lambda^1(\mathbb{S}^2)|_{\Sigma_t})) \hbox{  and  } \mathcal{T}^+({\phi}_a|_{\Sigma_0}, \slashed{\nabla}_t {\phi}_a|_{\Sigma_0}) = (\xi_a,\zeta_a).$$
The same assertion holds for Equation \eqref{ReFac02}.
\end{theorem}
\begin{proof}
Following Corollary \ref{partGo}, there exists a unique solution ${\phi}_a$ of the Gourast problem of Equation \eqref{Fac01} which satisfies the following properties.
\begin{itemize}
\item[$\bullet$] ${\phi}_a \in H^1(\Lambda^1(\mathbb{S}^2)|_{\mathcal{I}^+(\mathcal{S}}))$, where $\mathcal{I}^+(\mathcal{S})$ is the causal future of $\mathcal{S}$ in $\bar{\mathcal{M}}$. Since the support of the initial data is compact, the solution vanishes in the neighbourhood $\mathcal{V}$ of $i^+$ (where, the neighbourhood $\mathcal{V}$ is chosen as in Subsection \ref{GeneHorm}, and the solution $\phi_a$ vanishes in $\mathcal{V}$ as a consequence of the finite propagation speed (see the proof of Corollary \ref{partGo})). Then, we do not need to distinguish between $H^1(\Lambda^1(\mathbb{S}^2)|_{\mathcal{I}^+(\mathcal{S})})$ and $H^1_{loc}(\Lambda^1(\mathbb{S}^2)|_{\mathcal{I}^+(\mathcal{S}})$. 

\item[$\bullet$] For any foliation $\left\{ \mathcal{S}_\tau\right\}_{\tau\geq 0}$ of $\mathcal{I}^+(\mathcal{S})$, where $\mathcal{S}_0=\mathcal{S}$, we have ${\phi}_a(\tau)$ in $H^1(\Lambda^1(\mathbb{S}^2)|_{\mathcal{S}_\tau})$ and $\partial_\tau \phi_a$ in $L^2(\Lambda^1(\mathbb{S}^2)|_{\mathcal{S}_\tau})$ for all $\tau\geq 0$.

\item[$\bullet$] ${\phi}_a|_{\scri^+} = \zeta_a, \, {\phi}_a|_{\mathfrak{H}^+}= \xi_a$.
\end{itemize}
\begin{figure}[H]
\begin{center}
\includegraphics[scale=0.2]{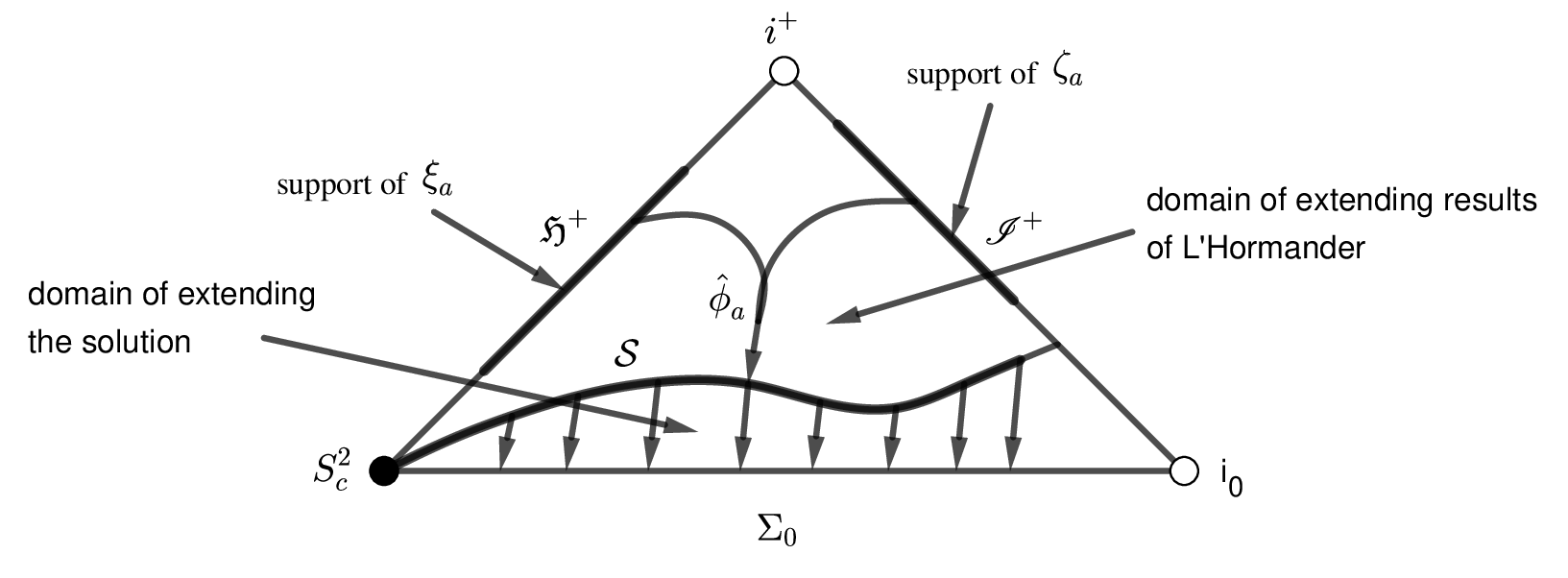}
\caption{Extending the solution $\hat{\phi}_a$ down to $\Sigma_0$.}
\end{center}
\end{figure}
We need to extend the solution down to $\Sigma_0$ in a manner that avoids the singularity at $i_0$.
Since the hypersurface $\mathcal{S}$ intersects the horizon at the crossing sphere and intersects $\scri^+$ strictly in the past of the support of the data, we have that the restriction of ${\phi}_a$ to $\mathcal{S}$ is in $H^1(\Lambda^1(\mathbb{S}^2)|_{\mathcal{S}})$ and its trace on $\mathcal{S}\cap \scri^+$ is also the trace of $\zeta_a$ on $\mathcal{S}\cap \scri^+$, hence this trace is zero. Therefore, ${\xi}_a|_{\Lambda^1(\mathcal{S}^2)|_{\mathcal{S}}}$ can be approached by a sequence $\left\{{\phi}^n_{a}|_{\Lambda^1(\mathbb{S}^2)|_{\mathcal{S}}} \right\}_{n\in \mathbb{N}}$ of the smooth tensorial fields on $\Lambda^1(\mathbb{S}^2)|_{\mathcal{S}}$ supported away from $\scri^+$ that converges towards ${\phi}_a|_{\Lambda^1(\mathbb{S}^2)|_{\mathcal{S}}}$ in $H^1(\Lambda^1(\mathbb{S}^2)|_{\mathcal{S}})$. Moreover, $\slashed{\nabla}_{\partial_t}{\phi}_a|_{\Lambda^1(\mathbb{S}^2)|_{\mathcal{S}}}$ can be approached by a sequence $\left\{\slashed{\nabla}_{\partial_t}{\phi}^n_{a}|_{\Lambda^1(\mathbb{S}^2)|_{\mathcal{S}}} \right\}_{n\in \mathbb{N}}$ of the smooth tensorial fields on $\Lambda^1(\mathbb{S}^2)|_{\mathcal{S}}$ supported away from $\scri^+$ that converges towards $\slashed{\nabla}_{\partial_t}{\phi}_a|_{\Lambda^1(\mathbb{S}^2)|_{\mathcal{S}}}$ in $L^2(\Lambda^1(\mathbb{S}^2)|_{\mathcal{S}})$.
For the initial data $({\phi}^n_{a}|_{\Lambda^1({\mathbb{S}^2})|_{\mathcal{S}}},\slashed{\nabla}_{\partial_t}{\phi}^n_{a}|_{\Lambda^1(\mathbb{S}^2)|_{\mathcal{S}}})$, we let $\psi_a^n$ be the smooth tensorial solution of Cauchy problem of Equation \eqref{Fac01} on $\Lambda^1(\mathbb{S}^2)|_{\bar{\mathcal{M}}}$ (the existence by Theorem \ref{Cauchyproblem}). This tensorial solution vanishes in the neighbourhood of $i_0$ and we can establish energy estimates for ${\psi}_a^n$ between $\mathcal{S}$ and $\Sigma_0$ by using the conservation law \eqref{zero1} as follows
\begin{equation}
\mathcal{E}^T_{\mathcal{S}}({\psi}_a^n) = \mathcal{E}^T_{\Sigma_0}({\psi}_a^n).
\end{equation}
By the same way, we have energy identities between $\mathcal{S}$ and the hypersurfaces $\Sigma_t \hbox{  for   } t> 0$. Therefore, the sequence $({\psi}_a^n,\slashed{\nabla}_{\partial_t}{\psi}_a^n)$ converges towards $({\psi}_a, \slashed{\nabla}_{\partial_t}{\psi}_a)$ in ${C}(\mathbb{R}_t,\cup_{t\in \mathbb{R}}\mathcal{H}(\Sigma_t))$, where ${\psi}_a$ is a solution of \eqref{Fac01}. By local uniqueness ${\psi}_a$ coincides with ${\phi}_a$ in the future of $\mathcal{S}$. Therefore, we have
$$\left( {\phi}_a|_{\Sigma_0},\, \slashed{\nabla}_{\partial_t}{\phi}_a|_{\Sigma_0}  \right) \in \mathcal{H}(\Lambda^1(\mathbb{S}^2)|_{\Sigma_0})$$
and  
$$\mathcal{T}^+\left( {\phi}_a|_{\Sigma_0},\, \slashed{\nabla}_{\partial_t}{\phi}_a|_{\Sigma_0} \right) = (\xi_a,\zeta_a).$$
Therefore, the range of $\mathcal{T}^+$ contains ${C}_0^\infty(\Lambda^1(\mathbb{S}^2)|_{\mathfrak{H}^+})\times {C}_0^\infty(\Lambda^1(\mathbb{S}^2)|_{\scri^+})$.
\end{proof}

Theorem \ref{Goursat} shows that the trace operator $\mathcal{T}^+: \mathcal{H}(\Sigma_0) \to \mathcal{H}^+$ is surjective. Combining with Theorem \ref{Trace}, we obtain that the trace operator $\mathcal{T}^+: \mathcal{H}(\Sigma_0) \to \mathcal{H}^+$ is an isometric operator. Similarly, we can construct the space $\mathcal{H}^-$ of past scattering data on the past horizon and the past null infinity and the past trace operator $\mathcal{T}^-: \mathcal{H}(\Sigma_0)\to \mathcal{H}^-$ which is an isometric operator. Therefore,
we can define the conformal scattering operator for the tensorial Fackerell-Ipser equation \eqref{Fac01} (resp. \eqref{Fac02}) as follows
\begin{definition}
The conformal scattering operator $S: \mathcal{H}^- \to \mathcal{H}^+$ of the tensorial Fackerell-Ipser equation \eqref{Fac01} (resp. \eqref{Fac02}) is an isometry which maps the past scattering data to the future scattering data, i.e.,
$$S:= \mathcal{T}^+\circ (\mathcal{T}^-)^{-1}.$$
\end{definition}

\section{Conformal scattering for the spin $\pm 1$ Teukolsky equations}\label{ConTeu}
In this section, we will use the results obtained in Section \ref{ConFac} to establish the conformal scattering operator for the spin $+1$ Teukolsky equations \eqref{Teu1}. The construction for the spin $-1$ Teukolsky equation \eqref{Teu2} is done by the same way. Our method is developed from the recent work \cite{Masao}.
\subsection{The tensorial field and scattering data spaces}\label{TraceScat}
First, we define the finite energy space for the spin $+1$ Teukolsky equation \eqref{Teu1} by the following proposition.
\begin{proposition}\label{norm}
If we put
\begin{equation}\label{Norm1}
\left\|(\alpha_{a}|_{\Sigma_0},\alpha'_{a}|_{\Sigma_0})\right\| := \left\| (\phi_{a}|_{\Sigma_0},\slashed{\nabla}_{\partial_t}\phi_{a}|_{\Sigma_0}) \right\|_{\mathcal{H}(\Sigma_0)},
\end{equation}
then \eqref{Norm1} determines a norm for $\alpha_a$ on $\Sigma_0$ and we define by $\mathcal{H}^1(\Sigma_0)$ the completion of ${C}_0^\infty (\Lambda^1(\mathbb{S}^2)|_{\Sigma_0})\times {C}_0^\infty (\Lambda^1(\mathbb{S}^2)|_{\Sigma_0})$ in the norm \eqref{Norm1}. Here, 
$\phi_{a} = \dfrac{r^2}{F}\slashed{\nabla}_{\underline L}(r\alpha_{a}), \, \alpha'_{a} = \slashed{\nabla}_{\partial_t}\alpha_{a}$ and the space $\mathcal{H}(\Sigma_0)$ is defined by Definition \ref{FunSpaFac}. Similarly, we have the definition of space $\mathcal{H}^1(\Sigma_\tau)$ for $\tau >0$.
\end{proposition}
\begin{proof}
We need to prove that if $\left\|(\alpha_{a},\alpha'_{a})\right\|_{\mathcal{H}^1(\Sigma_0)} = 0$,  for a smooth, compactly supported tensors $\alpha_{a}$ and $\alpha'_{a}$, then $\alpha_{a}|_{\Sigma_0} = \alpha'_{a}|_{\Sigma_0}=0$. Indeed, the equality $\left\|(\alpha_{a},\alpha'_{a})\right\|_{\mathcal{H}^1(\Sigma_0)} = 0$ and the definition \eqref{Norm1} lead to
$$\left\| (\phi_{a}|_{\Sigma_0},\slashed{\nabla}_{\partial_t}\phi_{a}|_{\Sigma_0}) \right\|_{\mathcal{H}(\Sigma_0)}=0.$$
By using \eqref{ENERGY}, the above equality is equivalent to
$$\frac{1}{\sqrt 2}\left( \int_{\Sigma_0}\left( |\slashed{\nabla}_{\partial_t}\phi_{a}|^2 + |\slashed{\nabla}_{\partial_{r_*}}{\phi}_{a}|^2 + R^2F|\slashed{\nabla}_{\mathbb{S}^2}{\phi}_{a}|^2 + R^2F|{\phi}_{a}|^2 \right) \d r_*\d \mathbb{S}^2 \right)^{1/2}=0.$$
Therefore,
$$\phi_{a}|_{\Sigma_0} = \slashed{\nabla}_{\partial_t}\phi_{a}|_{\Sigma_0} =0.$$
Since Equation \eqref{Tran} and the Teukolsky equation \eqref{Teu1}, we have
\begin{eqnarray*}
\slashed{\nabla}_L\phi_a &=& \slashed{\nabla}_L \left( \frac{r^2}{F}\slashed{\nabla}_{\underline L}(r\alpha_a) \right)\cr
&=& \slashed{\nabla}_L \left( \frac{r^2}{F} \right)\slashed{\nabla}_{\underline L}(r\alpha_a) + \frac{r^2}{F}\slashed{\nabla}_L \slashed{\nabla}_{\underline L}(r\alpha_a)\cr
&=& \frac{2r}{F}\left( 1- \frac{3M}{r}\right)\slashed{\nabla}_{\underline L}(r\alpha_a) + \frac{r^2}{F}\slashed{\nabla}_L \slashed{\nabla}_{\underline L}(r\alpha_a)\cr
&=& \frac{2r}{F}\left( 1- \frac{3M}{r}\right)\slashed{\nabla}_{\underline L}(r\alpha_a) + \frac{r^2}{F}\slashed{\nabla}_L \slashed{\nabla}_{\underline L}(r\alpha_a)\cr
&=& \slashed{\Delta}_{\mathbb{S}^2}(r\alpha_a) - r\alpha_a.
\end{eqnarray*}
Since $\partial_t$ is a Killing vector field, we obtain also that
$$\slashed{\nabla}_L \left(\slashed{\nabla}_{\partial_t}\phi_a \right) = \slashed{\Delta}_{\mathbb{S}^2}\left(r\slashed{\nabla}_{\partial_t}\alpha_a \right)  - r\slashed{\nabla}_{\partial_t}\alpha_a.$$
Combining these equalities with $\phi_{a}|_{\Sigma_0}= \slashed{\nabla}_{\partial_t}\phi_{a}|_{\Sigma_0} = 0$, we get
$$\slashed{\Delta}_{\mathbb{S}^2}(r\alpha_{a}|_{\Sigma_0}) - r\alpha_{a}|_{\Sigma_0} = \slashed{\Delta}_{\mathbb{S}^2}\left(r\slashed{\nabla}_{\partial_t}\alpha_a \right)  - r\slashed{\nabla}_{\partial_t}\alpha_a = 0.$$
Since the operator $\slashed{\Delta}_{\mathbb{S}^2} - \mathrm{Id}$ (where $\mathrm{Id}$ is identity operator) is uniformly elliptic on the set of symmetric, traceless $2$-tensor field on $\mathbb{S}^2$, we have that $\alpha_{a}|_{\Sigma_0} = \alpha'_{a}|_{\Sigma_0} =0$. Our proof is completed.
\end{proof}
Similar to Theorem \ref{Cauchyproblem}, we obtain the well-posedness of Cauchy problem for the Teukolsky equation \eqref{Teu1} in the conformal rescaled spacetime $(\bar{\mathcal{M}},\hat{g})$. The well-posedness of Cauchy problem allows us to define the trace operator on conformal boundary $\mathfrak{H}^+\cup \scri^+$. 
\begin{theorem}(Cauchy problem for Equation \eqref{Teu1} in $(\bar{\mathcal{M}},\hat{g})$).\label{CauchyTeu}
For the initial data $(\beta_{a},\,\beta'_{a}) \in \mathcal{H}^1(\Lambda^1(\mathbb{S}^2)|_{\Sigma_0})$, the Cauchy problem for \eqref{Teu1} on $\Lambda^1(\mathbb{S}^2)|_{\bar{\mathcal{M}}}$ is well-posed in $ \cup_{t\in \mathbb{R}}\mathcal{H}^1(\Lambda^1(\mathbb{S}^2)|_{\Sigma_t})$. In other words, there exists a unique solution $\alpha_a \in \mathcal{D}'(\Lambda^1(\mathbb{S}^2)|_{\bar{\mathcal{M}}})$ of \eqref{Teu1} such that
\begin{equation}\label{CDteukolsky}
(\alpha_a,\slashed{\nabla}_{\partial_t}\alpha_a) \in {C}(\mathbb{R}_t;\, \cup_{t\in \mathbb{R}}\mathcal{H}^1(\Lambda^1(\mathbb{S}^2)|_{\Sigma_t})): \, \alpha_a|_{\Sigma_0}={\beta}_{a}; \, \slashed{\nabla}_{\partial_t}\alpha_a|_{\Sigma_0} ={\beta}'_{a}.
\end{equation}
Moreover, $\alpha_a$ belongs to $H^1_{loc}(\Lambda^1(\mathbb{S}^2)|_{\bar{\mathcal{M}}})$.
\end{theorem}
\begin{proof}
We prove the well-posedness in the future domain $\mathcal{I}^+(\Sigma_0)=\left\{t\geq 0\right\}$ of $\bar{\mathcal{M}}$, the well-posedness in the past domain $\mathcal{I}^-(\Sigma)=\left\{ t\leq 0\right\}$ is done similarly.  Multiplying the spin $+1$ Teukolsky equation \eqref{Teu1} by the factor $\dfrac{r^2}{F}$, we get
\begin{equation}\label{ReTeu1}
\frac{r^2}{F}\slashed{\nabla}_L\slashed{\nabla}_{\underline L}(r\alpha_a) + \frac{2r}{F}\left( 1-\frac{3M}{r} \right)\slashed{\nabla}_{\underline L}(r\alpha_a) - \slashed{\Delta}_{\mathbb{S}^2}(r\alpha_a) + r\alpha_a=0.
\end{equation}
Projecting Equation \eqref{ReTeu1} on the basic frame $(\slashed{\nabla}_{\partial^a_\theta},\slashed{\nabla}_{\partial^b_\theta})$ of $\Lambda^1(\mathbb{S}^2)$, we get the matrix equation which is similar to \eqref{system-wave}: 
\begin{equation}\label{new}
P_{\hat g}\Psi + \widetilde{L}_1\Psi = 0
\end{equation}
but with the first order differential operator $\widetilde{L}_1=(\widetilde{L}_1^{ij})_{2\times 2}$ still satisfies $\widetilde{L}_1^{11}=\widetilde{L}_1^{22},\, \widetilde{L}_1^{12}=\widetilde{L}_1^{21}$ and the unknown vector 
$${\Psi} = \left( \begin{matrix} \psi_1\\ \psi_2\end{matrix} \right),$$
where $\psi_i=r\alpha_i$ and $\alpha_i \, (i=1,2)$ are the scalar components of $\alpha_a$ decomposed on the basic frame $(\slashed{\nabla}_{\partial^a_\theta},\slashed{\nabla}_{\partial^b_\theta})$ of $\Lambda^1(\mathbb{S}^2)$.

Similar to the proof of Theorem \ref{Cauchyproblem}, we cut off $\mathcal{I}^+(\Sigma_0)$ by $\mathcal{O}$ which is a union of far enough neighbourhoods of $i^+$ and $i_0$. Then, we extend $(\mathcal{I}^+(\Sigma_0)-\mathcal{O},\hat{g})$ onto a cylindrical globally hyperbolic spacetime $(\mathbb{R}_t\times \mathbb{S}^3, \mathfrak{g})$ where $\mathfrak{g} = \d t^2 - h$ with $h(t)$ is a Riemannian metric on $\mathbb{S}^3$ smoothly with respect to $t$. For each $t\geq 0$, the hypersurface $\Sigma_t-\mathcal{O}$ is extended inside $(\mathbb{R}_t\times \mathbb{S}^3, \mathfrak{g})$ as a spacelike hypersurface $\mathcal{U}_t$. 
The conformal boundary $(\mathfrak{H}^+\cup \scri^+)-\mathcal{O}$ is extended inside $(\mathbb{R}_t\times \mathbb{S}^3, \mathfrak{g})$ as a null hypersurface $\mathscr{C}$.
Equation \eqref{new} is exstended to the following equation (which is similar to \eqref{systemwave}):
\begin{equation}\label{newnew}
P_{\mathfrak{g}} \Psi + \widetilde{L}_1\Psi=0.
\end{equation}
Therefore, Equation \eqref{newnew} is equivalent to a symmetrical hyperbolic system which consists two following equations (which are similar to \eqref{EQ1} and \eqref{EQ2}, respectively):
\begin{equation}\label{EQ1'}
(\partial_t^2\psi_1 - \Delta_h\psi_1) + (\widetilde{L}_1^{11}\psi_1+ \widetilde{L}_1^{12}\psi_2) = 0
\end{equation}
and
\begin{equation}\label{EQ2'}
(\partial_t^2\psi_2 - \Delta_h\psi_2) + (\widetilde{L}_2^{21}\psi_1+ \widetilde{L}_2^{22}\psi_2) =0.
\end{equation}
By using Leray's theorem, for smooth intitial data on $\mathcal{U}_0$, equations \eqref{EQ1'} and \eqref{EQ2'} have a unique smooth solution $\widetilde{\Psi}=(\widetilde{\psi}_1,\widetilde{\psi}_2)$ with $\widetilde{\psi}_i = r\widetilde{\alpha}_i\, (i=1,2)$ in smooth globally hyperbolic spacetime $(\mathbb{R}_t\times \mathbb{S}^3, \mathfrak{g})$. This corresponds to smooth solution $\widetilde{\alpha}_a$ of Equation \eqref{ReTeu1} in $(\mathbb{R}_t\times \mathbb{S}^3, \mathfrak{g})$.
By extending \eqref{Norm1} on $\mathcal{U}_\tau$, we have
$$\| (\widetilde{\alpha}_a|_{\mathcal{U}_\tau},\slashed{\nabla}_{\partial_t}\widetilde{\alpha}_a|_{\mathcal{U}_\tau})\|_{\mathcal{H}^1(\mathcal{U}_\tau)} = \| (\widetilde{\phi}_a,\slashed{\nabla}_{\partial_t}\widetilde{\phi}_a) \|_{\mathcal{H}(\mathcal{U}_\tau)} = \sqrt{\mathcal{E}(\widetilde{\Phi})(\tau)},$$
where $\widetilde{\phi}_a$ is smooth solution of Equation \eqref{Fac01} in $(\mathbb{R}_t\times \mathbb{S}^3, \mathfrak{g})$ and $\sqrt{\mathcal{E}(\widetilde{\Phi})(\tau)}$ is given by \eqref{EN}. Setting $\widetilde{\mathcal E}(\widetilde{\Psi})(\tau) = \| (\widetilde{\alpha}_a|_{\mathcal{U}_\tau},\slashed{\nabla}_{\partial_t}\widetilde{\alpha}_a|_{\mathcal{U}_\tau})\|_{\mathcal{H}^1(\mathcal{U}_\tau)}^2$, we have
$$\widetilde{\mathcal E}(\widetilde{\Psi})(\tau) = {\mathcal E}(\widetilde{\Phi})(\tau).$$
By energy estimate as \eqref{ENG}, we obtain the similar estimate that 
\begin{equation}\label{ENG'}
\widetilde{\mathcal E}(\widetilde{\Psi})(t)\leq \widetilde{\mathcal E}(\widetilde{\Psi})(s) e^{\widetilde{D}|t-s|}
\end{equation}
Using the existence of smooth solutions and energy estimate \eqref{ENG'}, we can obtain the global well-posedness of Cauchy problem for \eqref{EQ1'} and \eqref{EQ2'} for the initial data on $\mathcal{H}^1(\mathcal{U}_0)$ and we obtain the global solution $(\widetilde{\Psi},\partial_t\widetilde{\Psi})$ in $C(\mathbb{R}_t,\cup_{t\geq 0}\mathcal{H}^1(\mathcal{U}_t))$ (the process is similar to the one for solution $\widetilde{\Phi}$ obtained in the proof of Theorem \ref{Cauchyproblem} but for the energy norm $\sqrt{\widetilde{\mathcal E}(\widetilde{\Psi})(t)}$).
By local uniqueness and causality, using in particular the fact that as a consequence of the
finite propagation speed, the solution $\Psi$ of Cauchy problem of Equation \eqref{new} is the restriction of $\widetilde{\Psi}$ on $\mathcal{I}^+(\Sigma_0)-\mathcal{O}$.
Our proof is completed.

\end{proof}

In order to define the trace operators and tensorial field spaces of scattering data, we find the restrictions of $\slashed{\nabla}_L\phi_a$ on $\mathfrak{H}^+$ and $\scri^-$ and the ones of $\slashed{\nabla}_{\underline{L}}\phi_a$ on $\mathfrak{H}^-$ and $\scri^+$. 
In the proof of Proposition \ref{norm}, we proved that
\begin{equation}\label{Tensor1}
\slashed{\nabla}_L\phi_a = \slashed{\Delta}_{\mathbb{S}^2}(r\alpha_a) - r\alpha_a.
\end{equation}
Equality \eqref{Tensor1} leads to the following restrictions of $\slashed{\nabla}_L\phi_a$ on $\scri^-$ and $\mathfrak{H}^+$:
\begin{equation}\label{trans1}
\slashed{\nabla}_L\phi_a|_{\scri^-} = \slashed{\Delta}_{\mathbb{S}^2}(r\alpha_a)|_{\scri^-} - (r\alpha_a)|_{\scri^-}
\end{equation}
and   
\begin{equation}\label{trans2}
\slashed{\nabla}_L\phi_a|_{\mathfrak{H}^+} = 2M\slashed{\Delta}_{\mathbb{S}^2}\alpha_a|_{\mathfrak{H}^+} - 2M\alpha_a|_{\mathfrak{H}^+}\hbox{   (because   } r|_{\mathfrak{H}^+}=2M),
\end{equation}
respectively.

By the same way as in the proof of Proposition \ref{norm}, we can establish that
\begin{eqnarray}\label{Tensor2}
\slashed{\nabla}_{\underline L}\phi_a &=& \slashed{\nabla}_{\underline L} \left( \frac{r^2}{F}\slashed{\nabla}_{\underline L}(r\alpha_a) \right)\cr
&=& \slashed{\nabla}_{\underline L} \left(\slashed{\nabla}_{\underline L}\left(\frac{r^3}{F}\alpha_a \right) - \slashed{\nabla}_{\underline L} \left( \frac{r^2}{F} \right)(r\alpha_a)\right)\cr
&=& \slashed{\nabla}^2_{\underline L} \left(\frac{r^3\alpha_a}{F} \right) + \slashed{\nabla}_{\underline L} \left( \frac{2r^2\alpha_a}{F}\left( 1- \frac{3M}{r}\right)\right).
\end{eqnarray}
Equality \eqref{Tensor2} leads to the restriction of $\slashed{\nabla}_{\underline L}\phi_a$ on $\scri^+$ as
\begin{eqnarray}\label{trans3}
\slashed{\nabla}_{\underline L}\phi_a|_{\scri^+} &=&\lim_{r\to \infty}\slashed{\nabla}^2_{\underline L} \left(\frac{r^3\alpha_a}{F} \right) + \lim_{r\to \infty}\slashed{\nabla}_{\underline L} \left( \frac{2r^2\alpha_a}{F}\left( 1- \frac{3M}{r}\right)\right)\cr
&=& \slashed{\nabla}^2_{\underline L} (r^3\alpha_a)|_{\scri^+} + \lim_{r\to \infty} \slashed{\nabla}_{\underline L} (2r^2\alpha_a)\cr
&& \hbox{   (because   } F=1-\frac{2M}{r} \to 1 \hbox{  and  }1-\frac{3M}{r} \to 1 \hbox{  as  } r\to\infty)\cr
&=& \slashed{\nabla}^2_{\underline L} (r^3\alpha_a)|_{\scri^+},
\end{eqnarray}
if we consider that the tensorial field $r^3\alpha_a$ is regular on $\scri^+$, hence $(r^2\alpha_a)|_{\scri^+} = \lim\limits_{r\to \infty} \dfrac{1}{r} (r^3\alpha_a)=0$. And by \eqref{Tensor2}, the restriction of $\slashed{\nabla}_{\underline L}\phi_a$ on $\mathfrak{H}^-$ is
\begin{eqnarray}\label{trans4}
\slashed{\nabla}_{\underline L}\phi_a|_{\mathfrak{H}^-} &=& 8M^3\slashed{\nabla}^2_{\underline L} \left(\frac{\alpha_a}{F} \right)|_{\mathfrak{H}^-} + 2(2M)^2 \left(1 - \frac{3M}{2M}\right)\slashed{\nabla}_{\underline L} \left( \frac{\alpha_a}{F}\right)|_{\mathfrak{H}^-} \hbox{   (because   } r|_{\mathfrak{H}^-}=2M)\cr
&=& 8M^3\slashed{\nabla}^2_{\underline L} \left(\frac{\alpha_a}{F} \right)|_{\mathfrak{H}^-} - 4M^2\slashed{\nabla}_{\underline L} \left( \frac{\alpha_a}{F}\right)|_{\mathfrak{H}^-}.
\end{eqnarray}

Combining \eqref{trans2} and \eqref{trans3} (resp. \eqref{trans1} and \eqref{trans4}) with the well-posedness of Cauchy problem in Theorem \ref{CauchyTeu}, we can define the future (resp. past) trace operator for Equation \eqref{Teu1} on the conformal boundary $\mathfrak{H}^+\cup \scri^+$ (resp. $\mathfrak{H}^-\cup \scri^-$) in the following definition.
\begin{definition}\label{TRACETEU} (Trace operator for the spin $+1$ Teukolsky equation). Let $({\beta}_{a}, {\beta}'_{a}) \in {C}_0^\infty(\Lambda^1(\mathbb{S}^2)|_{\Sigma_0})\times {C}_0^\infty(\Lambda^1(\mathbb{S}^2)|_{\Sigma_0})$. Consider the smooth solution ${\alpha}_a$ of Equation \eqref{Teu1} such that
$${\alpha}_a|_{\Sigma_0} = {\beta}_{a}, \, \slashed{\nabla}_t {\alpha}_{a}|_{\Sigma_0} = {\beta}'_{a}.$$
The future trace operator $\mathfrak{T}^+$ from ${C}_0^\infty(\Lambda^1(\mathbb{S}^2)|_{\Sigma_0})\times {C}_0^\infty(\Lambda^1(\mathbb{S}^2)|_{\Sigma_0})$ to ${C}_0^\infty(\Lambda^1(\mathbb{S}^2)|_{\mathfrak{H}^+})\times {C}_0^\infty(\Lambda^1(\mathbb{S}^2)|_{\scri^+})$ is defined by
$$\mathfrak{T}^+({\beta}_{a}, {\beta}'_{a}) = ({\alpha}_a|_{\mathfrak{H}^+}, (r^3{\alpha}_a)|_{\scri^+}).$$
The past trace operator $\mathfrak{T}^-$ from ${C}_0^\infty(\Lambda^1(\mathbb{S}^2)|_{\Sigma_0})\times {C}_0^\infty(\Lambda^1(\mathbb{S}^2)|_{\Sigma_0})$ to ${C}_0^\infty(\Lambda^1(\mathbb{S}^2)|_{\mathfrak{H}^-})\times {C}_0^\infty(\Lambda^1(\mathbb{S}^2)|_{\scri^-})$ is defined by
$$\mathfrak{T}^-({\beta}_{a}, {\beta}'_{a}) = \left((F^{-1}{\alpha}_a)|_{\mathfrak{H}^-}, (r{\alpha}_a)|_{\scri^-} \right).$$
\end{definition}
We define also the tensorial field space for scattering data of the spin $+1$ Teukolsky equation \eqref{Teu1} by density as follows: 
\begin{definition}\label{normboundary}
The tensorial field space for scattering data $\mathcal{H}^{2,+}$ is the completion of ${C}_0^\infty(\Lambda^1(\mathbb{S}^2)|_{\mathfrak{H}^+}) \times {C}_0^\infty(\Lambda^1(\mathbb{S}^2)|_{\scri^+})$ under the norm
\begin{equation}\label{H2+norm}
\left\| (\xi_a,\zeta_a) \right\|_{\mathcal{H}^{2,+}} = \frac{1}{\sqrt 2}\left(\int_{\mathfrak{H}^+} |2M\slashed{\Delta}_{\mathbb{S}^2}\xi_a - 2M\xi_a|^2 \d v\d \mathbb{S}^2 + \int_{\scri^+} |\slashed{\nabla}^2_{\underline L} \zeta_a|^2 \d u \d \mathbb{S}^2 \right)^{1/2},
\end{equation}
which means
$$\mathcal{H}^{2,+} \simeq \dot{H}^2(\mathbb{R}_v; \, L^2(\Lambda^1(\mathbb{S}^2)|_{\mathfrak{H}^+})) \times \dot{H}^2(\mathbb{R}_u; \, L^2(\Lambda^1(\mathbb{S}^2)|_{\scri^+})).$$
On the other hand, the tensor space for scattering data $\mathcal{H}^{2,-}$ is the completion of ${C}_0^\infty(\Lambda^1(\mathbb{S}^2)|_{\mathfrak{H}^-}) \times {C}_0^\infty(\Lambda^1(\mathbb{S}^2)|_{\scri^-})$ in the norm
\begin{equation}\label{H2-norm}
\left\| (\xi_a,\zeta_a) \right\|_{\mathcal{H}^{2,-}} = \frac{1}{\sqrt 2}\left(\int_{\mathfrak{H}^-} |8M^3\slashed{\nabla}^2_{\underline L} \xi_a - 4M^2\slashed{\nabla}_{\underline L} \xi_a|^2 \d v\d \mathbb{S}^2 + \int_{\scri^-} |\slashed{\Delta}_{\mathbb{S}^2}\zeta_a - \zeta_a|^2 \d u \d \mathbb{S}^2 \right)^{1/2},
\end{equation}
which means
$$\mathcal{H}^{2,-} \simeq \dot{H}^2(\mathbb{R}_v; \, L^2(\Lambda^1(\mathbb{S}^2)|_{\mathfrak{H}^-})) \times \dot{H}^2(\mathbb{R}_u; \, L^2(\Lambda^1(\mathbb{S}^2)|_{\scri^-})).$$
\end{definition}
As another consequence of the equality energy \eqref{equality}, we have the following theorem
\begin{theorem}\label{TeuTraceInj}
The trace operator $\mathfrak{T}^+$ extends uniquely as a bounded linear map from $\mathcal{H}^1(\Lambda^1(\mathbb{S}^2)|_{\Sigma_0}))$ to $\mathcal{H}^{2,+}$. The extended operator is a partial isometry, i.e., for any initial data $({\beta}_a,{\beta}'_a) \in \mathcal{H}^1(\Lambda^1(\mathbb{S}^2)|_{\Sigma_0})$, we have
$$\left\| \mathfrak{T}^+({\beta}_{a},{\beta}'_{a}) \right\|_{\mathcal{H}^{2,+}} = \left\| ({\beta}_{a},{\beta}'_{a}) \right\|_{\mathcal{H}^1(\Lambda^1(\mathbb{S}^2)|_{\Sigma_0})}.$$
The same property holds for the past trace operator $\mathfrak{T}^-$.
\end{theorem}
\begin{proof}
For $({\beta}_{a},{\beta}'_{a})$ in $\mathcal{H}^1(\Sigma_0)$, by \eqref{Norm1}, we have that $(\xi_{a},\zeta_{a}) = \left(\dfrac{r^2}{F}\slashed{\nabla}_{\underline L}(r\alpha_{a0}), \dfrac{r^2}{F}\slashed{\nabla}_{\underline L}(r\alpha_{a1})\right)$ belongs to $\mathcal{H}(\Sigma_0)$. For this initial data $(\xi_{a},\zeta_{a})$, the Fackerell-Ipser equation \eqref{Fac01} has a unique solution $(\phi_a,\slashed{\nabla}_t\phi_a) \in C(\mathbb{R}_t,\cup_{t\in \mathbb{R}}\mathcal{H}(\Lambda^1(\mathbb{S}^2)|_{\Sigma_t}))$ by Theorem \ref{Cauchyproblem}.
Since the energy equality \eqref{equality}, we have 
\begin{equation}\label{eq}
\mathcal{E}^T_{\Sigma_0}(\phi_a) = \mathcal{E}^T_{\mathfrak{H}^+}(\phi_a) + \mathcal{E}^T_{\scri^+}(\phi_a).
\end{equation}
By Proposition \ref{norm} and energy fluxes \eqref{energyFlux4}, \eqref{energyFlux5} of $\phi_a$ through $\scri^+,\mathfrak{H}^+$, the equality \eqref{eq} leads to
$$\| (\beta_a,\beta'_a) \|^2_{\mathcal{H}^1(\Lambda^1(\mathbb{S}^2)|_{\Sigma_0})} = \int_{\mathfrak{H}^+} |(\slashed{\nabla}_{L}{\phi_a})|_{\mathfrak{H}^+}|^2 \d v \d \mathbb{S}^2 + \int_{\scri^+} |(\slashed{\nabla}_{{\underline L}}{\phi_a})|_{\scri^+}|^2 \d u \d \mathbb{S}^2.$$
Combining this inequality with \eqref{trans2}, \eqref{trans3} and \eqref{H2+norm}, we obtain 
$$\| (\beta_a,\beta'_a) \|_{\mathcal{H}^1(\Lambda^1(\mathbb{S}^2)|_{\Sigma_0})} = \|(\alpha_a|_{\mathfrak{H}^+}, (r^3\alpha_a)|_{\scri^+}) \|_{\mathcal{H}^{2,+}}.$$
Using Definition \ref{TRACETEU}, the above equality leads to
$$\| (\beta_a,\beta'_a) \|_{\mathcal{H}^1(\Lambda^1(\mathbb{S}^2)|_{\Sigma_0})} = \| \mathfrak{T}^+(\beta_a,\beta'_a) \|_{\mathcal{H}^1(\Lambda^1(\mathbb{S}^2)|_{\Sigma_0})}.$$
This completes our proof.
\end{proof}
\begin{remark}
In fact, Theorem \ref{TeuTraceInj} is a direct consequence of the energy equality \eqref{equality} of the tensorial Fackerell-Ipser equation \eqref{Fac01}, because we use the energy norm of solution $\phi_a$ of \eqref{Fac01} in Proposition \ref{norm} and Definition \ref{normboundary} to define the tensorial field spaces $\mathcal{H}^1(\Lambda^1(\mathbb{S}^2)|_{\Sigma_0})$ (Cauchy data on $\Sigma_0$) and $\mathcal{H}^{2,+}$ (Goursat data (scattering data) on $\mathfrak{H}^+\cup \scri^+$) for spin $+1$ Teukolsky equation \eqref{Teu1}.
\end{remark}

\subsection{Goursat problem and conformal scattering operator}\label{GoursatTeu}
In this Subsection, we will use the results in Subsection \ref{GoursatFac} to prove the well-posedness of the Goursat problem of the spin $+1$ Teukolsky equation \eqref{Teu1}. In particular, the tensorial Fackerell-Ipser equation \eqref{Fac01} is a consequence equation from \eqref{Teu1} by commuting with the operator $\slashed{\nabla}_{\underline L}$. Using this fact and the relation $\phi_a = \dfrac{r^2}{F}\slashed{\nabla}_{\underline L}(r\alpha_a)$ we will prove the following theorem.
\begin{theorem}\label{GGoursatTeu}(Goursat problem of the spin $+1$ Teukolsky equation \eqref{Teu1} in $\mathcal{I}^+(\Sigma_0)$).
The Goursat problem of the spin $+1$ Teukolsky equation \eqref{Teu1} is well-posed in $\mathcal{I}^+(\Sigma_0)$. This means that for the initial data $(\xi_a,\zeta_a)\in C_0^\infty(\Lambda^1(\mathbb{S}^2)|_{\mathfrak{H}^+})\times C_0^\infty(\Lambda^1(\mathbb{S}^2)|_{\scri^+})$, there exists a unique solution of \eqref{Teu1} satisfying
$$(\alpha_a,\slashed{\nabla}_t{\alpha}_a)\in {C}(\mathbb{R}_t; \, \cup_{t\geq 0}\mathcal{H}^1(\Lambda^1(\mathbb{S}^2)|_{\Sigma_t})) \hbox{  and  } \mathfrak{T}^+({\alpha}_a|_{\Sigma_0}, \slashed{\nabla}_t {\alpha}_a|_{\Sigma_0}) = (\xi_a,\zeta_a).$$
\end{theorem}
\begin{proof}
Since equations \eqref{trans2}, \eqref{trans3} and Definition \ref{TRACETEU}, we consider the following equations
\begin{equation}\label{Ini}
\slashed{\nabla}_L \phi_a|_{\mathfrak{H}^+} = 2M \slashed{\Delta}_{\mathbb{S}^2}\xi_a - 2M \xi_a,\,\, \, \slashed{\nabla}_{\underline L}\phi_a|_{\scri^+} = \slashed{\nabla}_{\underline L}^2\zeta_a.
\end{equation}
We find that the following tensor fields satisfy the above equations
\begin{equation}\label{Initial}
\phi_a|_{\mathfrak{H}^+} (v_0)= \int_{v_0}^{+\infty} \left( 2M \slashed{\Delta}_{\mathbb{S}^2}\xi_a - 2M \xi_a 
\right) \d v\hbox{   (for all  } v_0\in \mathbb{R}),\,\,\, \phi_a|_{\scri^+} = \slashed{\nabla}_{\underline L}\zeta_a.
\end{equation}
By Theorem \ref{Goursat}, for the initial data \eqref{Initial}, the Goursat problem of the Fackerell-Ipser equation \eqref{Fac01} has a unique solution $(\phi_a, \slashed{\nabla}_{\partial_t}\phi_a)\in C(\mathbb{R}_t,\cup_{t\geq 0}\mathcal{H}(\Sigma_t))$.

Now, if we define
\begin{equation}\label{alpha}
r\alpha _a(u_0) = 2M\xi_a - \int_{u_0}^{+\infty} \frac{F}{r^2}\phi_a \d u \hbox{      (for all   } u_0\in \mathbb{R}),
\end{equation}
then $\phi_a$ and $\alpha_a$ satisfy the relation \eqref{Tran}, i.e., $\phi_a = \dfrac{r^2}{F}\slashed{\nabla}_{\underline L}(r\alpha_a)$. Since $\phi_a$ satisfies the Fackerell-Ipser equation \eqref{Fac01} and the proof of Proposition \ref{relationEq}, we have 
$$\slashed{\nabla}_{\underline L}\left( \slashed{\nabla}_{L}\left( \frac{r^2}{F}\slashed{\nabla}_{\underline L}(r\alpha_a) \right) - r^2\slashed{\Delta}(r\alpha_a) + r\alpha_a \right) =0.$$
This corresponds to 
\begin{equation}\label{GoursatEq}
\slashed{\nabla}_{\underline L} \left( \frac{r^2}{F} \mathbb{T} (r\alpha_a) \right) =0,
\end{equation}
where $\mathbb{T}$ is Teukolsky operator
\begin{equation}\label{TO}
\mathbb{T}(r\alpha_a)= \slashed{\nabla}_L\slashed{\nabla}_{\underline L}(r\alpha_a) + \frac{2}{r}\left( 1-\frac{3M}{r} \right)\slashed{\nabla}_{\underline L}(r\alpha_a) - F\slashed{\Delta}(r\alpha_a) + \frac{F}{r^2}(r\alpha_a).
\end{equation}
Since \eqref{alpha}, we have 
\begin{equation*}
(r\alpha_a)|_{\mathfrak{H}^+} = \left( 2M\xi_a  - \int_{u_0}^{+\infty} \frac{F}{r^2}\phi_a \d u\right)|_{\mathfrak{H}^+}.
\end{equation*}
This leads to
\begin{equation*}
2M \alpha_a|_{\mathfrak{H}^+} = 2M\xi_a - \lim_{u_0\to +\infty}\int_{u_0}^{+\infty} \frac{F}{r^2}\phi_a \d u = 2M\xi_a.
\end{equation*}
Therefore, we get 
\begin{equation}\label{RestricAlpha}
\alpha_a|_{\mathfrak{H}^+}=\xi_a.
\end{equation}
Now, we calculate
\begin{eqnarray}\label{sca}
\left( \frac{r^2}{F}\mathbb{T}(r\alpha_a) \right)|_{\mathfrak{H}^+} &=& \left( \frac{r^2}{F}\slashed{\nabla}_L\slashed{\nabla}_{\underline L}(r\alpha_a) \right)|_{\mathfrak{H}^+} - \frac{1}{2M}\left(\frac{r^2}{F}\slashed{\nabla}_{\underline L}(r\alpha_a)\right)|_{\mathfrak{H}^+}\cr
&& - \slashed{\Delta}_{\mathbb{S}^2} (r\alpha_a)|_{\scri^-} + (r\alpha_a)|_{\scri^-} \hbox{     (because   } r|_{\mathfrak{H}^+}=2M)\cr
&=& \slashed{\nabla}_L \left( \frac{r^2}{F}\slashed{\nabla}_{\underline L}(r\alpha_a) \right)|_{\mathfrak{H}^+} - \left(\slashed{\nabla}_L \left( \frac{r^2}{F} \right)\slashed{\nabla}_{\underline L}(r\alpha_a)\right)|_{\mathfrak{H}^+} \cr
&&- \frac{1}{2M}\left(\frac{r^2}{F}\slashed{\nabla}_{\underline L}(r\alpha_a)\right)|_{\mathfrak{H}^+} - 2M \slashed{\Delta}_{\mathbb{S}^2} \alpha_a|_{\mathfrak{H}^+} + 2M\alpha_a|_{\mathfrak{H}^+}\cr
&=& \slashed{\nabla}_L \left( \frac{r^2}{F}\slashed{\nabla}_{\underline L}(r\alpha_a) \right)|_{\mathfrak{H}^+} - \left(\frac{2}{r} \left( 1-\frac{3M}{r}\right)\frac{r^2}{F}\slashed{\nabla}_{\underline L}(r\alpha_a)\right)|_{\mathfrak{H}^+} \cr
&&- \frac{1}{2M}\left(\frac{r^2}{F}\slashed{\nabla}_{\underline L}(r\alpha_a)\right)|_{\mathfrak{H}^+} - 2M \slashed{\Delta}_{\mathbb{S}^2} \alpha_a|_{\mathfrak{H}^+} + 2M\alpha_a|_{\mathfrak{H}^+}\cr
&=& \slashed{\nabla}_L \left( \frac{r^2}{F}\slashed{\nabla}_{\underline L}(r\alpha_a) \right)|_{\mathfrak{H}^+} + \frac{1}{2M}\left(\frac{r^2}{F}\slashed{\nabla}_{\underline L}(r\alpha_a)\right)|_{\mathfrak{H}^+} \cr
&&- \frac{1}{2M}\left(\frac{r^2}{F}\slashed{\nabla}_{\underline L}(r\alpha_a)\right)|_{\mathfrak{H}^+} - 2M \slashed{\Delta}_{\mathbb{S}^2} \alpha_a|_{\mathfrak{H}^+} + 2M\alpha_a|_{\mathfrak{H}^+}\cr
&=& \slashed{\nabla}_L \left( \frac{r^2}{F}\slashed{\nabla}_{\underline L}(r\alpha_a) \right)|_{\mathfrak{H}^+} - 2M \slashed{\Delta}_{\mathbb{S}^2}\alpha_a|_{\mathfrak{H}^+} + 2M\alpha_a|_{\mathfrak{H}^+}\cr
&=& \slashed{\nabla}_L \phi_a|_{\mathfrak{H}^+} - 2M \slashed{\Delta}_{\mathbb{S}^2} \xi_a + 2M\xi_a\cr
&=& 0\,\,\,\, \hbox{  (due to equations \eqref{Ini} and \eqref{RestricAlpha}}).
\end{eqnarray}
Therefore, integrating Equation \eqref{GoursatEq} follows $u$ (recall that $\partial_u =\underline{L}$) and using \eqref{sca}, we get
$$\frac{r^2}{F} \mathbb{T} (r\alpha_a) = \left(\frac{r^2}{F} \mathbb{T} (r\alpha_a)|_{\mathfrak{H}^+} \right)=0,$$
which means $\alpha_a$ given by \eqref{alpha} satisfies the spin $+1$ Teukosky equation \eqref{Teu1}.

Now, using \eqref{Initial} and \eqref{alpha}, we prove 
$$\mathfrak{T}^+({\alpha}_a|_{\Sigma_0}, \slashed{\nabla}_t {\alpha}_a|_{\Sigma_0}) = (\xi_a,\zeta_a).$$
Indeed, by Definition \ref{TRACETEU}, we have 
$$\mathfrak{T}^+({\alpha}_a|_{\Sigma_0}, \slashed{\nabla}_t {\alpha}_a|_{\Sigma_0})= (\alpha_a|_{\mathfrak{H}^+},(r^3\alpha_a)|_{\scri^+}).$$
Therefore, we need to prove that $\alpha_a|_{\mathfrak{H}^+}=\xi_a$ and $(r^3\alpha_a)|_{\scri^+} = \zeta_a$. The first restriction holds by \eqref{RestricAlpha}. For the second restriction, by  \eqref{Initial}, we have
\begin{eqnarray}\label{scate}
\slashed{\nabla}_{\underline L}\zeta_a &=& \phi_a|_{\scri^+}\cr
&=& \left( \dfrac{r^2}{F}\slashed{\nabla}_{\underline L}(r\alpha_a) \right)|_{\scri^+}\cr
&=& \slashed{\nabla}_{\underline L}\left( \dfrac{r^3}{F}\alpha_a \right)|_{\scri^+} + \frac{r^2}{F}\frac{2}{r}\left( 1- \frac{3M}{r}\right) (r\alpha_a)|_{\scri^+}\cr
&=& \slashed{\nabla}_{\underline L}\left( \dfrac{r^3}{F}\alpha_a \right)|_{\scri^+} + \lim_{r\to \infty}2F\left( 1- \frac{3M}{r}\right) (r^2\alpha_a)|_{\scri^+}\cr
&=& \slashed{\nabla}_{\underline L}\left( r^3\alpha_a \right)|_{\scri^+},
\end{eqnarray}
here $F=1-\frac{2M}{r}\to 1;\,1-\frac{3M}{r}\to 1$ as $r\to \infty$ and we considered that $r^3\alpha_a$ is regular on $\scri^+$, hence $r^2\alpha_a$ vanishes on $\scri^+$.
Integrating \eqref{scate} follows $u$ and using the fact that $\zeta_a$ has compact support,
we obtain that $\left( r^3\alpha_a \right)|_{\scri^+} = \zeta_a$. 

This means that $\alpha_a$ (given by \eqref{alpha}) is a solution of the Goursat problem for the Teukolsky equation \eqref{Teu1} with initial data $(\xi_a,\zeta_a)$. If $\gamma_a$ is another solution of the Goursat problem with the same initial data $(\xi_a,\zeta_a)$, then we obtain that $\xi_a-\gamma_a$ is a solution of the Goursat problem of Equation \eqref{Teu1} with initial data $(0,\, 0)$ on $\mathfrak{H}^+\cup \scri^+$.
This shows that $\phi_a = \frac{r^2}{F}\slashed{\nabla}_{\underline L}(\alpha_a-\gamma_a)$ is a solution of the Goursat problem of \eqref{Fac01} with initial data $(0,\, 0)$. By the uniqueness we have $\phi_a=0$, hence $\frac{r^2}{F}\slashed{\nabla}_{\underline L}(\alpha_a-\gamma_a)=0$. Integrating this equation and note that $(\alpha_a-\gamma_a)|_{\mathfrak{H}^+}=0$, we get $\alpha_a=\gamma_a$ and the uniqueness of the Goursat problem for spin $+1$ Teukolsky equation \eqref{Teu1} holds.
\end{proof}

By the same way as the proof of Theorem \ref{GGoursatTeu}, we establish the well-posedness of the Goursat problem for Equation \eqref{Teu1} in $\mathcal{I}^-(\Sigma_0)$ in the following theorem.
\begin{theorem}\label{GGoursatTeu'} (Goursat problem of the spin $+1$ Teukolsky equation \eqref{Teu1} in $\mathcal{I}^-(\Sigma_0)$).
The Goursat problem of the spin $+1$ Teukolsky equation \eqref{Teu1} is well-posed in $\mathcal{I}^-(\Sigma_0)$. This means that for the initial data $(\xi_a,\zeta_a)\in C_0^\infty(\Lambda^1(\mathbb{S}^2)|_{\mathfrak{H}^-})\times C_0^\infty(\Lambda^1(\mathbb{S}^2)|_{\scri^-})$, there exists a unique solution of \eqref{Teu1} satisfying
$$(\alpha_a,\slashed{\nabla}_t{\alpha}_a)\in {C}(\mathbb{R}_t; \, \cup_{t\geq 0}\mathcal{H}^1(\Lambda^1(\mathbb{S}^2)|_{\Sigma_t})) \hbox{  and  } \mathfrak{T}^-({\alpha}_a|_{\Sigma_0}, \slashed{\nabla}_t {\alpha}_a|_{\Sigma_0}) = (\xi_a,\zeta_a).$$
\end{theorem}
\begin{proof}
The proof is done by the same way of the one of Theorem \ref{GGoursatTeu}. However, since the past trace is different from the future trace (see their formulas in Definition \ref{TRACETEU}), we give here the detailed calculations.

Since equations \eqref{trans1}, \eqref{trans4} and Definition \ref{TRACETEU}, we consider $\phi_a$ the following equations
\begin{equation}\label{Ini'}
\slashed{\nabla}_L \phi_a|_{\scri^-} = \slashed{\Delta}_{\mathbb{S}^2}\zeta_a - \zeta_a,\,\, \, \slashed{\nabla}_{\underline L}\phi_a|_{\mathfrak{H}^-} = 8M^3\slashed{\nabla}_{\underline L}^2\xi_a - 4M^2\slashed{\nabla}_{\underline L}\xi_a.
\end{equation}
We find that the following tensor fields satisfy the above relation
\begin{equation}\label{Initial'}
\phi_a|_{\scri^-}(v_0) = \int_{v_0}^{+\infty} \left( \slashed{\Delta}_{\mathbb{S}^2}\zeta_a - \zeta_a 
\right) \d v \hbox{     (for all    } v_0\in \mathbb{R});\,\,\, \phi_a|_{\mathfrak{H}^-} = 8M^3\slashed{\nabla}_{\underline L}\xi_a - 4M^2\xi_a.
\end{equation}
By the well-posedness of Goursat problem for the tensorial Fackerell-Ipser equation \eqref{Fac01}, for the initial data \eqref{Initial'}, the Goursat problem of \eqref{Fac01} has a unique solution $(\phi_a, \slashed{\nabla}_{\partial_t}\phi_a)\in C(\mathbb{R}_t,\cup_{t\leq 0}\mathcal{H}(\Sigma_t))$.

We define
\begin{equation}\label{alpha'}
r\alpha _a (u_0) = \zeta_a + \int^{u_0}_{-\infty} \frac{F}{r^2}\phi_a \d u \hbox{    (for all   } u_0\in \mathbb{R}),
\end{equation}
then $\phi_a$ and $\alpha_a$ satisfy the relation \eqref{Tran}, i.e., $\phi_a = \dfrac{r^2}{F}\slashed{\nabla}_{\underline L}(r\alpha_a)$. Since supports of $\xi_a$ and $\zeta_a$ are compact, we obtain that the supports of $\phi_a$ and $\alpha_a$ on $\mathfrak{H}^+\cup \scri^+$ are also compact. 

By the same way in the proof of Theorem \ref{GGoursatTeu}, we have
\begin{equation}\label{GoursatEq'}
\slashed{\nabla}_{\underline L} \left( \frac{r^2}{F} \mathbb{T} (r\alpha_a) \right) =0,
\end{equation}
where $\mathbb{T}$ is Teukolsky operator \eqref{TO}.

By \eqref{alpha'}, we obtain the following restriction on $\scri^-$:
\begin{eqnarray}\label{ReT}
(r\alpha _a)|_{\scri^-} &=& \left( \zeta_a + \int^{u_0}_{-\infty} \frac{F}{r^2}\phi_a \d u\right)|_{\scri^-}\cr
&=& \zeta_a + \lim_{u_0\to -\infty}\int^{u_0}_{-\infty}\frac{F}{r^2}\phi_a\d u\cr
&=&\zeta_a.
\end{eqnarray}
On the other hand
\begin{eqnarray}\label{sca'}
\left( \frac{r^2}{F}\mathbb{T}(r\alpha_a) \right)|_{\scri^-} &=& \left( \frac{r^2}{F}\slashed{\nabla}_L\slashed{\nabla}_{\underline L}(r\alpha_a) \right)|_{\scri^-} + \left(\frac{2r}{F}\left(1-\frac{3M}{r}\right)\slashed{\nabla}_{\underline L}(r\alpha_a)\right)|_{\scri^-}\cr
&& - \slashed{\Delta}_{\mathbb{S}^2} (r\alpha_a)|_{\scri^-} + (r\alpha_a)|_{\scri^-} \cr
&=& \slashed{\nabla}_L \left( \frac{r^2}{F}\slashed{\nabla}_{\underline L}(r\alpha_a) \right)|_{\scri^-} - \left(\slashed{\nabla}_L \left(\frac{r^2}{F} \right)\slashed{\nabla}_{\underline L}(r\alpha_a)\right)|_{\scri^-} \cr
&& +\left(\frac{2r}{F}\slashed{\nabla}_{\underline L}(r\alpha_a)\right)|_{\scri^-} - \slashed{\Delta}_{\mathbb{S}^2} (r\alpha_a)|_{\scri^-} + (r\alpha_a)|_{\scri^-}\cr
&=& \slashed{\nabla}_L \phi_a|_{\scri^-} - \left(\frac{2r}{F}\left( 1-\frac{3M}{r} \right)\slashed{\nabla}_{\underline L}(r\alpha_a)\right)|_{\scri^-} \cr
&&+ \left(\frac{2r}{F}\slashed{\nabla}_{\underline L}(r\alpha_a)\right)|_{\scri^-} - \slashed{\Delta}_{\mathbb{S}^2}(r \alpha_a)|_{\scri^-} + (r\alpha_a)|_{\scri^-}\cr
&=& \slashed{\nabla}_L \phi_a|_{\scri^-} - \left(2r\slashed{\nabla}_{\underline L}(r\alpha_a)\right)|_{\scri^-} + \left(2r\slashed{\nabla}_{\underline L}(r\alpha_a)\right)|_{\scri^-} \cr
&&- \slashed{\Delta}_{\mathbb{S}^2}(r \alpha_a)|_{\scri^-} + (r\alpha_a)|_{\scri^-}\cr
&&\hbox{     (because on   }\scri^-:\,  F=1-\frac{2M}{r}\to 1; \, 1-\frac{3M}{r}\to 1 \hbox{    as   } r\to -\infty)\cr
&=& \slashed{\nabla}_L \phi_a|_{\scri^-} - \slashed{\Delta}_{\mathbb{S}^2}(r\alpha_a)|_{\scri^-} + (r\alpha_a)|_{\scri^-}\cr
&=& \slashed{\nabla}_L \phi_a|_{\scri^-} - \slashed{\Delta}_{\mathbb{S}^2} \zeta_a + \zeta_a\cr
&=& 0\,\,\,\, \hbox{  (due to equations \eqref{Ini'} and \eqref{ReT})}.
\end{eqnarray}
Therefore, integrating Equation \eqref{GoursatEq'} follows $u$ (recall that $\partial_u =\underline{L}$) and using \eqref{sca'}, we get
$$\frac{r^2}{F} \mathbb{T} (r\alpha_a) = \left(\frac{r^2}{F} \mathbb{T} (r\alpha_a)|_{\scri^-} \right)=0.$$
This means that $\alpha_a$ given by \eqref{alpha'} satisfies the spin $+1$ Teukosky equation \eqref{Teu1}.

Now, using \eqref{Initial'} and \eqref{alpha'}, we prove 
$$\mathfrak{T}^-({\alpha}_a|_{\Sigma_0}, \slashed{\nabla}_t {\alpha}_a|_{\Sigma_0}) = (\xi_a,\zeta_a).$$
Indeed, by Definition \ref{TRACETEU}, we have 
$$\mathfrak{T}^-({\alpha}_a|_{\Sigma_0}, \slashed{\nabla}_t {\alpha}_a|_{\Sigma_0})= ((F^{-1}\alpha_a)|_{\mathfrak{H}^-},(r\alpha_a)|_{\scri^-}).$$
Therefore, we need to prove that $(F^{-1}\alpha_a)|_{\mathfrak{H}^-}=\xi_a$ and $(r\alpha_a)|_{\scri^-} = \zeta_a$. The second restriction holds by \eqref{ReT}. For the first restriction, by using \eqref{Initial'}, we have
\begin{eqnarray}\label{ResTric'}
\slashed{\nabla}_{\underline L}(r\alpha_a)|_{\mathfrak{H}^-} &=& \left(\frac{F}{r^2}\phi_a \right)|_{\mathfrak{H}^-} \cr
&=&\left( \frac{F}{r^2} (8M^3\slashed{\nabla}_{\underline L}\xi_a  - 4M^2\xi_a) \right)|_{\mathfrak{H}^-}\cr
&=& \left( \frac{F}{r^2} (r^3\slashed{\nabla}_{\underline L}\xi_a  - r^2\xi_a) \right)|_{\mathfrak{H}^-}\cr
&=& \left( \frac{F}{r^2} \left(r^3\slashed{\nabla}_{\underline L}\xi_a  - \left( F+ \frac{2M}{r}\right)r^2\xi_a \right) \right)|_{\mathfrak{H}^-}\cr
&=& \left(rF\slashed{\nabla}_{\underline L}\xi_a  - \left( F+ \frac{2M}{r}\right)F\xi_a \right)|_{\mathfrak{H}^-}\cr
&=& \slashed{\nabla}_{\underline L}(rF\xi_a)|_{\mathfrak{H}^-}. 
\end{eqnarray}
Here, we used the fact that $r|_{\mathfrak{H}^-}=2M$ and $F|_{\mathfrak{H^-}} = \left(1-\frac{2M}{r}\right)|_{\mathfrak{H}^-}=0$. Equality \eqref{ResTric'} is equivalent to
$$\slashed{\nabla}_{\underline L}(r\alpha_a - rF\xi_a)|_{\mathfrak{H}^-} = 0.$$
Integrating this equality follows $u$ (recall that $\partial_u = \underline{L}$), we get that $(r\alpha_a - rF\xi_a)|_{\mathfrak{H}^-}$ is a constant. Since the supports of $\alpha_a$ and $\xi_a$ are compact on $\mathfrak{H}^-$, we have $\alpha_a(P)-\xi_a(P) =0$ at a point $P$ which does not belong to the union of supports of $\alpha_a|_{\mathfrak{H}^-}$ and $\xi_a$. Therefore, we obtain $\alpha_a|_{\mathfrak{H}^-} - (F\xi_a)|_{\mathfrak{H}^-}=0$. This is equivalent to $(F^{-1}\alpha_a)|_{\mathfrak{H}^-}=\xi_a|_{\mathfrak{H}^-}$. 

We proved that $\alpha_a$ (given by \eqref{alpha'}) is a solution of the Goursat problem for the Teukolsky equation \eqref{Teu1} with initial data $(\xi_a,\zeta_a)$ on $\mathfrak{H}^-\cup \scri^-$. The uniqueness is done by the same way as in the proof of Theorem \ref{GGoursatTeu}. Our proof is completed.
\end{proof}

As a direct consequence of the well-posedness of the Goursat problem in Theorem \ref{GGoursatTeu} we have that the future trace operator $\mathfrak{T}^+ : \mathcal{H}^1(\Lambda^1(\mathbb{S}^2)|_{\Sigma_0}) \to \mathcal{H}^{2,+}$ is surjective. Combining this with Theorem \ref{TeuTraceInj}, we obtain that the operator $\mathfrak{T}^+$ is an isometric operator. Similarly, the past trace operator $\mathfrak{T}^- : \mathcal{H}^1(\Lambda^1(\mathbb{S}^2)|_{\Sigma_0}) \to \mathcal{H}^{2,+}$ is isometric. Therefore, we can define the conformal scattering operator for the spin $+1$ Teukolsky equation \eqref{Teu1} as follows:
\begin{definition}
The conformal scattering operator $\mathfrak{S}: \mathcal{H}^{2,-}\to \mathcal{H}^{2,+}$ of the spin $+1$ Teukolsky equation \eqref{Teu1} is an isometric operator that maps the past scattering data to the future scattering data, i.e.,
$$\mathfrak{S}:= \mathfrak{T}^+\circ (\mathfrak{T}^-)^{-1}.$$
\end{definition}
\begin{remark}
\item[$\bullet$] By the same way as above, we can construct the conformal scattering operator for the spin $-1$ Teukolsky equation \eqref{Teu2}.

\item[$\bullet$] This work can be extended to construct conformal scattering theories for the tensorial Fackerell-Ipser and spin $\pm 1$ Teukolsky equations on the other symmetric spherical spacetimes such as Reissner-Nordstr\"om-de Sitter back hole. First, we extend the work \cite{Pa2019} to obtain the energy decays (where the results of Giorgi \cite{El2019,El2020'} can be useful) and then use these decays to establish the construction of the theory, where it remains useful to use the timelike Killing vector field $\partial_t$ to establish the energies of the fields on the Cauchy hypersurface $\Sigma_T= \left\{ t=T \right\}$.  

\item[$\bullet$] The extension of the conformal scattering theory for the scalar wave or tensorial wave equations (such as tensorial Fackerell-Ipser and spin $\pm 1$, $\pm 2$ Teukolsky equations) on the non-static and non-symmetric spherical spacetimes such as Kerr spacetimes is more complicated. In Kerr spacetimes, we have the energy and pointwise decay results obtained by Dafermos et all. \cite{Da2016}. However, the existence of the orbiting null geodesics and the fact that the vector field $\partial_t$ is no longer global timelike in the exterior domains. This fact leads to an issue that the energy on $\Sigma_T$ can be not defined by using $\partial_t$ as in Schwarzschild and Reissner-Nordstr\"om-de Sitter spacetimes. We need to choose another global timelike vector field on the exterior domain to define the energies on $\Sigma_T$ and we will not obtain conserved currents for the equations. This fact leads to the complication of the case of Kerr spacetime. In a recent work \cite{Pha2022}, we have established the conformal scattering theory for the massless Dirac equation on Kerr spacetime, where it remains to have a conserved current for the equation. This work can be useful for the construction of conformal scattering theories on Kerr spacetime for the scalar wave, tensorial wave and Maxwell equations. 

\item[$\bullet$] The peeling properties of the tensorial wave equations \eqref{TensorWave} and \eqref{TensorWave1} (where their rescaled forms are the tensorial Fackerell-Ipser equations \eqref{Fac01} and \eqref{Fac02}, respectively) on Schwarzschild spacetime can be established by the same method as in the previous work \cite{MaNi2009} (see recent work \cite{Pham2022}). However, the peeling probelms for the spin $\pm 1$ Teukolsky equations \eqref{Teu1} and \eqref{Teu2} (or spin $\pm 2$ Teukolsky equations) on Kerr spacetimes remain open and put an interesting question, where the method can be developed from \cite{NiXu2019}.
\end{remark}

\section{Appendix}
\subsection{The commutators}\label{appen}
We give the proof of following commutators which were used in Subsection \ref{Equation}: 
\begin{equation*}\label{commutator}
[r\slashed\nabla_{\partial_{x^a}},\, \slashed{\nabla}_L]= [r\slashed\nabla_{\partial_{x^a}},\slashed{\nabla}_{\underline{L}}]=0,\, [\slashed{\nabla}_L,\, \slashed{\nabla}_{\underline L}] =0, \, [r\slashed{\nabla}_{\partial_{x^a}},\slashed{\Delta}]= \dfrac{1}{r^2}(r\slashed{\nabla}_{\partial_{x^a}}).
\end{equation*}
The three first commutators are valid on both scalar functions and tensor fields and the last commutator is valid on scalar functions. By using the commutation formulas for projected covariant derivatives (which valid on both tensor fields and scalar fields) in the Schwarzschild spacetime obtained in Subsection 4.3.2 in \cite{Da2019}, we have
\begin{equation}\label{com1}
\slashed{\nabla}_{e_3}(\slashed{\nabla}_{\partial_{x^a}}\psi) - \slashed{\nabla}_{\partial_{x^a}}(\slashed{\nabla}_{e_3}\psi) = -\frac{1}{2}\mathrm{tr}\underline{\chi}\slashed{\nabla}_{\partial_{x^a}}\psi,
\end{equation}
where $e_3= \frac{1}{\sqrt{F}}(\partial_t-\partial_{r_*})$, $\underline{\chi} = -\frac{\sqrt{F}}{r}r^2g_{\mathbb{S}^2} = -\frac{\sqrt{F}}{r}\slashed{g}$ and $\psi$ is a scalar function or tensor field.

Since $e_3 = \frac{1}{\sqrt{F}}\underline{L}$, $\partial_{x^a}\left( \frac{1}{\sqrt F}\right)=0$ and $\mathrm{Tr}(\underline{\chi})=-\frac{\sqrt{F}}{r}\mathrm{Tr}\slashed{g} = -\frac{2\sqrt{F}}{r}$ with respect to local coordinates $(x^a,x^b)$, Equality \eqref{com1} leads to
\begin{equation*}
\frac{1}{\sqrt{F}}\slashed{\nabla}_{\underline L}(\slashed{\nabla}_{\partial_{x^a}}\psi) - \frac{1}{\sqrt{F}}\slashed{\nabla}_{\partial_{x^a}}(\slashed{\nabla}_{\underline L}\psi) = \frac{\sqrt{F}}{r} \slashed{\nabla}_{\partial_{x^a}}\psi.
\end{equation*}
This is equivalent to
$$\slashed{\nabla}_{\underline L}(\slashed{\nabla}_{\partial_{x^a}}\psi) - \slashed{\nabla}_{\partial_{x^a}}(\slashed{\nabla}_{\underline L}\psi) = \frac{F}{r} \slashed{\nabla}_{\partial_{x^a}}\psi.$$
Therefore, we obtain that
\begin{eqnarray}
[r\slashed\nabla_{\partial_{x^a}},\slashed{\nabla}_{\underline{L}}]\psi &=& r\slashed\nabla_{\partial_{x^a}}(\slashed{\nabla}_{\underline{L}}\psi) - \slashed{\nabla}_{\underline{L}}(r\slashed\nabla_{\partial_{x^a}}\psi)\cr
&=& r\slashed\nabla_{\partial_{x^a}}(\slashed{\nabla}_{\underline{L}}\psi) - (\slashed{\nabla}_{\underline{L}}r)(\slashed\nabla_{\partial_{x^a}}\psi) - r\slashed{\nabla}_{\underline{L}}(\slashed\nabla_{\partial_{x^a}}\psi)\cr
&=&  r\slashed\nabla_{\partial_{x^a}}(\slashed{\nabla}_{\underline{L}}\psi) - r\slashed{\nabla}_{\underline{L}}(\slashed\nabla_{\partial_{x^a}}\psi) + (\partial_{r^*}r) (\slashed\nabla_{\partial_{x^a}}\psi)\cr
&=& -r\frac{F}{r}(\slashed\nabla_{\partial_{x^a}}\psi) + F\slashed\nabla_{\partial_{x^a}}\psi \cr
&=& 0.
\end{eqnarray}
Hence, $[r\slashed\nabla_{\partial_{x^a}},\slashed{\nabla}_{\underline{L}}]=0$. Similarly, we can prove that $[r\slashed\nabla_{\partial_{x^a}},\slashed{\nabla}_{L}]=0$. As consequences of the two first commutators, we can obtain that
$$[r^2\slashed{\Delta},\slashed{\nabla}_L]=[r^2\slashed{\Delta},\slashed{\nabla}_{\underline{L}}]=0,$$
which were used in the proof of Proposition \ref{relationEq}.

To prove the third commutator, we have (see Subsection 4.3.2 in \cite{Da2019}):
\begin{equation}\label{com2}
\slashed{\nabla}_{e_3}(\slashed{\nabla}_{e_4}\psi) - \slashed{\nabla}_{e_4}(\slashed{\nabla}_{e_3}\psi) = \hat{\omega}\slashed{\nabla}_{e_3}\psi -\underline{\hat\omega}\slashed{\nabla}_{e_4}\psi,
\end{equation}
where $e_3= \frac{1}{\sqrt{F}}(\partial_t-\partial_{r_*}) = \frac{1}{\sqrt{F}}\underline{L}$, $e_4= \frac{1}{\sqrt{F}}(\partial_t+\partial_{r_*}) = \frac{1}{\sqrt{F}}L$, $\hat\omega = -\underline{\hat\omega} = \frac{M}{r^2\sqrt{F}}$ and $\psi$ is a scalar function or tensor field.

We can calculate that $\underline{L}\left(\frac{1}{\sqrt F}\right) = \frac{M}{r^2\sqrt F}$ and $L\left(\frac{1}{\sqrt F}\right) = -\frac{M}{r^2\sqrt F}$. Therefore, equality \eqref{com2} is equivalent to
\begin{equation*}
\frac{1}{\sqrt F}\slashed{\nabla}_{\underline L}(\slashed{\nabla}_{L}\psi) + \frac{M}{r^2F}(\slashed{\nabla}_{L}\psi)  - \left(\frac{1}{\sqrt F}\slashed{\nabla}_L(\slashed{\nabla}_{\underline L}\psi) - \frac{M}{r^2F}\slashed{\nabla}_{\underline L}\psi \right) = \frac{M}{r^2F}\slashed{\nabla}_{\underline L}\psi + \frac{M}{r^2F}\slashed{\nabla}_L\psi.
\end{equation*}
This leads to
$$\slashed{\nabla}_{\underline L}(\slashed{\nabla}_{L}\psi) - \slashed{\nabla}_L(\slashed{\nabla}_{\underline L}\psi) =0$$ and we obtain the third commutator.

We prove the last commutator $[r\slashed{\nabla}_{\partial_{x^a}},\slashed{\Delta}]\psi = \frac{1}{r^2}(r\slashed{\nabla}_{\partial_{x^a}})\psi$, for all scalar function $\psi$. For simplicity, we denote $\slashed{\nabla}_{\partial_{x^a}}=\slashed{\nabla}_a$. Using the Ricci identity and the symmetries of the Riemannian curvature tensor, we have 
\begin{eqnarray*}
\slashed{\Delta}\slashed{\nabla}_a\psi &=& \slashed{g}^{bc}\slashed{\nabla}_b\slashed{\nabla}_c\slashed{\nabla}_a\psi = \slashed{\nabla}^c\slashed{\nabla}_c\slashed{\nabla}_a \psi = \slashed{\nabla}^c\slashed{\nabla}_a \slashed{\nabla}_c \psi\cr
&=& - {{{\mathrm{Rm}^c}_a}{^d}}_c\slashed{\nabla}_d\psi + \slashed{\nabla}_a\slashed{\nabla}^c\slashed{\nabla}_c\psi\cr
&=& \slashed{g}^{cd}\mathrm{Ric}_{ac}\slashed{\nabla}_d\psi + \slashed{\nabla}_a\slashed{\Delta}\psi,
\end{eqnarray*}
where $\mathrm{Rm}$ and $\mathrm{Ric}$ are Riemannian and Ricci curvature tensors associated with the rough metric $\slashed{g}$. It is known that $\mathrm{Ric} =g_{\mathbb{S}^2}= \frac{1}{r^2}\slashed{g}.$
Therefore, we obtain that
$$\slashed{\Delta}\slashed{\nabla}_a\psi = \frac{1}{r^2}\slashed{\nabla}_a\psi + \slashed{\nabla}_a\slashed{\Delta}\psi.$$
This leads to $[r\slashed{\nabla}_{\partial_{x^a}},\slashed{\Delta}]\psi = \frac{1}{r^2}(r\slashed{\nabla}_{\partial_{x^a}})\psi$. The last commutator holds.

\subsection{The Goursat problem for tensorial wave equations}\label{appendix}
In this appendix we give a brief proof of Lemma \ref{partlyGoursat} that is a modification of H\"ormander's work \cite{Ho1990} (see Theorem 2 and its proof) for the wave equations on vector fields \eqref{systemwave}. For convenience to follow the proof we use the same notations in \cite{Ho1990}.
Without loss of generality, we can replace $\mathbb{S}^3$ by $X$ which is a compact Riemannian manifold without boundary of dimension $n$ equipped with metric $h(t)=\sum_{jk}h_{jk}(t,x)\d x^j\d x^k$. 
In local coordinates, we have $\d \nu = \gamma\d x$, where $\d \nu$ is a fixed smooth density on $X$.
On $(\widetilde{X},\mathfrak{g}) = (\mathbb{R}_t\times X,\mathfrak{g})$, where $\mathfrak{g}=\d t^2-h$, we consider the wave equation on vector fields \eqref{systemwave}. We consider the folowing hypersurface of initial data
\begin{equation*}
\Sigma = \left\{ (\varphi(x),x); x\in X \right\},\, \, \varphi: X\to \mathbb{R},
\end{equation*}
where $\varphi$ is a Lipschitz continuous function on $X$ and satisfies the weak spacelike condition
\begin{equation}\label{graph}
\sum_{j,k=1}^n h^{jk}(\varphi(x),x)\partial_j\varphi(x)\partial_k\varphi(x) \leq 1 \,\,\, (x\in X)
\end{equation}
almost every where on $X$. The hypersurface $X$ is spacelike if the right-hand side (RHS) of \eqref{graph} is less than $1$ for almost every where $x\in X$ and it is characteristic (or null) if the RHS of \eqref{graph} is equal to $1$ for almost all $x\in X$.

The main difference between Equation \eqref{systemwave} and the wave equations on scalar functions in \cite{Ho1990} is the term $L_1(\Phi)$ which consists of the first order differential operators on vector fields. In particular, Equation \eqref{systemwave} is equivalent to a symmetrical hyperbolic system which consists two wave equations on scalar functions with the first order differential operators on both $\phi_1$ and $\phi_2$ in each equation
\begin{equation}\label{EQ1}
(\partial_t^2\phi_1 - \Delta_h\phi_1) + (L_1^{11}\phi_1+ L_1^{12}\phi_2) = 0
\end{equation}
and
\begin{equation}\label{EQ2}
(\partial_t^2\phi_2 - \Delta_h\phi_2) + (L_1^{21}\phi_1+ L_1^{22}\phi_2) =0,
\end{equation}
where $L_1^{11}=L_1^{22},\, L_1^{12}=L_1^{21}$ with $L_1^{ij} = b_0^{ij}\partial_t + b_1^{ij}\partial_x + c^{ij}$, and $\Delta_h$ is the Laplace-Beltrami operator associated with metric $h$ on $X$:
$$\Delta_h = \sum_{j,k} \gamma^{-1}\partial_j (\gamma h^{jk}(t,x)\partial_k),\, \partial_k=\frac{\partial}{\partial x_k},$$ 
with  $(h^{jk})=(h_{jk})^{-1}$.

Assume that $\Phi=(\phi_1,\phi_2)$ is a smooth solution of the couple equations \eqref{EQ1} and \eqref{EQ2}. Similar to \cite{Ho1990} (see equation 4, page 272), Equation \eqref{EQ1} leads to
\begin{eqnarray}\label{EQ3}
0&=&2\partial_t\bar{\phi}_1(\partial_t^2\phi_1 - \Delta_h\phi_1) + (\partial_t\bar{\phi}_1)(L_1^{11}\phi_1 ) + (\partial_t\bar{\phi}_2)(L_1^{12}\phi_2)\cr
&=& \partial_t \left( (\partial_t\phi_1)^2 + \sum_{j,k} h^{jk}\partial_j\phi_1\partial_k\phi_1 + (\phi_1)^2\right)\cr
&& - 2\gamma^{-1}\sum_{j,k}\partial_j\left( \gamma h^{jk}\partial_t\phi_1\partial_k\phi_1\right) -\sum_{j,k}(\partial_t h^{jk})\partial_j\phi_1\partial_k\phi_1\cr
&& + \underbrace{2(\partial_t\bar{\phi}_1)((L_1^{11}-1)\phi_1) + 2(\partial_t\bar{\phi}_1)(L_1^{12}\phi_2)}_{I_1}.
\end{eqnarray}
Similarly, Equation \eqref{EQ2} leads to
\begin{eqnarray}\label{EQ4}
0&=&2\partial_t\bar{\phi}_2(\partial_t^2\phi_2 - \Delta_h\phi_2) + (\partial_t\bar{\phi}_1)(L_2^{21}\phi_2) + (\partial_t\bar{\phi}_2)(L_2^{22}\phi_2)\cr
&=& \partial_t \left( (\partial_t\phi_2)^2 + \sum_{j,k} h^{jk}\partial_j\phi_2\partial_k\phi_2 + (\phi_2)^2\right)\cr
&& - 2\gamma^{-1}\sum_{j,k} \partial_j\left( \gamma h^{jk}\partial_t\phi_2\partial_k\phi_2\right) -\sum_{j,k}(\partial_t h^{jk})\partial_j\phi_2\partial_k\phi_2\cr
&& + \underbrace{2(\partial_t\bar{\phi}_2)(L_1^{21}\phi_1) + 2(\partial_t\bar{\phi}_2)((L_1^{22}-1)\phi_2)}_{I_2}.
\end{eqnarray}
The equations \eqref{EQ3} and \eqref{EQ4} have the mixed terms $I_1$ and $I_2$ of scalar functions $\phi_1$ and $\phi_2$. In order to control these terms, we define the pointwise norm of the vector $\Phi=(\phi_1,\,\phi_2)$ by
$$\left\|\Phi\right\| =  |\phi_1| + |\phi_2|$$
and we introduce the energy on the vector field $\Phi=(\phi_1,\,\phi_2)$ by
\begin{equation}\label{EN}
\mathcal{E}(\Phi)(t)=\int_X\left( \left\|\partial_t\Phi\right\|^2 + \sum_{j,k} h^{jk}\partial_j\Phi\partial_k\Phi + \left\|\Phi\right\|^2 \right)\d \nu(x),
\end{equation}
where $\partial_t\Phi = \left(\begin{matrix}
\partial_t\phi_1 \\
\partial_t\phi_2
\end{matrix} \right)$ and $\partial_j\Phi\partial_k\Phi = \partial_j\phi_1\partial_k\phi_1+\partial_j\phi_2\partial_k\phi_2$. We can see that $\mathcal{E}(\Phi)(t) = \mathcal{E}(\phi_1)(t) + \mathcal{E}(\phi_2)(t)$, where
\begin{equation}
\mathcal{E}(\phi_i)(t)=\int_X \left( \left\|\partial_t\phi_i\right\|^2 + \sum_{j,k} h^{jk}\partial_j\phi_i\partial_k\phi_i + \left\|\phi_i\right\|^2 \right)\d \nu(x) \,\, (i=1,2),
\end{equation}
which is energy of scalar fields defined in \cite{Ho1990} (and also \cite{Ni2006}). Moreover, we have
\begin{equation}
\mathcal{E}(\Phi)(t) = \|\Phi(t)\|^2_{H^1(X)} + \|\partial_t\Phi(t)\|^2_{L^2(X)},
\end{equation}
where
\begin{eqnarray*}
&&\|\Phi(t)\|_{H^1(X)}^2 = \int_X \left( \sum_{j,k} h^{jk}\partial_j\phi_i\partial_k\phi_i + \|\phi_i\|^2 \right)\d \nu(x),\cr
&&\|\partial_t\Phi(t)\|_{L^2(X)}^2 = \int_X \left\|\partial_t\phi_i\right\|^2 \d \nu(x).
\end{eqnarray*}
\begin{remark}
If we use the energy momentum tensor for Equation \eqref{systemwave} and the Killing vector field $T=\partial_t$, we can also obtain the energy $\mathcal{E}^T(\Phi)(X)$ on spacelike hypersurface $X$ which is equivalent to $\mathcal{E}(\Phi)(t)$ (see Subsection \ref{ConserFac} or more details in \cite[page 184]{Saka1985}). Moreover, the restriction of energy norm $\sqrt{\mathcal{E}(\Phi)(t)}$ of $\Phi$ on $\Sigma_t$ is also equivalent to energy of tensorial field $\phi_a$ on $\Sigma_t$ given by \eqref{ENERGY} in Subsection \ref{FacSpace}.
\end{remark}
Observe that, the term $I_1$ in \eqref{EQ3} can be controlled as 
\begin{eqnarray}\label{I1}
I_1&=&2(\partial_t\bar{\phi}_1)((L_1^{11}-1)\phi_1) + 2(\partial_t\bar{\phi}_1)(L_1^{12}\phi_2) \cr
&\leq& 2|\partial_t\phi_1||(L_1^{11}-1)\phi_1| + 2|\partial_t\phi_1||L_1^{12}\phi_2| \cr
&\leq& 2\left\|\partial_t\Phi \right\| \left(  \left\|\partial_t\Phi\right\| + \left\|\partial_x\Phi\right\| + \left\|\Phi\right\| \right)\cr
&\leq& C \left( \left\|\partial_t\Phi \right\|^2 + \left\|\partial_x\Phi\right\|^2 + \left\|\Phi\right\|^2   \right)\cr
&=& C\left( \left\|\partial_t\Phi\right\|^2 + \sum_{j,k} h^{jk}\partial_j\Phi\partial_k\Phi + \left\|\Phi\right\|^2 \right).
\end{eqnarray}
By the same way, we have
\begin{equation}\label{I2}
I_2 \leq C \left( \left\|\partial_t\Phi\right\|^2 + \sum_{j,k} h^{jk}\partial_j\Phi\partial_k\Phi + \|\Phi\|^2 \right).
\end{equation}
Therefore, integrating \eqref{EQ3} and \eqref{EQ4} on $\left\{t\leq \tau\leq s\right\}\times X$, we obtain the energy estimate
\begin{equation}\label{ENG}
\mathcal{E}(\Phi)(t)\leq \mathcal{E}(\Phi)(s) e^{D|t-s|}.
\end{equation}
For any foliation $\left\{X_\tau\right\}_{\tau\in \mathbb{R}}$, where $X_\tau$ is Cauchy hypersurface and $X_0=X$. We can see that $X_\tau$ is topological $X$ endowed with the Riemannian metric $-\mathfrak{g}|_{\widetilde{\mathcal{S}}_\tau}$. By using Leray's theorem for symmetrical hyperbolic systems on smooth globally hyperbolic spacetimes (see \cite{Le1952}), we obtain the existence of smooth solution $\Phi$ of Cauchy problem for Equation \eqref{systemwave} with the smooth initial data $(\Phi|_X,\partial_t\Phi|_X) \in C^\infty(X)\times C^\infty(X)$.
Using energy norm $\sqrt{\mathcal{E}(\Phi)(\tau)}$ (where $\mathcal{E}(\Phi)(\tau)$ given by \eqref{EN}) and energy estimate \eqref{ENG}, we can prove the local well-posedness in ${C}^0([0,T];\, \cup_{0\leq \tau \leq T} H^1(X_\tau)) \cap {C}^1([0,T];\, \cup_{0 \leq \tau\leq T} L^2(X_\tau))$ of Cauchy problem of Equation \eqref{systemwave} with the initial data satisfying $(\Phi|_X,\partial_\tau\Phi|_X) \in H^1(X)\times L^2(X)$ by the same method in \cite[Theorem 2]{Saka1985} (see also \cite{Bruhat1979}).
Using the local well-posedness result and energy estimate \eqref{ENG}, we can establish the global well-posedness of Cauchy problem of Equation \eqref{systemwave} in ${C}^0(\mathbb{R}_\tau;\, \cup_{\tau \geq 0} H^1(X_\tau)) \cap {C}^1(\mathbb{R}_{\tau};\, \cup_{\tau\geq 0} L^2(X_\tau))$ by the same methods in \cite[Theorem 2]{Cagnac} and \cite[Theorem 1]{Dossa}. The same process is used to prove Theorem \ref{Cauchyproblem} and Theorem \ref{CauchyTeu}.

We denote by $\mathscr{E}$ the closure of space of all smooth solutions of Cauchy problem for \eqref{systemwave} in ${C}^0(\mathbb{R}_\tau;\, \cup_{\tau\geq 0} H^1(X_\tau)) \cap {C}^1(\mathbb{R}_\tau;\, \cup_{\tau\geq 0} L^2(X_\tau))$ under the energy norm $\sup\limits_{\tau\in \mathbb{R}}\sqrt{\mathcal{E}(\Phi)(\tau)}$. The space $\mathcal{E}$ is called finite energy space.

Let $\Sigma= \left\{ (\varphi(x),x),\, x\in X\right\}$ be a weakly spacelike hypersurface. We define on hypersurface $\Sigma$ the density measure
$$\d \nu^0_\Sigma = \left(1-\sum_{j,k} h^{jk}\partial_j\varphi\partial_k\varphi\right)\d \nu_\Sigma.$$
This density measure is positive if $\Sigma$ is spacelike and vanishes if $\Sigma$ is null. Using this we define the norm of $\partial_t\Phi$ in $L^2(\Sigma;\, \d\nu^0|_\Sigma)$ by
\begin{equation}\label{newnorm}
\left\| \partial_t\Phi\right\|_{L^2(\Sigma;\, \d\nu^0_\Sigma)} = \int_\Sigma \left\| \partial_t\Phi\right\|^2 \d\nu^0_\Sigma = \int_\Sigma \left\| \partial_t\Phi\right\|^2 \left(1-\sum_{j,k} h^{jk}\partial_j\varphi\partial_k\varphi\right)\d \nu_\Sigma.
\end{equation}
Integrating \eqref{EQ3} and \eqref{EQ4} over $\left\{ (t,x): T\leq t \leq \varphi\right\}$ (where $T$ is chosen such that $T\leq \min\varphi$) and using \eqref{newnorm} and estimates \eqref{I1}, \eqref{I2} for $2\partial_t\Phi L_1(\Phi)$, we can establish the same estimate as (7) in \cite{Ho1990}: 
\begin{eqnarray}\label{e1}
&&\int_\Sigma \left\| \partial_t\Phi\right\|^2 \d\nu^0_\Sigma + \int_\Sigma \left( \sum_{j,k} h^{jk}(\partial_j\Phi + \partial_j\varphi\partial_t \Phi)(\partial_k\Phi + \partial_k\varphi\partial_t \Phi) + \left\| \Phi \right\|^2 \right)\d\nu_\Sigma\cr
&=& \int_\Sigma \left\| \partial_t\Phi\right\|^2 \d\nu^0_\Sigma + \int_\Sigma \left( \sum_{j,k} h^{jk}(\partial_j\Phi|_\Sigma) (\partial_k\Phi|_\Sigma) + \left\| \Phi\right\|^2\right)\d\nu_\Sigma\cr
&=&\left\| \partial_t\Phi \right\|_{L^2(\Sigma;\d\nu^0_\Sigma)} + \left\| \Phi \right\|_{H^1(\Sigma)}\cr
&\leq& C\left\| \Phi\right\|_{\mathscr{E}},
\end{eqnarray}
where we used $\partial_j(\Phi(\varphi(x),x)) = (\partial_j\varphi\partial_t\Phi + \partial_j\Phi)|_\Sigma$. 

Similar to the proof of Theorem 2 in \cite{Ho1990} (see page 274), we can establish an opposite estimate of \eqref{e1}. Indeed, we define
$$\mathcal{E}_\varphi(\Phi)(t) = \int_{\varphi(x)\leq t} \left( \left\|\partial_t\Phi\right\|^2 + \sum_{j,k}\partial_j\Phi\partial_k\Phi + \left\|\Phi\right\|^2  \right) \d \nu(x).$$
Integrating \eqref{EQ3} and \eqref{EQ4} over $\left\{ (s,x): \varphi(x)\leq s \leq t \right\}$ and using estimates \eqref{I1}, \eqref{I2}, we get
\begin{eqnarray}\label{re1}
\mathcal{E}_\varphi(\Phi)(t) &\leq& C\int_T^t \mathcal{E}_\varphi(\Phi)(s) + \int_\Sigma \left\| \partial_t\Phi\right\|^2 \d\nu^0_\Sigma + \int_\Sigma \left( \sum_{j,k} h^{jk}(\partial_j\Phi|_\Sigma) (\partial_k\Phi|_\Sigma) + \left\| \Phi\right\|^2\right)\d\nu_\Sigma\cr
&\leq& C\int_T^t \mathcal{E}_\varphi(\Phi)(s) + (\left\| \partial_t\Phi \right\|_{L^2(\Sigma;\d\nu^0_\Sigma)} + \left\| \Phi \right\|_{H^1(\Sigma)}),
\end{eqnarray}
where $T\leq \min\varphi$, so that $\mathcal{E}_\varphi(T)=0$. By using Gronwall’s lemma for \eqref{re1}, we get
\begin{equation*}
\mathcal{E}_\varphi(\Phi)(t)\leq e^{C(t-T)}(\left\| \partial_t\Phi \right\|_{L^2(\Sigma;\d\nu^0_\Sigma)} + \left\| \Phi \right\|_{H^1(\Sigma)}).
\end{equation*}
Hence, for $t\geq \max\varphi$, we obtain the opposite estimates of \eqref{e1} as
\begin{equation}\label{Ree1}
\left\| \Phi\right\|_{\mathscr{E}}\leq \widehat{C}(\left\| \partial_t\Phi \right\|_{L^2(\Sigma;\d\nu^0_\Sigma)} + \left\| \Phi \right\|_{H^1(\Sigma)}),
\end{equation}
where $\widehat{C}$ is a positive constant.

By using the energy on vector fields \eqref{EN} and energy estimates \eqref{ENG}, \eqref{e1}, \eqref{Ree1}, we can continue the process to get the proof of Theorem 2 in \cite{Ho1990} for wave equation on vector fields \eqref{systemwave} and get the same results that: for every weak spacelike hypersurface $\Sigma$, the map
\begin{align}
\Gamma:\hbox{      } {C}^\infty(\mathbb{R}_{\tau}, \cup_{\tau\in \mathbb{R}} H^1(X_\tau))\cap {C}^\infty(\mathbb{R}_\tau, \cup_{\tau\in \mathbb{R}} L^2(X_\tau)) &\to H^1(\Sigma) \oplus L^2(\Sigma;\, \d\nu^0_\Sigma)\cr
\Phi &\mapsto \left(\Phi|_\Sigma,\, \partial_t\Phi|_\Sigma \right)
\end{align}
is well defined for smooth solutions. Moreover, we can extend this map as an isometry on the finite energy space $\mathscr{E}$ as
\begin{align}
\Gamma:\hbox{      } \mathscr{E} &\to H^1(\Sigma) \oplus L^2(\Sigma;\, \d\nu^0_\Sigma)\cr
\Phi &\mapsto \left(\Phi|_\Sigma,\, \partial_t\Phi|_\Sigma \right).
\end{align}
Since $\Gamma$ is surjective, the Goursat problem of \eqref{systemwave} is well-posedness in $(\widetilde{X},\mathfrak{g})$ for the smooth initial data on null hypersurface $\Sigma$.

Therefore, we obtain the well-posedness of the Goursat problem of Equation \eqref{systemwave} in $(\mathbb{R}_t\times \mathbb{S}^3,\mathfrak{g})$ in Lemma \ref{partlyGoursat} as an applications of the above result with $X=\left\{ 0 \right\}\times \mathbb{S}^3$, $\Sigma = \mathcal{C}$ is a null hypersurface (which is extension of $\mathfrak{H}^+\cup \scri^+$ in $(\mathbb{R}_t\times \mathbb{S}^3,\mathfrak{g})$) and the foliation $\left\{ X_\tau\right\}_{\tau\geq 0}$, where $X_\tau = \widetilde{\mathcal{S}}_\tau$ is a spacelike hypersurface (which is an extension of $\mathcal{S}_\tau$ in $(\mathbb{R}_t\times \mathbb{S}^3,\mathfrak{g})$)).  In particular, for the initial data $ (\widetilde{\xi}_i,\widetilde{\zeta}_i)\in {C}_0^\infty(\mathcal{C})\times {C}_0^\infty(\mathcal{C})\, (i=1,2)$, Equation \eqref{systemwave} has a unique  smooth solution $\widetilde{\Phi}=(\widetilde{\phi}_1,\widetilde{\phi}_2)$ satisfying
\begin{equation*}
\widetilde{\Phi} \in {\mathcal {C}}^\infty(\mathbb{R}_\tau;\, \cup_{\tau\geq 0}H^1(\widetilde{\mathcal{S}}_\tau)), \, \partial_\tau \widetilde{\Phi} \in {C}^\infty(\mathbb{R}_\tau; \, \cup_{\tau\geq 0}L^2(\widetilde{\mathcal{S}}_\tau)).
\end{equation*} 
This means that
\begin{equation*}
\widetilde{\phi}_i \in {\mathcal {C}}^\infty(\mathbb{R}_\tau;\, \cup_{\tau\geq 0}H^1(\widetilde{\mathcal{S}}_\tau)), \, \partial_\tau \widetilde{\phi}_i \in {C}^\infty(\mathbb{R}_\tau; \, \cup_{\tau\geq 0}L^2(\widetilde{\mathcal{S}}_\tau))
\end{equation*} 
for $i=1,2$.
\begin{remark}
The Goursat problems for wave equations on spinor fields were also established in some other works \cite{Mo2019,Mo2022,Pha2022}.  
\end{remark}

\vspace{0.3cm}


\end{document}